\RequirePackage{fix-cm}
\documentclass[smallextended]{svjour3}       % onecolumn (second format)
\smartqed  % flush right qed marks, e.g. at end of proof
\usepackage[linesnumbered,ruled,vlined]{algorithm2e}
\usepackage{graphicx}
\usepackage{multirow}
\usepackage{float}
\newtheorem{mylem}{Lemma}
\newtheorem{myclaim}{Claim}
\newtheorem{mythem}{Theorem}

\newtheorem{mydef}{Definition}
\usepackage[framemethod=TikZ]{mdframed}
\mdfdefinestyle{MyFrame}{%
	linecolor=blue,
	outerlinewidth=1.0pt,
	roundcorner=20pt,
	innertopmargin=\baselineskip,
	innerbottommargin=\baselineskip,
	innerrightmargin=20pt,
	innerleftmargin=20pt,
	backgroundcolor=gray!50!white}
\usepackage[pagebackref]{hyperref}
\hypersetup{
	colorlinks=true,
	linkcolor=red,
	filecolor=blue, 
	citecolor=blue,     
	urlcolor=blue,
}
\usepackage{amssymb}
\usepackage{amsmath}
\usepackage[sort&compress,numbers]{natbib}
\usepackage{enumerate}
\begin{document}

\title{First Stretch then Shrink and Bulk: A Two Phase Approach for Enumeration  of Maximal $(\Delta, \gamma)$\mbox{-}Cliques of a Temporal Network
\thanks{Suman Banerjee is supported by the Post Doctoral Fellowship Grant sponsored by the Indian Institute of Technology Gandhinagar (Project No. MIS/IITGN/PD-SCH/201415/006). \\
Both the authors have contributed equally in this work and they are joint first authors.}
}


\author{Suman Banerjee         \and
        Bithika Pal %etc.
}


\institute{Suman Banerjee \at
              Department of Computer Science and Engineering \\
              Indian Institute of Technology Gandhinagar, India. \\
%              Fax: +123-45-678910\\
              \email{suman.b@iitgn.ac.in}           %  \\
           \and
           Bithika Pal \at
              Department of Computer Science and Engineering, \\
              Indian Institute of Technology Kharagpur, India. \\
              \email{bithikapal@iitkgp.ac.in}
}

\date{Received: date / Accepted: date}

\maketitle

\begin{abstract}
A \emph{Temporal Network} (also known as \emph{Link Stream} or \emph{Time-Varying Graph}) is often used to model a time-varying relationship among a group of agents. It is typically represented as a collection of triplets of the form $(u,v,t)$ that denotes the interaction between the agents $u$ and $v$ at time $t$. For analyzing the contact patterns of the agents forming a temporal network, recently the notion of classical \textit{clique} of a \textit{static graph} has been generalized as \textit{$\Delta$\mbox{-}Clique} of a Temporal Network. In the same direction, one of our previous studies introduces the notion of \textit{$(\Delta, \gamma)$\mbox{-}Clique}, which is basically a \textit{vertex set}, \textit{time interval} pair, in which every pair of the clique vertices are linked atleast $\gamma$ times in every $\Delta$ duration of the time interval. In this paper, we propose a different methodology for enumerating all the maximal $(\Delta, \gamma)$\mbox{-}Cliques of a given temporal network. The proposed  methodology is broadly divided into two phases. In the first phase, each temporal link is processed for constructing $(\Delta, \gamma)$\mbox{-}Clique(s) with maximum duration. In the second phase, these initial cliques are expanded by vertex addition to form the maximal cliques. By sequential arguments, we show that the proposed methodology correctly enumerates all the maximal $(\Delta, \gamma)$\mbox{-}Cliques. Comprehensive analysis for running time and space requirement of the proposed methodology has also been done. From the experimentation carried out on $5$ real\mbox{-}world temporal network datasets, we observe that the proposed methodology enumerates all the maximal $(\Delta,\gamma)$\mbox{-}Cliques efficiently, particularly when the dataset is sparse. As a special case ($\gamma=1$), the proposed methodology is also able to enumerate $(\Delta,1) \equiv \Delta$\mbox{-}cliques with much less time compared to the existing methods.  
\end{abstract}

\section{Introduction} \label{intro}
Network (also called graph) is a mathematical object which is used extensively to represent a \textit{binary relation} among a group of agents. Analyzing such networks for different structural patterns remains an active area of study in different domains including \emph{Computational Biology} \citep{hulovatyy2015exploring}, \emph{Social Network Analysis},  \emph{Computational Epidemiology} \citep{masuda2017temporal} and many more. Among many one such structural pattern is the maximally connected subgraphs, which is popularly called as \emph{cliques}. Finding the maximum cardinality clique in a given network is a well known \textit{NP\mbox{-}Complete} Problem \citep{garey2002computers}. However, in network analysis perspective more general problem is not only just finding the maximum size clique, but also to enumerate all the maximal cliques present in the network. Bron and Kerbosch \citep{bron1973algorithm} first proposed an enumeration algorithm for maximal cliques in the network which forms the foundation of study on this problem. Later, there were advancements for this problem for different types of networks \citep{cheng2012fast, eppstein2011listing, eppstein2013listing1} etc.
\par Real-world networks from biological to social are \emph{time varying}, which means that the existence of an edge between any two agents changes with time. Temporal networks \citep{holme2012temporal} (also known as \emph{link streams} or \emph{time varying networks}) are the mathematical objects used to formally represent the time varying relationships. For these type of networks, a natural supplement of clique is the \textit{temporal clique} which consists of two things: a subset of the vertices, and a time interval. In this direction, recently, Virad et al. \citep{viard2015revealing, viard2016computing} put forward the notion of $\Delta$\mbox{-}Clique, where a vertex subset along with a time interval is said to be a $\Delta$\mbox{-}Clique if every vertex pair from that set have at least a single edge in every $\Delta$ duration within the time interval. Next, we report the existing studies on clique enumeration on networks.
\subsection{Relevant Studies}
The problem of maximal clique enumeration is a classic computational problem on network algorithms and has been extensively studied on static networks. \cite{akkoyunlu1973enumeration} was the first to propose an algorithm for this problem. Later, \cite{DBLP:journals/cacm/BronK73} introduced a recursive approach for the maximal clique enumeration problem. These two studies are the foundations on maximal clique enumeration and trigger a huge amount of research due to many practical applications from computational biology to spatial data analytic \cite{al2007enumeration, bhowmick2015clustering}. Since past two decades several methodologies have been developed for enumerating maximal cliques in different computational paradigms, and different kinds of networks, such as in sparse graphs \cite{eppstein2011listing, eppstein2013listing}, in large networks \cite{cheng2010finding, cheng2011finding, rossi2014fast}, in map reduce framework \cite{hou2016efficient, xiang2013scalable}, in uncertain graphs \cite{mukherjee2015mining, mukherjee2016enumeration, zou2010finding}, in parallel computing framework \cite{chen2016parallelizing, du2006parallel, rossi2015parallel, schmidt2009scalable} and many more.
\par Though there are many existing studies on maximal clique enumeration on static networks, however, the literature on temporal graphs is limited. Viard et al. \citep{viard2015revealing} proposed an enumeration  algorithm for maximal $\Delta$\mbox{-}Clique of a temporal network. They did a detailed analysisof contact relationship among a group of students, based on their introduced methodology. Thry were able to show that their analysis draws deeper insights of their communication pattern \citep{viard2015revealing}. Later, Himmel et al. \cite{himmel2016enumerating} proposed a different approach for maximal $\Delta$\mbox{-}Clique enumeration problem. Their methodology is based on the \textit{Bron\mbox{-}Kerbosch Algorithm} for maximal clique enumeration in static graphs. Their methodology is better in both of the following aspects: theoretically (measured in terms of worst case computational complexity analysis), as well as practically (measured in terms of computational time when the algorithm is implemented with real\mbox{-}world datasets). 
%Reported results in \citep{himmel2016enumerating} show that, their proposed algorithm performs much better than that of in \citep{viard2016computing} in terms of both worst case computational time analysis as well as in experimentaion with real\mbox{-}life data sets. 
Recently, Molter et al. \cite{molter2019enumerating} introduced the notion of \emph{isolation} in clique enumeration of a time varying graph. They developed fixed parameter enumeration algorithm based on different notion of isolation employing the parameter ``degree of isolation". Recently, Banerjee and Pal \cite{DBLP:conf/comad/BanerjeeP19} proposed an enumeration algorithm for maximal $(\Delta, \gamma)$\mbox{-}\textit{Cliques} present in a time varying graph. As far as we know, other than the last one there is no other work available which studies $(\Delta, \gamma)$\mbox{-}\textit{Cliques}.
\subsection{Contribution of the Paper}
As mentioned previously, a temporal network consists of a set of agents and a time varying relationship. Now, the following questions are essential to understand the contact pattern among them: which subset of agents comes in contact very frequently among each other?  Given a time duration how many times they contact with each other? etc. The frequency of communication also adds another dimension of information to their relationship strength. Motivated by such questions, recently the notion of $\Delta$\mbox{-}Clique has been extended to $(\Delta, \gamma)$\mbox{-}\textit{Cliques}, which is basically a vertex subset and time interval pair in which each pair of communicating vertices of the subset has minimum $\gamma$ interactions in every $\Delta$ duration within the time interval. In this paper, we give a different approach for listing out all the maximal $(\Delta, \gamma)$\mbox{-}Cliques that are there in a temporal network. The main  contributions of this paper are as follows:
\begin{itemize}
	\item In this paper, we propose a different approach, namely, first stretch and then shrink and bulk, for listing out maximal $(\Delta, \gamma)$\mbox{-}Cliques that are there in a temporal network.
	\item By drawing sequential arguments, we prove the correctness of the proposed methodology.
	\item A detailed analysis of the proposed methodology has been done to understand its computational time and space requirement. 
	%\item Proposed methodology has been analyzed in detail and elaborated with an example. 
	\item  The proposed methodology has been implemented with five publicly available temporal network datasets to bring out nontrivial insights about contact patterns and compare the efficiency of the proposed methodology with the existing one.
	\item Also, a set of experiments have been conducted to show that the proposed methodology of maximal $(\Delta, \gamma)$\mbox{-}Clique enumeration can also be efficiently used for enumerating maximal $\Delta$\mbox{-}Clique as well (By putting $\gamma=1$).
\end{itemize}

\subsection{Structure of this Article} 
Remaining portion of this article is arranged in the following way: Section \ref{Sec:Preli} discusses some preliminary concepts regarding temporal network and formally defines the maximal $(\Delta, \gamma)$\mbox{-}Clique enumeration problem formally. Section \ref{Sec:PET} contains the proposed enumeration technique with its detailed analysis, proof of correctness and an illustrative example. Section \ref{Sec:Experiment} describes experimental evaluation of the proposed methodology in details. Finally, Section \ref{Sec:CFD} concludes study and gives future directions.
\section{Background and Problem Definition} \label{Sec:Preli}
In this section we present some preliminary concepts to understand the problem, that we work out in this paper, and the proposed solution methodology. In a \textit{temporal network}, its edges are marked with the corresponding occurrence timestamp(s). Formally, it is stated in Definition \ref{Def:1}.
\begin{mydef}[Temporal Network] \cite{holme2013temporal} \label{Def:1}
A temporal network is defined as $\mathcal{G}(V, E, \mathcal{T})$, where $V(\mathcal{G})$ is the set of vertices of the network and $E(\mathcal{G})$ is the set of edges among them. $\mathcal{T}$ is the mapping that maps each edge of the graph to its occurrence time stamp(s).
\end{mydef}
\begin{figure}
	\centering
	\resizebox{6.5 cm}{1.8 cm}{\includegraphics{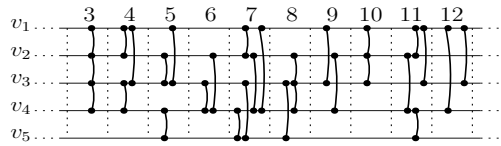}}
	\caption{Links of a Temporal Network}
	\label{Fig:TG}
\end{figure}
Figure \ref{Fig:TG} shows a temporal graph with $5$ vertices and $29$ edges, where edges are shown in the time horizon. In temporal network analysis, it is assumed that the network changes its topology in discrete time steps. So, starting at time $t$, if the network is observed in every $dt$ time difference till $t^{'}$, the time instances are $\mathbb{T}=\{t, t+dt, t+2dt, \dots, t^{'}\}$. In rest of our study we assume, $t,t^{'} \in \mathbb{Z}^{+}$ and $dt=1$. The difference between the beginning and ending time stamp, i.e., $t^{'}-t$ is called as the \textit{Life Time of the Network}. In the temporal network $\mathcal{G}$, if there is an edge between two vertices $v_i$ and $v_j$ at time $t$, then it is symbolized as $(v_i,v_j,t)$, signifying that there is a contact between $v_i$ and $v_j$ at time $t$. For some $t \in \mathbb{T}$ if $(v_i,v_j,t) \in E(\mathcal{G})$, then we say, that there exists a static edge between $v_i$ and $v_j$. The \textit{frequency} of an edge is defined as the number of $t \in \mathbb{T}$ such that $(v_i,v_j,t) \in E(\mathcal{G})$ and denoted as $f_{(v_iv_j)}$, i.e., $f_{(v_iv_j)}=|\{t \in \mathbb{T}: (v_i,v_j,t) \in E(\mathcal{G}) \}|$. If $(v_i,v_j) \notin E(\mathcal{G})$, then we say that $f_{(v_iv_j)}=0$. In rest of our study, we work with undirected temporal network, i.e., there is no difference between $(v_i,v_j,t)$ and $(v_j,v_i,t)$.
\par In a static network, a subset of vertices, where every pair is adjacent is known as a \textit{clique}. The size of the clique is defined as the number of vertices it contains. A clique is said to be maximal if it is not part of another clique of larger size. The general notion of clique is  extended for temporal graphs as $\Delta$\mbox{-}clique, which is vertex subset  and time sub\mbox{-}interval pair, such that, in each $\Delta$ duration of the sub\mbox{-}interval there exist at least one link between every pair of vertices in the vertex subset. Formally it is stated in Definition \ref{Def:2}.
\begin{mydef}[$\Delta$\mbox{-}Clique] \cite{viard2016computing} \label{Def:2}
	Given a temporal network $\mathcal{G}(V, E, \mathcal{T})$ and time duration $\Delta$, a $\Delta$\mbox{-}Clique of $\mathcal{G}$ is a vertex set, time interval pair, i.e., $(\mathcal{X}, [t_a,t_b])$ with $\mathcal{X} \subseteq V(\mathcal{G})$, $\vert \mathcal{X} \vert \geq 2$ and $[t_a,t_b] \subset \mathbb{T}$, such that $\forall v_i,v_j \in \mathcal{X}$ and $\tau \in [t_a, max(t_b - \Delta, t_a)]$ there is an edge $(v_i, v_j, t_{ij}) \in E(\mathcal{G})$ with $t_{ij} \in [\tau, min (\tau + \Delta, t_b)]$.
\end{mydef}
In one of our recent study, we introduced the notion of $(\Delta, \gamma)$\mbox{-}clique by extending the concept of $\Delta$\mbox{-}Clique and incorporating an additional parameter $\gamma$ as a frequency threshold. This is stated in Definition \ref{Def:DG}.
\begin{mydef}[$(\Delta, \gamma)$\mbox{-}Clique] \cite{DBLP:conf/comad/BanerjeeP19} \label{Def:DG}
	Given a temporal network $\mathcal{G}(V, E, \mathcal{T})$, time duration $\Delta$, and a frequency threshold $\gamma \in \mathbb{Z}^{+}$, a $(\Delta, \gamma)$\mbox{-}Clique of $\mathcal{G}$ is a tuple consisting of vertex subset, and time interval, i.e., $(\mathcal{X}, [t_a,t_b])$ where $\mathcal{X} \subseteq V(\mathcal{G})$, $\vert \mathcal{X} \vert \geq 2$, and  $[t_a,t_b] \subseteq \mathbb{T}$. Here $\forall v_i,v_j \in \mathcal{X}$ and $\tau \in [t_a, max(t_b - \Delta, t_a)]$, there must exist at least $\gamma$ number of edges, i.e., $(v_i, v_j, t_{ij}) \in E(\mathcal{G})$ and $f_{(v_iv_j)} \geq \gamma$ with $t_{ij} \in [\tau, min (\tau + \Delta, t_b)]$. Here, $f_{(v_iv_j)}$ denotes the frequency of the static edge $(v_i, v_j)$.
\end{mydef}

In a static graph $G(V, E)$, a maximal clique is formed as $\mathcal{S} \subset V(G)$, if for each $ v \in V(G) \setminus \mathcal{S}$, $\mathcal{S} \cup \{v\}$ is not a clique. Now, as the $(\Delta, \gamma)$\mbox{-}Clique is defined in the setting of temporal networks, so its maximality depends on two parameters: one is the \textit{cardinality} and the other one is the \textit{time interval}. We introduce the maximality conditions for an arbitrary $(\Delta, \gamma)$\mbox{-}Clique in Definition \ref{Def:MDG} considering both the factors.
\begin{mydef}[Maximal $(\Delta, \gamma)$\mbox{-}Clique] \label{Def:MDG} \label{Def:maximal}
	Given a temporal network $\mathcal{G}(V, E, \mathcal{T})$ and a $(\Delta, \gamma)$\mbox{-}Clique $(\mathcal{X}, [t_a,t_b])$ of $\mathcal{G}$, $(\mathcal{X}, [t_a,t_b])$ will be maximal if none of the following is true.
	\begin{itemize}
		\item $\exists v \in V(\mathcal{G}) \setminus \mathcal{X}$ such that $(\mathcal{X} \cup \{v\}, [t_a,t_b])$ is a $(\Delta, \gamma)$\mbox{-}Clique.
		\item $(\mathcal{X}, [t_a - 1,t_b])$ is a $(\Delta, \gamma)$\mbox{-}Clique. This applies only if $t_a - 1 \geq t$.
		\item $(\mathcal{X}, [t_a,t_b + 1])$ is a $(\Delta, \gamma)$\mbox{-}Clique. This applies only if $t_b + 1 \leq t^{'}$.
	\end{itemize}
\end{mydef}
From Definition \ref{Def:maximal}, it is clear that the first condition addresses the cardinality, whereas the next two are due to time duration. In a static graph, among all of its maximal cliques, one with the highest cardinality is called the maximum clique or largest size clique. However, in case of $(\Delta, \gamma)$\mbox{-}Clique, maximum can be both in terms of  cardinality or duration. Hence, maximum $(\Delta, \gamma)$\mbox{-}Clique of a temporal network can be defined as follows. 
\begin{mydef}[Maximum $(\Delta, \gamma)$\mbox{-}Clique]
	Given a temporal network $\mathcal{G}(V, E, \mathcal{T})$, let $\mathcal{C}$ be the set of all maximal $(\Delta, \gamma)$\mbox{-}Cliques of $\mathcal{G}$. Now, $(\mathcal{X}, [t_a,t_b]) \in \mathcal{C}$ will be
	\begin{itemize}
		\item temporally maximum if $\forall (\mathcal{Y}, [t_a^{'},t_b^{'}]) \in \mathcal{S} \setminus (\mathcal{X}, [t_a,t_b])$, $t_b-t_a \geq t_b^{'} - t_a^{'}$.
		\item cardinally maximum if $\forall (\mathcal{Y}, [t_a^{'},t_b^{'}]) \in \mathcal{S} \setminus (\mathcal{X}, [t_a,t_b])$, $\vert \mathcal{X} \vert \geq \vert \mathcal{Y} \vert$.
	\end{itemize}
\end{mydef}
In this paper, we study the problem of listing out all the maximal $(\Delta,\gamma)$ cliques of a given temporal network, which we call as the Maximal $(\Delta, \gamma)$\mbox{-}Clique Enumeration Problem defined next.
\begin{mydef}[Maximal $(\Delta, \gamma)$\mbox{-}Clique Enumeration Problem] \label{Def:MDG Problem}
Given a temporal network $\mathcal{G}(V, E, \mathcal{T})$, $\Delta$, and $\gamma$ the maximal $(\Delta, \gamma)$\mbox{-}Clique Enumeration Problem asks to list out all the maximal $(\Delta, \gamma)$\mbox{-}Cliques (as mentioned in Definition \ref{Def:MDG}) present in $\mathcal{G}$.
\end{mydef}
Next, we proceed to describe the proposed enumeration methodology for maximal $(\Delta, \gamma)$\mbox{-}Cliques.
\section{Proposed Enumeration Technique} \label{Sec:PET}
As stated earlier, the proposed methodology is broadly divided into two steps and each of them is described in the following two subsections. The broad idea of the proposed enumeration process is as follows: given all the links with time duration of the temporal network, initially, we find out the maximal cliques of cardinality two. We call this phase as the \emph{Stretching} phase, because all the cliques after this phase are duration wise maximal, as if, we are stretching the cliques across the time horizon. Next, taking these duration wise maximal cliques, we add vertices into the clique without violating the definition of $(\Delta, \gamma)$\mbox{-}clique, as if, we are putting vertices into the initialized cliques to make them bulk. Hence, duration of the newly generated cliques are shrinking. Hence, we call the second phase as the \emph{Shrink and Bulk} Phase.

\subsection{Stretching Phase (Initialization)}
Algorithm \ref{Algo:Ini} describes the initialization process of the proposed methodology. For a given temporal network $\mathcal{G}$, initially, we construct the dictionary $\mathcal{D}_{e}$ with the static edges as the \textit{keys} and correspondingly, the occurrence time stamps are the \textit{values}. By the definition of $(\Delta, \gamma)$\mbox{-}clique, if the end vertices of an edge is part of a clique, then the edge has to occur atleast $\gamma$ times in the link stream. Hence, for each static edge $(uv)$ of $\mathcal{G}$, if its frequency is at least $\gamma$, it is processed further. The occurrence time stamps of $(uv)$ are fed into the list $\mathcal{T}_{(uv)}$. A temporary list, $Temp$, is created to store each current processing timestamps from $\mathcal{T}_{(uv)}$ with its previous occurrences, till it has maintained $(\Delta, \gamma)$\mbox{-}clique property. Now, the \texttt{for}\mbox{-}loop from Line 8 to 32 computes all the $(\Delta, \gamma)$\mbox{-}cliques with maximum duration where $\{u,v\}$ is the vertex set. During the processing of $\mathcal{T}_{(uv)}$, any of the following two cases can happen. In the first case, if the current length of $Temp$ is less than $\gamma$, the difference between the current timestamp from $\mathcal{T}_{(uv)}$ and the first entry of $Temp$ is checked (Line 10). Now, if the difference is less than or equal to $\Delta$, current timestamp is appended in $Temp$. Otherwise, all the previous timestamps that have occurred within past $\Delta$ duration from the current timestamp are added in $Temp$ (Line 14). This process basically checks $\Delta$ timestamp backward from each occurrence times of the static edge $(u,v)$. In the second case, when the current length of $Temp$ is greater than or equal to $\gamma$, it is checked whether the current processing time from $\mathcal{T}_{(uv)}$ falls within the interval of (last $\gamma$-th occurrence time + 1) to (last $\gamma$-th occurrence time + 1 + $\Delta$). Now, if it is true, the current timestamp is appended in $Temp$. It can be easily observed that this appending is done iff the at least consecutive $\gamma$ occurrences are within each $\Delta$ duration. Otherwise, the clique is added in $ \mathcal{C}^I_T$ with the vertex set $\{u,v\}$ and time interval $[t_a, t_b]$ (Line 22), where $t_a$ is the $\Delta$ ahead timestamp from the first $\gamma$-th entry in $Temp$ and $t_b$ is the $\Delta$ on\mbox{-}wards timestamp from the last $\gamma$-th entry in $Temp$. Next, all the previous timestamps that have occurred within past $\Delta$ duration from the current timestamp are added in $Temp$ as before (Line 24). It allows to consider overlapping clique. Now, this may happen when we process the last occurrence from $\mathcal{T}_{(uv)}$, it is added in $Temp$. However, no clique can be added by the condition of 9 to 26 if the length of $Temp$ is greater than or equal to $\gamma$. This situation is handled by Line 27 to 31. This process is iterated for each key from the dictionary $\mathcal{D}_{e}$. Now, we present few lemmas and all together they will help to argue the correctness of the proposed methodology.

\begin{algorithm}[h]
	\SetAlgoLined
	\KwData{The temporal network $\mathcal{G}(V, E, \mathcal{T}), \ \Delta, \ \gamma \in \mathbb{Z}^{+}$.}
	\KwResult{The initial clique set $ \mathcal{C}^I_T$ of $\mathcal{G}$}
	$\text{Construct the Dictionary } \mathcal{D}_{e}$\;
%	$\text{Prepare the Static Graph } G$\;
	$\mathcal{C}^I_T=\phi$\;
	\For{\text{Every } $(uv) \in \mathcal{D}_{e}.keys()$}{
		\If{$f_{(uv)} \geq \gamma$}{
			$\mathcal{T}_{(uv)}=\text{Time Stamps of } (uv)$\;
			$Temp=[ \ ]$\;
			$Temp.append(\mathcal{T}_{(uv)}[1])$\;
			
			\For{ $i=2$ to $len(\mathcal{T}_{(uv)})$}{

			   \eIf{$len(Temp) < \gamma$}{
			   	 \eIf{$\mathcal{T}_{(uv)}[i]-Temp[1] \leq \Delta$}{
			   	 	$Temp.append(\mathcal{T}_{(uv)}[i])$\;
		   	    	} 
	   	 	      {
	   	 	      $Temp=[ \ ]$\;
		   	 	   $Temp.append(\text{time stamps of the links occured in previous } \Delta \ Duration )$
	   	 	      }
			   	
			   }
		       %\uElseIf{$len(Temp) \geq \gamma$}{
		   	     {
		   	     \eIf{$Temp[len(Temp) - \gamma+1] +1 + \Delta  \geq \mathcal{T}_{(uv)}[i] $}{
		   	     	$Temp.append(\mathcal{T}_{(uv)}[i])$\;}
	   	     	   {
		   	     	$t_a = Temp[\gamma] - \Delta$ \tcp*{first $\gamma$\mbox{-}th occurrence of $(u,v)$ in \textit{Temp}}
		   	     	$t_b = Temp[len(Temp) - \gamma+1] + \Delta$ \tcp*{last $\gamma$\mbox{-}th occurrence of $(u,v)$ in \textit{Temp}}
		   	     	$\mathcal{C}^I_T.\texttt{addClique}(\{u,v\}, [t_a, t_b])$\;
		   	     	$Temp=[ \ ]$\;
		   	     	$Temp.append(\text{time stamps of the links occured in previous } \Delta \ Duration )$
	   	     	   }
		   	     
		       }

			  \If{$i = len(\mathcal{T}_{(uv)}) \text{ and }  len(Temp) \geq \gamma$}{
			  	$t_a = Temp[\gamma] - \Delta$ \tcp*{first $\gamma$\mbox{-}th occurrence of $(u,v)$ in \textit{Temp}}
			  	$t_b = Temp[len(Temp) - \gamma+1] + \Delta$ \tcp*{last $\gamma$\mbox{-}th occurrence of $(u,v)$ in \textit{Temp}}
			  	$\mathcal{C}^I_T.\texttt{addClique}(\{u,v\}, [t_a, t_b])$\;
		  	}

			}
		}
	}
	\caption{Stretching Phase of the $(\Delta, \gamma)$\mbox{-}Clique Enumeration}
	\label{Algo:Ini}
\end{algorithm}

\begin{mylem}\label{Lemma:ini_1}
 For a link $(uv)$, if there exist any consecutive $\gamma$ occurrences within $\Delta$ duration, then it has to be in `$Temp$' at some stage, in Algorithm \ref{Algo:Ini}.
\end{mylem}
\begin{proof}
	Follows from the description of Algorithm \ref{Algo:Ini}.
\end{proof}

\begin{mylem}\label{Lemma:ini_1a}
In any arbitrary iteration of the `for loop' at Line 8 in Algorithm \ref{Algo:Ini}, each consecutive $\gamma$ occurrences of `$Temp$' will be within $\Delta$ duration.

 %the entries of `Temp' always follow the condition: ``Each consecutive $\gamma$ occurences of any particular $(uv)$ be within $\Delta$ duration". 
\end{mylem}
\begin{proof}
To prove this statement, we use the method of \textit{contradiction}. Initially, $Temp$ contains the first occurrence of a link. Now, when the length of $Temp$ is less than $\gamma$ (Line 9), next occurrence times are added in $Temp$ (Line 11) if the difference from initial to current occurrence time lies within $\Delta$ (Line 10), else the times at which the links have occurred in previous $\Delta$ duration from the current time are added (Line 13, 14). This clears that all the entries in $Temp$ are within $\Delta$ duration when the length of $Temp$ is less than $\gamma$. 

When the length of $Temp$ is greater than or equal to $\gamma$, without the loss of generality, let us take any arbitrary $\gamma$ occurrences of $Temp$ as $ t^1, t^2, \dots t^{(\gamma-1)}, t^{\gamma}$, which is not within $\Delta$ duration, i.e., $t^{\gamma} -t^1 > \Delta$. Let us also assume that from $t^{(\gamma-1)}$, all the previous occurrences in $Temp$ follow the statement of this lemma. %Without loss of generality, let us also assume that its first $\gamma-1$ occurences are within $\Delta$ duration, i.e., $t^{(\gamma-1)} - t^1 \leq \Delta$. 
Now, from our assumptions, we have the following conditions:
\begin{equation} \label{Eq:1}
t^{0}+\Delta \geq t^{\gamma-1} \implies t^{1}+\Delta>t^{(\gamma-1)}
\end{equation}

\begin{equation} \label{Eq:2}
t^{1}+\Delta<t^{\gamma}
\end{equation}

\begin{equation} \label{Eq:3}
t^{1} \geq t^{0}+1
\end{equation}

Now, let us assume the previous occurrence of the link from $t^1$ in $Temp$ is $t^0$ and our goal is to infer the possible positions of $t^{0}$ in the time horizon. From the definition of $(\Delta, \gamma)$\mbox{-}clique, there will be $\gamma$ occurrences from $t^{1}-\Delta$ to $t^{1}$. If first $(\gamma-1)$ links have occurred in consecutive times then $t^{0}=t^{1}-\Delta + \gamma -2$. This is the minimum value for $t^{0}$. From Equation \ref{Eq:3}, the maximum value for $t^{0}$ is $t^{1}-1$. Hence, $t^{0} +1 \leq t^{1} \leq t^{0}+ \Delta +2 -\gamma$. Now, from Equation \ref{Eq:2}, we have $t^{0} + \Delta + 1 < t^{\gamma}$, when $t^{1} = t^{0} +1$ and replacing $t^1$ with $ t^{0}+ \Delta +2 -\gamma$ in Equation \ref{Eq:2}, we get $t^{0} + \Delta + 1 + (\Delta+1 -\gamma) < t^{\gamma} \implies t^{0} + \Delta + 1 < t^{\gamma} $ as $\Delta+1 \geq \gamma$. This violates the condition imposed in Line 17. Hence, $t^{\gamma}$ can not be added in $Temp$. So, we reach the contradiction and this completes the proof. 
%Now, during processing $t^{\gamma}$ the condition imposed in Line 17

\end{proof}

\begin{figure}
	\centering
	\includegraphics[scale=0.8]{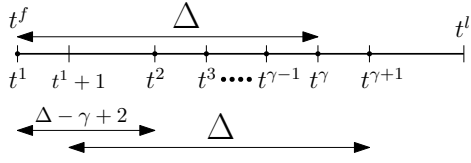}
	\caption{Demonstration Diagram for Lemma \ref{Lemma:ini_1b}} \label{fig:1}
\end{figure}

\begin{mylem}\label{Lemma:ini_1b}
Let, $t^f$ and $t^l$ be the first and last occurrence time of a link in $Temp$. In the interval $[t^f, t^l]$, $Temp$ contains at least $\gamma$ links in each $\Delta$ duration.
\end{mylem}
\begin{proof}
When the length of $Temp$ is less than $\gamma$, Line 9 to 15 in Algorithm \ref{Algo:Ini} allows to hold the statement of the lemma by adding consecutive $\gamma$ occurrences in $\Delta$ duration. So, it is trivial that we need to prove the statement when length of $Temp$ is greater than $\gamma$.
Let us assume that the occurrence times of first $\gamma+1$ entries of $Temp$ are $t^1, t^2, \dots , t^{\gamma}, t^{(\gamma+1)}$, where $t^1=t^f$ and $t^{(\gamma+1)} \leq t^l$. 

Now, by Lemma \ref{Lemma:ini_1a}, $t^{\gamma} - t^1 \leq \Delta$ and $t^{(\gamma+1)} - t^2 \leq \Delta$. Without loss of generality, we want to show that there exist at least $\gamma$ links from $t^1+1$ to $t^1 +1 + \Delta$. As $t^{\gamma} - t^1 \leq \Delta$, the maximum difference between $t^1$ and $t^2$ can be $(\Delta - \gamma + 2)$ and this case will arise when all the $\gamma-1$ links appear in each consecutive timestamp from $t^1+\Delta$ towards $t^1$ (shown in Figure \ref{fig:1}). Now, as $t^{(\gamma+1)} - t^2 \leq \Delta$, we have to show $t^{(\gamma+1)}=t^{\gamma}+1$. This extreme case will intuitively prove the rest of the cases. So, we can infer the following conclusion from Lemma \ref{Lemma:ini_1a} and the assumption $t^2=t^1+\Delta-\gamma+2$. Now,
\begin{equation*}
\begin{split}
t^{(\gamma+1)} - t^2 & \leq \Delta \\
t^{(\gamma+1)} - t^1-\Delta+\gamma-2 & \leq \Delta \\
t^{(\gamma+1)} & \leq  t^1 + \Delta + 1 + \{ (\Delta + 1) - \gamma\} \\
\end{split}
\end{equation*}
Again, from the condition imposed at Line 17 in Algorithm \ref{Algo:Ini}, we also have $t^{(\gamma+1)} \leq  t^1 + \Delta + 1$. Now, as per our assumption of extreme case $t^{\gamma} = t^1 + \Delta$. So, $t^{(\gamma+1)} \leq t^{\gamma} + 1 \implies t^{(\gamma+1)}= t^{\gamma} + 1$. 

Now, as $t^{(\gamma+1)} \leq  t^1 + \Delta + 1$, we can argue $t^{(\gamma+1)} <  t + \Delta $, for all $ t \in (t^1 + 1 , t^2]$. Moreover, from Lemma \ref{Lemma:ini_1a} there is $\gamma$ links within $[t^2, t^{(\gamma+1)} ]$, which concludes the existence of at least $\gamma$ links from $t$ to $t + \Delta$. Now, for any $t^i \in [t^f, t^l-\Delta]$, there will be atleast $\gamma$ links in $Temp$ from $t^i$ to $t^i+\Delta$. This completes the proof of the claimed statement.

\end{proof}

\begin{mylem}\label{Lemma:ini_2}
	In Algorithm \ref{Algo:Ini}, the contents of $\mathcal{C}_{T}^{I}$ are $(\Delta, \gamma)$\mbox{-}Cliques of size $2$.
\end{mylem}
\begin{proof}
	We are processing each static edge of the temporal network $\mathcal{G}$ in its time horizon and add the $(\Delta, \gamma)$\mbox{-}clique(s) formed by the end vertices of the edge into $\mathcal{C}_{T}^{I}$. Hence, the cliques in $\mathcal{C}_{T}^{I}$ are of size 2. Now, in Algorithm \ref{Algo:Ini}, the cliques are added into $\mathcal{C}_{T}^{I}$ in Line 22 and 30. In both the cases, cliques are added if the current length of the $Temp$ is greater than or equal to $\gamma$. As per Lemma \ref{Lemma:ini_1b}, $Temp$ at least $\gamma$ links in each $\Delta$ duration. While adding the duration of the clique, $t_a$ is obtained by subtracting $\Delta$ duration from first $\gamma$-th occurrence time and $t_b$ is obtained by adding $\Delta$ duration from last $\gamma$-th occurrence time in $Temp$. This ensures the existence of at least $\gamma$ occurrences of the link in each $\Delta$ duration between $t_a$ to $t_b$.
\end{proof}

\begin{mylem}\label{Lemma:ini_3}
	All the cliques returned by Algorithm \ref{Algo:Ini} and contained in $\mathcal{C}_{T}^{I}$ are duration wise maximal.
\end{mylem}
\begin{proof}
	
 We prove the duration wise maximality of each clique in $\mathcal{C}_{T}^{I}$ by \textit{contradiction}. Let us assume, a clique $(\{u,v\}, [t_a,t_b]) \in \mathcal{C}_{T}^{I}$ is not duration wise maximal. Then, there exists a $t_a^{'}$ with $t_a^{'} < t_a$ such that $(\{u,v\}, [t_a^{'},t_b])$ is a $(\Delta, \gamma)$\mbox{-}clique or a $t_{b}^{'}$ with $t_{b}^{'} > t_b$ such that $(\{u,v\}, [t_a,t_b^{'}])$ is a $(\Delta, \gamma)$\mbox{-}clique.

 Now, if $(\{u,v\}, [t_a^{'},t_b])$ is a $(\Delta, \gamma)$\mbox{-}clique, then its first $\gamma$ occurrences will be in $Temp$ at some stage as per Lemma \ref{Lemma:ini_1}. Later, this $Temp$ is expanded till $t_b$ either by Line 11 or 18 in Algorithm \ref{Algo:Ini}. Hence,  $(\{u,v\}, [t_a^{'},t_b])$ will be added in $\mathcal{C}_{T}^{I}$, instead of $(\{u,v\}, [t_a,t_b])$. So, the assumption that there exists a $t_a^{'}$ with $t_a^{'} < t_a$ is false. 
 
 Now, by Lemma \ref{Lemma:ini_2}, as $(\{u,v\}, [t_a,t_b])$ is a $(\Delta, \gamma)$\mbox{-}clique, in each $\Delta$ duration within $t_a$ to $t_b$ there will be atleast $\gamma$ links between $u$ and $v$. Let us assume, that $l^{\gamma}$ and $l^{(\gamma -1)}$ are the last $\gamma$\mbox{-}th and $(\gamma-1)$\mbox{-}th occurrence time of $(u,v)$ respectively. From the definition of $(\Delta, \gamma)$\mbox{-}clique, $l^{\gamma}+\Delta \geq t_b$, hence, $l^{(\gamma -1)}+\Delta > t_b$.  Now, to be $\{u, v\}$ a $(\Delta, \gamma)$\mbox{-}clique in the interval $[t_a, l^{(\gamma -1)}+\Delta]$, there must be atleast one link between $u$ and $v$ in the interval $[t_b, l^{(\gamma -1)}+\Delta]$. If there exists such links, it indicates the presence of $\gamma$ or more links in the interval $[l^{(\gamma-1)}, l^{(\gamma -1)}+\Delta]$. This case is handled by Algorithm \ref{Algo:Ini} either in Line $11$ or $18$ and $(\{u,v\}, [t_a,t_b])$ will not be added to $\mathcal{C}_{T}^{I}$. So, there can not exist any $t_b^{'}$ which is greater than $t_b$. 
 
 Hence, all the cliques of $\mathcal{C}_{T}^{I}$ returned by Algorithm \ref{Algo:Ini} are duration wise maximal.

%In the similar way, we assume that $s^{\gamma}$ and $s^{(\gamma -1)}$ are the first $\gamma$\mbox{-}th and $(\gamma-1)$\mbox{-}th occurrence time of $(u,v)$ respectively within $[t_a,t_b]$. Now, to be $\{u, v\}$ a $(\Delta, \gamma)$\mbox{-}clique in the interval $[s^{(\gamma -1)}-\Delta, t_b]$, there must be atleast one link between $u$ and $v$ in the interval $[s^{(\gamma -1)}-\Delta, t_a]$. If there exists such links, it indicates the presence of $\gamma$ or more links in the interval $[s^{(\gamma-1)}-\Delta, s^{(\gamma -1)}]$.

\end{proof}

\begin{mylem}\label{Lemma:ini_4}
	All the duration wise maximal $(\Delta, \gamma)$\mbox{-}cliques of size 2 are contained in $\mathcal{C}_{T}^{I}$.
\end{mylem}
\begin{proof}
	In Lemma \ref{Lemma:ini_2} and \ref{Lemma:ini_3}, we have already shown that each $(\Delta, \gamma)$\mbox{-}clique of $\mathcal{C}_{T}^{I}$ is of size $2$, and duration wise maximal, respectively. Hence, in this lemma, we have to prove that none of such cliques are missed out in the final $\mathcal{C}_{T}^{I}$. As each edge is processed independently by Algorithm \ref{Algo:Ini}, it is sufficient to prove that all the duration wise maximal $(\Delta, \gamma)$\mbox{-}cliques for a particular vertex pair (corresponding to an edge) are contained in $\mathcal{C}_{T}^{I}$.
	
	Let, $(\{u,v\}, [t_a,t_b])$ is a duration wise maximal $(\Delta, \gamma)$\mbox{-}clique and not present in $\mathcal{C}_{T}^{I}$. Now, as $(\{u,v\}, [t_a,t_b])$ is a $(\Delta, \gamma)$\mbox{-}clique, so there exist at least $\gamma$ links in each $\Delta$ duration from $t_a$ to $t_b$, and let $f^{\gamma}$ and $l^{\gamma}$ are the first ${\gamma}$-th and last ${\gamma}$-th occurrence time of the link $(uv)$, between $t_a$ to $t_b$. We denote the occurrence timestamps for the static edge $(u,v)$ as $t^1, t^2, \dots, t^{f_{(uv)}}$, and $f_{(uv)} \geq \gamma$. Now, there can be one of the following cases for the values of $t_a$ and $t_b$.
	\begin{enumerate}[i.]
		\item $t_a = t^{1+\gamma -1} - \Delta$ and $t_b \leq t^{f_{(uv)} - \gamma +1} + \Delta $ : The clique is formed at the beginning of the occurrence stream of $(u,v)$. According to Lemma 1, all the occurrence time will be in $Temp$. Now, if $t_b = t^{f_{(uv)} - \gamma +1} + \Delta$, it will be added in $\mathcal{C}_{T}^{I}$ by Line 30 of Algorithm \ref{Algo:Ini}. Otherwise, $\exists t^k : t^k > l^{\gamma} + 1+ \Delta$ and $t^{k-1} \leq t_b$. Hence, it breaks the if condition at Line 17, and the clique will be added in $\mathcal{C}_{T}^{I}$ by Line 22.
		\item $t_a \geq t^{1+\gamma -1} - \Delta$ and $t_b = t^{f_{(uv)} - \gamma +1} + \Delta $ : The clique is formed at the end of the occurrence stream of $(u,v)$. If $t_a = t^{1+\gamma -1} - \Delta$, it follows from the above case. For the else part, we need to show that $t_a = f^{\gamma}+\Delta > t^{1+\gamma -1} - \Delta $ is handled by the Algorithm \ref{Algo:Ini}. Here, $\exists t^k : t^k < f^{\gamma} -1 - \Delta$ and $t^{k-1} \geq t_a$. Along with Lemma 1 and 2, the Line 14 and 24 are responsible to have all the timestamps within $[t_a, t_b]$ must be $Temp$. So, the clique will be added in $\mathcal{C}_{T}^{I}$ by Line 30.
		\item $t_a > t^{1+\gamma -1} - \Delta$ and $t_b < t^{f_{(uv)} - \gamma +1} + \Delta $ : The clique is formed in the middle of the occurrence stream of $(u,v)$. Both the scenarios of $t_a$ and $t_b$ values are shown in the above two cases, so the clique will be added in $\mathcal{C}_{T}^{I}$ by Line 22.
	\end{enumerate}
%	It is sufficient to show the contradiction that the links occurred within $[t_a, t_b]$ for any of the values of $t_a$ and $t_b$, are within $Temp$
	
	%Now, let $f^{\gamma}$ and $l^{\gamma}$ are the first ${\gamma}$-th and last ${\gamma}$-th occurrence time of the link $(uv)$, between $t_a$ to $t_b$. So, $f^{\gamma} - t_a = \Delta$ and $ t_b - l^{\gamma} = \Delta$, due to duration wise maximality. Among all the occurred links between $u$ and $v$ in $[t_a, t_b]$, other than the appeared ones at $f^{\gamma}$ and $l^{\gamma}$, there may exist other occurrence patterns of $(uv)$, which may lead to $(\{u,v\}, [t_a,t_b])$ is a duration wise maximal $(\Delta, \gamma)$\mbox{-}clique. 
\end{proof}

\begin{mylem}\label{Lemma:ini_5}
	Running time of finding all the duration wise maximal $(\Delta, \gamma)$\mbox{-}cliques of size $2$ in Algorithm \ref{Algo:Ini} is of $\mathcal{O}(\gamma m)$. 
\end{mylem}
\begin{proof}
Preparing the dictionary $\mathcal{D}_{e}$ at Line 1 in Algorithm \ref{Algo:Ini} will take $\mathcal{O}( \sum_{(u,v,t) \in E(\mathcal{G})} f_{(uv)})$. Assuming the frequency of each static edge is atleast $\gamma$, we evaluate the running time for processing a static edge. It will be identical for rest of the edges. During the processing, all the operations from Line 8 to 32 take $\mathcal{O}(1)$ times except, the appending at Line 14 and 24. Now, the appending of previous occurrences within past $\Delta$ duration can leads to copying of at most $\gamma -2 $ previous entries in $Temp$, which takes $\mathcal{O}(\gamma)$ times. Now, the worst case may occur when in every iteration of the for loop at Line 8, $\gamma-2$ previous occurrences are copied in $Temp$ (at Line 24) and this case may occur at most $f_{(uv)} -\gamma + 1$ times. In this case, the running time of the for loop from Line 8 to 32 is $(\gamma-2)(f_{(uv)} -\gamma + 1) \approx \mathcal{O}(\gamma f_{(uv)})$ for a particular static edge. Now, for all the static edges the for loop at Line 3 will run with $\mathcal{O}(\sum_{(u,v,t) \in E(\mathcal{G})} \gamma f_{(uv)} )$ times. Now, the total running time of Algorithm \ref{Algo:Ini} is $\mathcal{O}( \sum_{(u,v,t) \in E(\mathcal{G})} f_{(uv)} +  \gamma \sum_{(u,v,t) \in E(\mathcal{G})} f_{(uv)}) = \mathcal{O}( \gamma \sum_{(u,v,t) \in E(\mathcal{G})} f_{(uv)}) $. Here, summing up all the frequencies of the static edges gives the total number of links of the temporal network, i.e., $m = \sum_{(u,v,t) \in E(\mathcal{G})} f_{(uv)}$. So, the time complexity of the initialization is of  $\mathcal{O}(\gamma m)$.
\end{proof}
We have provided a weak upper bound on running time of the initialization process (Algorithm \ref{Algo:Ini}) in Lemma \ref{Lemma:ini_5}. Now, we focus on space requirement of Algorithm \ref{Algo:Ini}. Storing the Dictionary $\mathcal{D}_{e}$ in Line Number $1$ requires $\mathcal{O}(m)$ space. In the worst case, space requirement by the list $\mathcal{T}_{uv}$ is of $\mathcal{O}(m)$. The size of $Temp$ can go upto the maximum number of times that any static edge has occurred consecutively more than gamma times in each delta duration, and in the worst case it may take $\mathcal{O}(m)$ space. As all the initial cliques are of size $2$, hence space requirement due to $\mathcal{C}_{T}^{I}$ is of $\mathcal{O}(n^{2}.f_{max})$, where $f_{max}$ is the highest frequency of the initial cliques. So, total space requirement by Algorithm \ref{Algo:Ini} is of $\mathcal{O}(m+n^{2}.f_{max})= \mathcal{O}(n^{2}.f_{max})$. Hence, Lemma \ref{Lemma:Space:Algo1} holds.
\begin{mylem} \label{Lemma:Space:Algo1}
The space requirement of Algorithm \ref{Algo:Ini} is of $\mathcal{O}(n^{2}.f_{max})$.
\end{mylem}

\par
Now for the temporal network shown in Figure \ref{Fig:TG}, the initial cliques with $\Delta=3$ and $\gamma=2$, in $\mathcal{C}_{\mathcal{L}}^I$ are $(\{v_1, v_2\}, [1,7])$, $(\{v_1, v_2\}, [7,13])$, $(\{v_1, v_3\}, [2,7])$, $(\{v_1, v_3\}, [8,14])$, $(\{v_2, v_3\}, [2,6])$, $(\{v_2, v_3\}, [7,11])$, $(\{v_2, v_3\}, [5,8])$, $(\{v_2, v_4\}, [4, 12])$, $(\{v_3, v_4\}, [1,9])$, $(\{v_3, v_5\}, [5, 10])$, $(\{v_4, v_5\}, [4,8])$.

\subsection{Shrink and Bulk Phase (Enumeration)}
Algorithm \ref{Algo:Enum} describes the enumeration strategy of our proposed methodology. For the given temporal network $\mathcal{G}$, we construct a static graph $G$ where $V(G)$ is the vertex set of $\mathcal{G}$ and each link of $\mathcal{G}$ induces the corresponding edge in $E(G)$ without the time component, which we call as a static edge. Next, the dictionary $\mathcal{D}$ is built from the initial clique set $\mathcal{C}^I_T$ of Algorithm \ref{Algo:Ini}, where the vertex set of the clique is the key and corresponding occurrence time intervals are the values. This data structure is also updated in the intermediate steps of algorithm \ref{Algo:Enum}. Now, two sets $\mathcal{C}^{\mathcal{T}_1}$ and $\mathcal{C}^{\mathcal{T}_2}$ are maintained during the enumeration process. At any $i$\mbox{-}th iteration of the while loop at Line 5, $\mathcal{C}^{\mathcal{T}_1}$ maintains the current set of cliques which are yet to be processed for vertex addition and $\mathcal{C}^{\mathcal{T}_2}$ stores the new cliques formed in that $i$\mbox{-}th iteration. At the beginning, all the initial cliques from $\mathcal{C}^I_T$ are copied into $\mathcal{C}^{\mathcal{T}_1}$. A clique $(\mathcal{X},[t_a,t_b])$ is taken out from $\mathcal{C}^{\mathcal{T}_1}$ which is duration wise maximal and the IS\_MAX flag is set to true for indicating the current clique as maximal $(\Delta, \gamma)$\mbox{-}clique. For vertex addition, it is trivial to convince that only for the neighboring vertices of $\mathcal{X}$ $(v \in \mathcal{N}_{G}(\mathcal{X}))$, there is a possibility of $(\mathcal{X} \cup \{v\},[t_a^{'},t_b^{'}])$ to be a $(\Delta, \gamma)$\mbox{-}clique. If the new vertex set $\mathcal{X} \cup \{v\}$ is found in $\mathcal{D}$ with one of its value as $[t_a,t_b]$, the IS\_MAX flag is set to false, signifying that the processing clique $(\mathcal{X},[t_a,t_b])$ is not maximal. Otherwise, if $\mathcal{X} \cup \{v\}$ is not present in $\mathcal{D}$, all the possible time intervals in which $\mathcal{X} \cup \{v\}$ can form a $(\Delta, \gamma)$\mbox{-}clique are computed from Line 16 to 37. This process is iterated for all the neighboring vertices of $\mathcal{X}$ (Line $10$ to $38$). Now, we describe the statements from Line 17 to 36 in detail. As mentioned earlier, to form a $(\Delta, \gamma)$\mbox{-}clique with the new vertex set $\mathcal{X} \cup \{v\}$ all the possible combinations from $\mathcal{X} \cup \{v\}$ of size $\vert \mathcal{X} \vert$, (represented as C$(\mathcal{X} \cup \{v\}, \mathcal{X})$), has to be a $(\Delta, \gamma)$\mbox{-}clique. Now, for all $z \in $ C$(\mathcal{X} \cup \{v\}, \mathcal{X})$), if $z$ is present in $\mathcal{D}.keys()$, it signifies the possibility of forming a new clique with the vertex set $\mathcal{X} \cup \{v\}$ (Line 17). Now, all the entries of these combinations are taken into a temporary data structure $\mathcal{D}_{Temp}$ from $\mathcal{D}$. For the clarity of presentation, we describe the operations from Line 19 to 35 for one vertex addition, i.e., $\mathcal{X} \cup \{v\}$ with the help of an example shown in Figure \ref{fig:enum_describe}. Now, let the entries of $\mathcal{D}_{Temp}$ are $z_1, z_2, \dots z_n$, i.e., all $z_i \in$ C$(\mathcal{X} \cup \{v\}, \mathcal{X})$ and the length corresponding entries in $\mathcal{D}_{Temp}$ are $l_1, l_2, \dots l_n$ respectively. So, one sample from $z_1 \otimes z_2 \otimes \dots \otimes z_n$ is taken as $timeSet$ in Line 19 of Algorithm \ref{Algo:Enum}. One possible value of $timeSet$ is $[t_{11}, t_{21}, \dots, t_{n1}]$. For this value, the resultant interval $[t_{a}^{'}, t_{b}^{'}]$ is computed as $t_{11} \cap t_{21} \dots \cap t_{n11}=[max(t_{z_1}^{a^1}, t_{z_2}^{a^1}, \dots, t_{z_n}^{a^1}), \ min(t_{z_1}^{b^1}, t_{z_2}^{b^1}, \dots, t_{z_n}^{b^1})]$. If the difference between $t_{b}^{'}$ and $t_{a}^{'}$ is more than or equal to $\Delta$, then the newly formed $(\Delta, \gamma)$\mbox{-}clique, $(\mathcal{X} \cup \{v\}, [t_{a}^{'}, t_{b}^{'}])$, is added in $\mathcal{C}^{\mathcal{T}_2}$ and $\mathcal{D}$. Also, if   $[t_{a}^{'}, t_{b}^{'}]$ matches with the current interval of $\mathcal{X}$, then the flag $IS\_MAX$ is set to False, i.e., $(\mathcal{X}, [t_a, t_b])$ is not maximal. Now, this step is repeated for all the samples from $z_1 \otimes z_2 \otimes \dots \otimes z_n$ from Line 19 to 35. This ensures that all the intervals in which $\mathcal{X} \cup \{v\}$ forms $(\Delta, \gamma)$\mbox{-}clique are added in $\mathcal{D}$. Now, if none of the vertices from $\mathcal{N}_{G}(\mathcal{X}) \setminus \mathcal{X}$ is possible to add in $\mathcal{X}$, $(\mathcal{X}, [t_a, t_b])$ becomes maximal $(\Delta, \gamma)$\mbox{-}clique and added into final maximal clique set $\mathcal{C}_{\mathcal{L}}$ at Line 40. Vertex addition checking is performed for all the cliques of $\mathcal{C}^{\mathcal{T}_{1}}$ in the while loop from Line 7 to 42. When $\mathcal{C}^{\mathcal{T}_{1}}$ is exhausted and $\mathcal{C}^{\mathcal{T}_{2}}$ is not empty, the contents of $\mathcal{C}^{\mathcal{T}_{2}}$ are copied back into $\mathcal{C}^{\mathcal{T}_{1}}$ for further processing, signifying that all the maximal cliques have not been found yet. This is controlled using the flag $ALL\_MAXIMAL$ in the While loop at Line 5. If no clique is added into $\mathcal{C}^{\mathcal{T}_{2}}$, the flag $ALL\_MAXIMAL$ is set to true so that in the next iteration the condition of the While loop at Line 5 will be false and finally Algorithm \ref{Algo:Enum} terminates. At the end, for the temporal network $\mathcal{G}$, $\mathcal{C}_{T}$ contains all the maximal $(\Delta, \gamma)$\mbox{-}cliques of it. One illustrative example of the enumeration Algorithm is given in Figure \ref{Fig:Algo2_demo}.

\begin{algorithm}
	\SetAlgoLined
	\caption{Shrinking and bulking phase of the maximal $(\Delta, \gamma)$\mbox{-}Clique Enumeration} \label{Algo:Enum}	
	\KwData{A Temporal Network $\mathcal{G}$, Initial Clique Set $\mathcal{C}^I_T, \ \Delta, \ \gamma$.}
	\KwResult{Maximal $(\Delta, \gamma)$ Clique Set $\mathcal{C}_T$ of $\mathcal{G}$.}
	$\text{Construct the Static Graph } G$\;
	$\text{Prepare the dictionary } \mathcal{D} \text{ from } \mathcal{C^I_L}$ \tcp*{ \text{ with the index as vertex set and time intervals as entries}}
	$\mathcal{C}^{\mathcal{T}_1} \leftarrow \mathcal{C^I_L}$\;
	
	$\text{ALL\_MAXIMAL}=False$\;
	\While{$\lnot $ \text{ALL\_MAXIMAL}}  {
		$\mathcal{C}^{\mathcal{T}_2} \leftarrow \phi$\;
		
		\While{$\mathcal{C}^{\mathcal{T}_1} \neq \phi$}{
			Take and remove a clique $(\mathcal{X},[t_a,t_b])$\;
			$\text{IS\_MAX}=True$\;
			\For{Every $v \in \mathcal{N}_{G}(\mathcal{X}) \setminus \mathcal{X}$}{
				$\mathcal{X}_{new}=\mathcal{X} \cup \{v\}$\;
				\eIf{$\mathcal{X}_{new} \in \mathcal{D}$}{
					\If{$[t_a,t_b] \in \mathcal{D}[\mathcal{X}_{new}]$}{
						$\text{IS\_MAX}=False$\;}
				}
				{
					\If{$\forall z \in \{ C(\mathcal{X}_{new}, \mathcal{X})\} \text{ and } z \in \mathcal{D}$}{
						$\mathcal{D}_{Temp} \leftarrow \text{ Get the entries from } \mathcal{D} \text{ for C}(\mathcal{X}_{new}, \mathcal{X})$\;
						
						\ForEach{ permutation of $\mathcal{D}_{Temp}$ entries as $timeSet$}{
							$max\_t_a =[\ ]$\;
							$min\_t_b =[\ ]$\;
							\For{$t \in timeSet$}{
								$max\_t_a.append(t[1])$\;
								$min\_t_b.append(t[2])$\;
							}
							${t^{'}_{a}} = MAX(max\_t_a)$\;
							${t^{'}_{b}} = MIN(min\_t_b)$\;
							\If{$t^{'}_{b} -t^{'}_{a} \geq \Delta$}{
								$\mathcal{C}^{\mathcal{T}_2}.add(\mathcal{X}_{new},[t^{'}_{a},t^{'}_{b}])$\;
								$\mathcal{D}[\mathcal{X}_{new}].append([t^{'}_{a},t^{'}_{b}])$\;
								\If{$t^{'}_{a}= t_a \land t^{'}_{b}=t_b$}{
									$\text{IS\_MAX}=False$\;}
							}
						}
					}	
				}
			}	
			\If{$IS\_MAX$}{
				$\mathcal{C}_T.append(\mathcal{X},[t_a,t_b])$\;}
		}
		\eIf{$len(\mathcal{C}^{\mathcal{T}_2}) > 0$}{
			
			$\mathcal{C}^{\mathcal{T}_1} \leftarrow \mathcal{C}^{\mathcal{T}_2}$\;
			
		}
		{
			$\text{ALL\_MAXIMAL}=True$\;
		}
	}
	
\end{algorithm}

\par Now, from the description of the enumeration process of our proposed methodology, we have the following claims:

\begin{myclaim}
For any arbitrary clique $(\mathcal{X}, [t_a, t_b]) \in \mathcal{C}^{\mathcal{T}_{1}}$ and $v \in \mathcal{N}_G(\mathcal{X}) \setminus \mathcal{X}$, all the time intervals in the whole lifespan of the linked stream $\mathcal{L}$, at which $\mathcal{X} \cup \{v\}$ forms a $(\Delta, \gamma)$\mbox{-}clique are added in $\mathcal{D}$.
\end{myclaim}

\begin{myclaim}
In any arbitrary iteration $i$ of the While loop at Line 5, the cliques of $\mathcal{C}^{\mathcal{T}_{1}}$ and $\mathcal{C}^{\mathcal{T}_{2}}$ are of size $i+1$ and $i+2$ respectively.
\end{myclaim}

\begin{figure}
	\centering
	\includegraphics[scale=0.8]{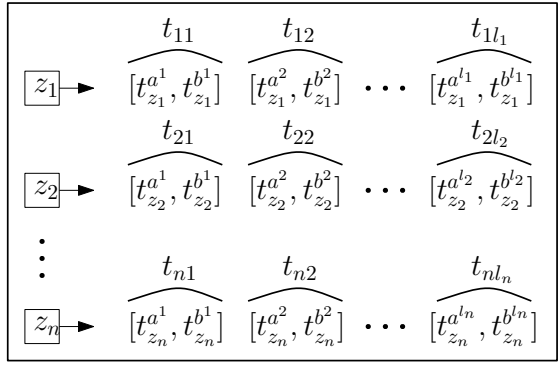}
	\caption{The entries of $\mathcal{D}_{Temp}$ and $z_i \in \mathcal{D}_{Temp}.keys()$}
	\label{fig:enum_describe}
\end{figure}

\begin{mylem}\label{Lemma:enum_1}
	In Algorithm \ref{Algo:Enum}, the elements of $\mathcal{C}_{T}$ are $(\Delta, \gamma)$\mbox{-}cliques.
\end{mylem}
\begin{proof}
%From Lemma \ref{Lemma:ini_2}, all the contents of $\mathcal{C}_{\mathcal{L}}^{I}$ are $(\Delta, \gamma)$\mbox{-}cliques. During the vertex addition (say $v$) to an initial clique $(\mathcal{X},[t_a,t_b])$ all $u \in \mathcal{X}$, the intersection of each intervals of $\{u,v\}$ are taken in the For loop at Line 19 of Algorithm \ref{Algo:Enum}. Let us assume that the obtained interval is $[t_a^{'},t_{b}^{'}]$. It is ensured that  the length of the obtained interval is greater than or equal to $\Delta$ at Line 28. This ensures that each vertex pair of $\mathcal{X}\cup\{v\}$ occurs atleast $\gamma$ times in each $\Delta$ duration within $[t_{a}^{'},t_{b}^{'}]$. Hence, the contents of $\mathcal{C}_{\mathcal{L}}$

All the cliques are added in $\mathcal{C}_{T}$, only from $\mathcal{C}^{\mathcal{T}_{1}}$ at Line 40 in Algorithm \ref{Algo:Enum}. Now, initially $\mathcal{C}^{\mathcal{T}_{1}}$ contains the elements from $\mathcal{C}_{\mathcal{L}}^{I}$, which are $(\Delta, \gamma)$\mbox{-}cliques from Lemma \ref{Lemma:ini_2} and later it is updated with the entries of $\mathcal{C}^{\mathcal{T}_{2}}$. So, if we show that the elements of $\mathcal{C}^{\mathcal{T}_{2}}$ are $(\Delta, \gamma)$\mbox{-}cliques, the statement will be proved. Now, all the cliques of $\mathcal{C}^{\mathcal{T}_{2}}$ are of atleast $\Delta$ duration, from the condition at Line 28. Also, from the description of the Algorithm \ref{Algo:Enum}, it is easy to verify that in each iteration of vertex addition to a clique of $\mathcal{C}^{\mathcal{T}_{1}}$ can only be made, if all the possible combinations of vertices form $(\Delta, \gamma)$\mbox{-}cliques. This ensures that all the vertex pairs of the clique in $\mathcal{C}^{\mathcal{T}_{2}}$ are linked atleast $\gamma$ times in each $\Delta$ duration within the intersected time interval of all the combinations. Hence, the elements of $\mathcal{C}_{T}$ are $(\Delta, \gamma)$\mbox{-}cliques.
\begin{figure}[h]
	\centering
	\includegraphics[scale=0.7]{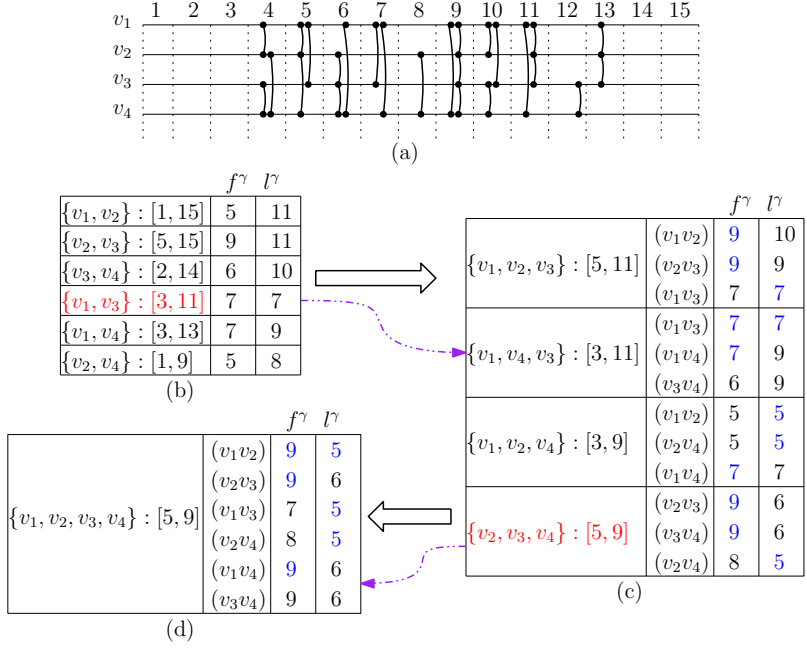}
	\caption{Illustrative example of the proposed Maximal $(\Delta, \gamma)$-Clique Enumeration Algorithm, (a) Input Temporal Graph with $\Delta=4$ and $\gamma=2$, (b) Output of the Algorithm \ref{Algo:Ini} - Stretching Phase, (c)-(d) The content of $\mathcal{C}^{\mathcal{T}_1}$ at Different Iteration of Algorithm  \ref{Algo:Enum}. The cliques in red are duration-wise maximal but not w.r.t. cardinality.} \label{Fig:Algo2_demo}  
\end{figure}

\end{proof}

\begin{mylem}\label{Lemma:enum_2}
	In Algorithm \ref{Algo:Enum}, all the intermediate cliques are duration wise maximal.
\end{mylem}
\begin{proof}
From the proof of Lemma \ref{Lemma:enum_1}, it is sufficient to show that the contents of $\mathcal{C}^{\mathcal{T}_{1}}$ are duration wise maximal. We prove the statement by induction. From Lemma \ref{Lemma:ini_3} the contents of initial clique set are duration wise maximal. Let us assume that in the $i$\mbox{-}th iteration of the While loop at Line 5, the contents of $\mathcal{C}^{\mathcal{T}_{1}}$ are duration wise maximal. We need to show that the same will hold in the $(i+1)$\mbox{-}th iteration also. After adding a vertex to an existing clique obtained in $i$-th iteration for possible expansion, the new vertex set is considered to be a $(\Delta, \gamma)$\mbox{-}clique within the intersected interval of all $(i+2)$\mbox{-}combinations, if the length of the intersected interval is more than $\Delta$ (Line 17 to 36 in Algorithm \ref{Algo:Enum}). Now, it can be observed that the latest first $\gamma$\mbox{-}th occurrence time $(f_{i+1}^{\gamma})$ of the resultant clique must be same with the latest first $\gamma$\mbox{-}th occurrence time $(f_{i}^{\gamma})$ of the constituiting clique from which $t_a$ is coming. Similarly, the earliest last $\gamma$\mbox{-}th occurrence time $(l_{i+1}^{\gamma})$ of the resultant clique must be same with the earliest last $\gamma$\mbox{-}th occurrence time $(l_{i}^{\gamma})$ of the constituiting clique from which $t_b$ is coming. When both the $t_a$, $t_b$ are coming from the same constituting clique, the original clique is not maximal as vertex addition is possible. Now, for the resultant clique, the begining time $t_a$ can not be extended to $t_a - 1$ as in the $i$-th iteration the constituting clique is also duration wise maximal from the assumption, i.e., $f_i^{\gamma} - \Delta=t_a \implies f_{i+1}^{\gamma} - \Delta=t_a$. Similarly, $t_b$ can not be extended to $t_b + 1$ as  in the $i$-th iteration the constituting clique is also duration wise maximal from the assumption, i.e., $l_i^{\gamma} + \Delta=t_b \implies l_{i+1}^{\gamma} + \Delta=t_b$. So, the resultanat clique at $(i+1)$\mbox{-}th iteration is also duration wise maximal. This is true for all the cliques generated in each iteration. Hence, all the intermediate cliques in Algorithm \ref{Algo:Enum} are duration wise maximal.

%During the vertex addition of an arbitrary clique $(\mathcal{X}, [t_a, t_b])$ all the $(i+2)$ combinations are checked 

%Now, let us assume at any $i$\mbox{-}th iteration of the While loop at Line 5 any arbitrary clique $(\mathcal{X}, [t_a, t_b]) \in \mathcal{C}^{\mathcal{T}_{1}}$ is not duration wise maximal. Then, there exists a $t_a^{'}$ with $t_a^{'} < t_a$ such that $(\mathcal{X}, [t_a^{'},t_b])$ is a $(\Delta, \gamma)$\mbox{-}clique or a $t_{b}^{'}$ with $t_{b}^{'} > t_b$ such that $(\mathcal{X}, [t_a,t_b^{'}])$ is a $(\Delta, \gamma)$\mbox{-}clique. As 
\end{proof}

\begin{mylem}\label{Lemma:enum_2a}
In Algorithm \ref{Algo:Enum}, at the begining of any $i$-th iteration, $\mathcal{C}^{\mathcal{T}_{1}}$ holds all the duration wise maximal $(\Delta, \gamma)$\mbox{-}cliques of size $i+1$.
\end{mylem}
\begin{proof}
	For $i=1$, $\mathcal{C}^{\mathcal{T}_{1}}$ holds all the duration wise maximal $(\Delta, \gamma)$\mbox{-}cliques of size $2$ from Lemma \ref{Lemma:ini_4}. Let, $\mathcal{C}^{\mathcal{T}_{1}}_{i-1}$ and $\mathcal{C}^{\mathcal{T}_{1}}_{i}$ are the clique sets at the beginning of the iteration $i-1$ and $i$ respectively and $\mathcal{C}^{\mathcal{T}_{1}}_{i-1}$ holds all the duration wise maximal $(\Delta, \gamma)$\mbox{-}cliques of size $i$. Then, we have to show that during the construction of $\mathcal{C}^{\mathcal{T}_{1}}_{i}$ from $\mathcal{C}^{\mathcal{T}_{1}}_{i-1}$, the clique set $\mathcal{C}^{\mathcal{T}_{1}}_{i}$ remains exhaustive. For a clique from $\mathcal{C}^{\mathcal{T}_{1}}_{i-1}$, we check for all the possible $i+1$ vertex combinations in Line 17 of Algorithm \ref{Algo:Enum}, which does not leave any possible vertex addition to the clique. Next, for each added vertex, all the possible time interval combinations are generated and checked from Line 19 to 35. Now, for each possible time combination, the $(\Delta, \gamma)$\mbox{-}clique is generated from the maximum possible common interval of them. This guarntees that all the possible cliques are generated during this process. Again, from Lemma \ref{Lemma:enum_2}, in the $i$-th iteration all the generated cliques are also duration wise maximal, which are now in $\mathcal{C}^{\mathcal{T}_{1}}_{i}$. So, the same can be proved in the clique building from $i$-th to $i+1$-th iteration. Hence, for any value of $i$ the claimed statement is true.
\end{proof}
\begin{mylem}\label{Lemma:enum_3}
	All the $(\Delta, \gamma)$\mbox{-}Cliques returned by Algorithm \ref{Algo:Enum} and contained in $\mathcal{C}_{T}$ are maximal .
\end{mylem}
\begin{proof}
	We prove this statement by contradiction. Assume that $C_{i}=(\mathcal{X} , [t_a, t_b])$ be an element of $\mathcal{C}_{\mathcal{L}}$, which is not maximal. In Algorithm \ref{Algo:Enum}, the cliques are added in $\mathcal{C}_{\mathcal{L}}$ from $\mathcal{C}^{\mathcal{T}_{1}}$ and all the cliques in $\mathcal{C}^{\mathcal{T}_{1}}$ are duration wise maximal $(\Delta, \gamma)$\mbox{-}cliques from Lemma \ref{Lemma:enum_2}. If, $C_{i}$ is not maximal, then the only thing that can happen is that one or more vertex addition is possible to make $C_{i}$ maximal. Now, let us assume that $\exists v \in \mathcal{N}_{G}(\mathcal{X})$, such that $(\mathcal{X} \cup \{v\} , [t_a, t_b])$ is a $(\Delta, \gamma)$\mbox{-}clique. From the enumaration process described in Algorithm \ref{Algo:Enum}, if a clique is added to $\mathcal{C}_{\mathcal{L}}$, it has to be in $\mathcal{C}^{\mathcal{T}_{1}}$ in any previous iteration. As $(\mathcal{X} \cup \{v\} , [t_a, t_b])$ is a $(\Delta, \gamma)$\mbox{-}clique, the $IS\_MAX$ flag becomes false so that it is not going to be added in $\mathcal{C}_{\mathcal{L}}$ but in $\mathcal{C}^{\mathcal{T}_{2}}$. Hence, the assumption $C_{i} \in \mathcal{C}_{\mathcal{L}} $ is a contradiction. So, all the elements of $\mathcal{C}_{\mathcal{L}}$ returned by Algorithm \ref{Algo:Enum} are maximal $(\Delta, \gamma)$\mbox{-}cliques.  
\end{proof}

\begin{mythem}\label{Theorem:1}
	All the maximal $(\Delta, \gamma)$\mbox{-}Cliques of $\mathcal{G}$ are contained in $\mathcal{C}_{T}$.
\end{mythem}
\begin{proof}
	We prove this statement by contradiction. For the time being assume, that a maximal clique $C_{i}=(\mathcal{X}  , [t_a, t_b])$ of the temporal network $\mathcal{G}$ is not present in $\mathcal{C}_{T}$. Now, the following two cases may happen:
   \begin{itemize}
   	\item $C_{i}$ is a maximal clique of size $2$. From Lemma \ref{Lemma:ini_4}, it is understood that at the begining of Algorithm \ref{Algo:Enum}, $\mathcal{C}^{\mathcal{T}_{1}}$ contains all the duration wise size $2$ maximal cliques. Now, in this situation if none of the following three cases happen:
   	\begin{itemize}
   		\item No vertex addition is possible. If it is so then it will not enter in the for loop at Line 10.
   		\item Vertex addition is possible. However, for the generated possible clique(s), it may happen the duration of the time interval(s) is less than $\Delta$ (Line 28). 
   		\item  Vertex addition is possible and for any of the neighboring vertices the duration of the generated possible clique(s) is greater than $\Delta$, however, none of the intervals are equal with $[t_a,t_b]$ (Line 31).    
   	\end{itemize}
   then $C_i$ is a maximal $(\Delta, \gamma)$\mbox{-}clique of size 2. So, the $IS\_MAX$ flag remains true and the clique $C_i$ is added in $\mathcal{C}_{\mathcal{L}}$.
   \item  $C_{i}$ is a maximal clique of size greater than equal to $3$. Now, without loss of generality, here, we show for $\vert \mathcal{X} \vert=3$ and assume $\mathcal{X}=\{v_i, v_j, v_k\}$. Now, $C_i$ to be a $(\Delta, \gamma)$\mbox{-}clique, it nust have generated from any one of the following three size $2$ $(\Delta, \gamma)$\mbox{-}cliques; let $(\{v_i,v_j\}  , [t_a^{'}, t_b^{'}])$, $(\{v_i,v_k\}  , [t_a^{''}, t_b^{''}])$, and $(\{v_k,v_j\}  , [t_a^{'''}, t_b^{'''}])$, where all the three intervals are super interval of $[t_a, t_b]$. With out loss of generality, we start with the cliques, say, $(\{v_i,v_j\}  , [t_a^{'}, t_b^{'}])$ from $\mathcal{C}^{\mathcal{T}_{1}}$ and $v_k$ is added (Line 10 to 38). Here, $(\{v_i, v_j, v_k\}  , [t_a, t_b])$ is duration wise maximal from Lemma \ref{Lemma:enum_2} and added in $\mathcal{C}^{\mathcal{T}_{2}}$. In the next iteration, $(\{v_i, v_j, v_k\}  , [t_a, t_b])$ is tested for further expansion and as $C_i$ is a maximal $(\Delta, \gamma)$\mbox{-}clique from the assumption, none of the subcases mentioned in Case 1 will occur. So, the $IS\_MAX$ flag will remain true and the clique $C_i$ will be added in $\mathcal{C}_{\mathcal{L}}$. Now, the same will happen for the cliques with larger size as in every iteration all the duration wise maximal $(\Delta, \gamma)$\mbox{-}cliques are generated (by Lemma \ref{Lemma:enum_2a}).
   
   \end{itemize}
Hence, we reach the contradiction. So, for the temporal network $\mathcal{G}$, $\mathcal{C}_{T}$ contains all the maximal $(\Delta, \gamma)$\mbox{-}cliques of it.
\end{proof}
 Theorem \ref{Theorem:1} is basically the correctness statement of the proposed methodology. Next, we proceed towards the analysis of Algorithm \ref{Algo:Enum} for its time and space requirement.

\par  As mentioned previously, $m$ denotes the temporal links in the time varying graph $\mathcal{G}$. At Line Number $2$, computing the static graph from the given time varying graph requires $\mathcal{O}(m)$ time. Time requirement for creating the dictionary $\mathcal{D}$ will be of $\mathcal{O}(|\mathcal{C}_{T}|. f_{max})$ time, where $f_{max}$ denotes the highest number of times a clique appeared. Copying the cliques from the list $\mathcal{O}(\mathcal{C}_{T})$ to $\mathcal{C}^{T_{1}}$ requires $\mathcal{O}(|\mathcal{C}_{T}|)$ time. Setting the $ALL\_MAXIMAL$ flag to `false' in Line Number $4$ requires $\mathcal{O}(1)$ time. So, from Line Number $1$ to $4$, the time requirement is of $\mathcal{O}(m+|\mathcal{C}_{T}^{I}|. f_{max})$. Now, it is easy to verify that the instructions in Line Number $6$, $8$, and $9$ require $\mathcal{O}(1)$ time. The \texttt{for} loop in Line Number $10$ can run at most $\mathcal{O}(n)$ time. Adding the vertex $v$ to the existing clique $\mathcal{X}$ to form $\mathcal{X}_{new}$ in Line Number $11$ requires $\mathcal{O}(1)$ time. The maximum number of comparisons in the condition of the \texttt{if} statement in Line Number $12$ will be $\mathcal{O}(|\mathcal{C}_{T}|)$. In the worst case, each comparison can take at most $\mathcal{O}(n^{2})$ time. Hence, total time requirement for Line Number $12$ requires $\mathcal{O}(|\mathcal{C}_{T}|. n^{2})$ time. Number of comparisons   in the conditional statement in Line Number $13$ requires at most $\mathcal{O}(f_{max})$ time. Setting the $IS\_MAX$ flag to `False' in Line Number $14$ requires $\mathcal{O}(1)$ time. Now, in the \texttt{if} statement of Line Number $17$, the number of combinations can be $\mathcal{O}(n)$ in the worst case. Hence, the number of comparisons for checking the existence in the dictionary $\mathcal{D}$ is of $\mathcal{O}(n|\mathcal{C}_{T}|)$. As mentioned previously, each individual comparison requires $\mathcal{O}(n^{2})$ time. Hence, total execution time for Line $17$ is of $\mathcal{O}(n^{3}.|\mathcal{C}_{T}|)$ time. Now, copying the newly generated combinations from the dictionary $\mathcal{D}$ to $\mathcal{D}_{Temp}$ requires $\mathcal{O}(n f_{max})$. It can be verified from the description of the Algorithm \ref{Algo:Enum} that the number of possible combinations among the time duration is of $\mathcal{O}(f_{max}^{n})$. Hence the \texttt{for} loop in Line Number $19$ will execute $\mathcal{O}(f_{max}^{n})$ times. Line Number $20$ and $21$ takes $\mathcal{O}(1)$ time. Executing the \texttt{for} loop from Line Number $22$ to $25$ requires $\mathcal{O}(n)$ time. Computing the maximum and minimum value among the elements of the list $max\_t_a$ and $min\_t_b$ requires $\mathcal{O}(n)$ time. It is easy to verify that execution of Line Number $28$ to $34$, $39$ to $41$, $43$ to $45$ and $46$ require $\mathcal{O}(1)$ time. Copying the cliques from in Line Number $44$ can take $\mathcal{O}(|\mathcal{C}_{T}|)$ time.  Now, we need to wrap up the computational time requirement for the looping structures to obtain the total time requirement of Algorithm \ref{Algo:Enum}. From the previous analysis, it can be verified that the time requirement for executing the \texttt{for} loop from Line Number $19$ to $35$ will be of $\mathcal{O}(f_{max}^{n}.n)$. The \texttt{for} loop from Line Number $10$ to $38$ will execute at max $\mathcal{O}(n)$ times. Hence, the running time from $10$ to $38$ is of $\mathcal{O}(n(n^{2}.|\mathcal{C}_{T}|.f_{max}+n^{3}.|\mathcal{C}_{T}|+n.f_{max}+f_{max}^{n}.n))=\mathcal{O}(n^{3}.|\mathcal{C}_{T}|.f_{max}+n^{4}.|\mathcal{C}_{T}|+n^{2}.f_{max}+f_{max}^{n}.n^{2})=\mathcal{O}(n^{3}.|\mathcal{C}_{T}|.f_{max}+n^{4}.|\mathcal{C}_{T}|+f_{max}^{n}.n^{2})$. The \texttt{while} loop from Line Number $7$ to $42$ can execute at most $\mathcal{O}(|\mathcal{C}_{T}|)$ times. Hence, execution time of this \texttt{while} loop is of $\mathcal{O}(n^{3}.|\mathcal{C}_{T}|^{2}.f_{max}+n^{4}.|\mathcal{C}_{T}|^{2}+ |\mathcal{C}_{T}|. f_{max}^{n}.n^{2})$. Also, the number of times the \texttt{while} loop from Line Number $5$ to $48$ can execute at most $\mathcal{O}(n)$ times. Hence time requirement for execution of Line Number $5$ to $48$ is $\mathcal{O}(n(n^{3}.|\mathcal{C}_{T}|^{2}.f_{max}+n^{4}.|\mathcal{C}_{T}|^{2}+ |\mathcal{C}_{T}|. f_{max}^{n}.n^{2} + |\mathcal{C}_{T}|))=\mathcal{O}(n^{4}.|\mathcal{C}_{T}|^{2}.f_{max}+n^{5}.|\mathcal{C}_{T}|^{2}+ |\mathcal{C}_{T}|. f_{max}^{n}.n^{3} + n.|\mathcal{C}_{T}|)=\mathcal{O}(n^{4}.|\mathcal{C}_{T}|^{2}.f_{max}+n^{5}.|\mathcal{C}_{T}|^{2}+ |\mathcal{C}_{T}|. f_{max}^{n}.n^{3})$. As already derived that running time from Line Number $1$ to $4$ is of $\mathcal{O}(m+|\mathcal{C}_{T}^{I}|. f_{max})$, hence, total time requirement for Algorithm \ref{Algo:Enum} is of $\mathcal{O}(n^{4}.|\mathcal{C}_{T}|^{2}.f_{max}+n^{5}.|\mathcal{C}_{T}|^{2}+ |\mathcal{C}_{T}|. f_{max}^{n}.n^{3} + m+|\mathcal{C}_{T}^{I}|. f_{max})=\mathcal{O}(n^{4}.|\mathcal{C}_{T}|^{2}.f_{max}+n^{5}.|\mathcal{C}_{T}|^{2}+ |\mathcal{C}_{T}|. f_{max}^{n}.n^{3})$. Maximum number of cliques could be at max $2^{n}$. Hence, plugging the worst case value of $|\mathcal{C}_{T}|$, we have the running time of Algorithm \ref{Algo:Enum} is $\mathcal{O}(n^{4}.2^{2n}.f_{max}+n^{5}.2^{2n}+ 2^{n}. f_{max}^{n}.n^{3})$.
\par Additional space requirement of the Algorithm \ref{Algo:Enum} is due to the `static graph' $G$, which requires $\mathcal{O}(m)$ space; dictionary $\mathcal{D}$, which requires $\mathcal{O}(|\mathcal{C}_{T}^{I}|. f_{max})$ space; dictionary $\mathcal{D}_{Temp}$ which requires $\mathcal{O}(n.f_{max})$ space, the list $\mathcal{X}_{new}$ which requires $\mathcal{O}(n)$ space, the lists $\mathcal{C}^{\mathcal{T}_{1}}$, $\mathcal{C}^{\mathcal{T}_{2}}$, and $\mathcal{C}_{T}$ in the worst case these may require $\mathcal{O}(n 2^{n})$ space; the lists $max\_t_{a}$ and  $min\_t_{b}$ which require $\mathcal{O}(|\mathcal{C}_{T}|)$ space. Hence, total space requirement of Algorithm \ref{Algo:Enum} is of $\mathcal{O}(m+|\mathcal{C}_{T}^{I}|. f_{max}+n.f_{max} + n + n.2^{n}+2^{n})=\mathcal{O}(m+|\mathcal{C}_{T}^{I}|. f_{max}+n.f_{max}+ n.2^{n})$. Hence, Lemma \ref{Lemma:Analysis_Algo_2} holds.
\begin{mylem} \label{Lemma:Analysis_Algo_2}
Running time and space requirement of Algorithm \ref{Algo:Enum} is of $\mathcal{O}(n^{4}.2^{2n}.f_{max}+n^{5}.2^{2n}+ 2^{n}. f_{max}^{n}.n^{3})$ and $\mathcal{O}(m+|\mathcal{C}_{T}^{I}|. f_{max}+n.f_{max}+ n.2^{n})$, respectively.
\end{mylem}
As mentioned previously, Algorithm \ref{Algo:Ini} and \ref{Algo:Enum} together constitute the proposed enumeration strategy for maximal $(\Delta,\gamma)$\mbox{-}Cliques of a temporal network. It has been shown in Lemma \ref{Lemma:ini_5} that the time requirement of Algorithm \ref{Algo:Ini} is of $\mathcal{O}(\gamma.m)$. Hence, total time requirement of the proposed methodology (i.e., Algorithm \ref{Algo:Ini} and \ref{Algo:Enum}) is of $\mathcal{O}(n^{4}.2^{2n}.f_{max}+n^{5}.2^{2n}+ 2^{n}. f_{max}^{n}.n^{3} + \gamma.m)$. As mentioned in Lemma  \ref{Lemma:Space:Algo1}, the space requirement is of $\mathcal{O}(n^{2}. f_{max})$. Hence, total space requirement of the proposed methodology is of $\mathcal{O}(m+|\mathcal{C}_{T}^{I}|. f_{max}+n.f_{max}+ n.2^{n}+n^{2}. f_{max})=\mathcal{O}(m+|\mathcal{C}_{T}^{I}|. f_{max}+ n.2^{n}+n^{2}. f_{max})$. Now, the Theorem \ref{Theorem:2} states regarding the time and space requirement of the proposed methodology. 

\begin{mythem} \label{Theorem:2}
The computational time and space requirement of the proposed methodology is of $\mathcal{O}(n^{4}.2^{2n}.f_{max}+n^{5}.2^{2n}+ 2^{n}. f_{max}^{n}.n^{3} + \gamma.m)$ and $\mathcal{O}(m+|\mathcal{C}_{T}^{I}|. f_{max}+ n.2^{n}+n^{2}. f_{max})$, respectively.
\end{mythem}

\section{Experimental Evaluation} \label{Sec:Experiment}
In this section, we present the experimental evaluation of the proposed methodology and compare its efficacy with the existing methods from the literature. Initially, we briefly outline the background of the used datasets, followed by the objectives, comparing algorithm description, and result discussion.
\subsection{Description of the Datasets} 
In our experiments, we have used the following datasets:
\begin{itemize}
\item \textbf{Hypertext 2009 dynamic contact network (Hypertext)} \cite{isella2011s}: This dataset was collected during the ACM Hypertext 2009 conference, where the SocioPatterns project deployed the Live Social Semantics application. Conference attendees volunteered to wear radio badges that monitored their face-to-face proximity. The dataset published here represents the dynamical network of face-to-face proximity of ~110 conference attendees over about 2.5 days.
%\item \textbf{Haggle} \cite{chaintreau2007impact}: This dataset contains an undirected network formed by a group of people carried with wireless devices. A node represent a person and there will be a link between two users at a given time if they came within a threshold distance.
\item \textbf{College Message Temporal Network (College Message)}  \cite{panzarasa2009patterns}: This dataset contains the interaction information among a group of students from University of California, Irvine. It contains sequence of tuples of the form $(u,v,t)$, which signifies that the students $u$ and $v$ interacted with a private message at time $t$.
\item \textbf{Bitcoin OTC Trust Weighted Signed Network (Bitcoin)} \footnote{\url{https://snap.stanford.edu/data/soc-sign-bitcoin-otc.html}} \cite{kumar2016edge, kumar2018rev2}: This is who-trusts-whom network of people who trade using Bitcoin on a platform called \emph{Bitcoin OTC}. Since Bitcoin users are anonymous, there is a need to maintain a record of users' reputation to prevent transactions with fraudulent and risky users. Members of Bitcoin OTC rate other members in a scale of -10 (total distrust) to +10 (total trust) in steps of 1. This is a weighted, signed, and directed network. However, as per our requirement, we do not consider the direction
\item \textbf{Infectious SocioPatterns Dynamic Contact Network I \& II (Infectious I (69) \& II (old))} \cite{isella2011s}: This dataset contains the daily dynamic contact networks collected during the Infectious SocioPatterns event that took place at the Science Gallery in Dublin, Ireland, during the artscience exhibition INFECTIOUS: STAY AWAY. This dataset contains set of tuples of the form $(t,u,v)$, where $u$ and $v$ are the anonymous ids of the person who are in contact for at least $20$ seconds.
\end{itemize}
As the name of the datasets are a bit lengthy, hence through out the rest of this paper, we refer to them by their abbreviated names as mentioned in the bracket. Basic statistics of the datasets are given in Table \ref{Tab:Data_Stat}.

\begin{table}[H]
\centering
\caption{Basic statistics of the datasets (with increasing order of number of nodes)}
\label{Tab:Data_Stat}
    \begin{tabular}{ | p{2.5 cm} | p{1cm} | p{1cm} | p{2cm} | p{2cm} |}
    \hline
    Datasets & \#Nodes & \#Links & \#Static Edges & Lifetime/Total Duration \\ \hline
    Hypertext & 113 & 20818 & 2196 & 2.5 Days \\ \hline
    %Haggle & 274 & 28244 & 2899 & 4 Days \\ \hline
    Infectious II (old) & 410 & 17298 & 2765 & 8 Hours \\ \hline
    College Message & 1899 & 59835 & 20296 & 193 Days \\ \hline
    Bitcoin & 5881 & 35592 & 21492 & 5.21 Years \\ \hline
    Infectious I (69) & 10972 & 415843 & 44516 & 80 Days \\ 
     
    \hline
    \end{tabular}
\end{table}

\subsection{Setup of Our Experimentation}
This sub section reports the setup of our experimentation. The only parameters involved in our study are $\Delta$ and $\gamma$. For analyzing a temporal network datasets, one intuitive question will be just to find out the frequently connected groups for a given time duration, which is comparable with the lifetime of the network. For this reason, we select the $\Delta$ value based on the network lifetime only. For the `Infectious II (old)' dataset, we start with the $\Delta$ value of $1$ minute keep on increasing it by $1$ minute till it reaches to $10$ minute. Whereas it is increased in multiplicative order of 10 starting from 1 and 2 minutes to 100 and 200 minutes in the `Infectious (69)' dataset, due to its larger lifetime. The same is followed in `Bitcoin' as well. For the `Hypertext' dataset, we start with a $\Delta$ value of $60$ second and keep on increasing it by $60$ second till we reach to $600$ seconds and then considers $\Delta$ as $1800$ seconds, $3600$ seconds, and $7200$ seconds. For the `College Message' dataset, we choose the $\Delta$ value as $1$, $12$, $64$, $72$, $168$ hours. 
\par For $\Delta$ Clique enumeration in all the datasets, we have to set $\gamma$ value as $1$. Now, for enumerating $(\Delta, \gamma)$\mbox{-}Clique, in case of the  `Infectious II(old)', we start with the $\gamma$ value as $2$, keep on increasing it by $1$ till the maximal clique set becomes empty. In case of `Infectious69' dataset for initial $\Delta$ values (e.g., $60$, $120$) we start $\gamma$ value is chosen similarly with that of the `Infectious II(old)' dataset. However, for larger $\Delta$ values (e.g., $6000$, $12000$), we start with a $\gamma$ value of $5$, and then $10$; next incremented by $10$ till it  reaches $30$, and subsequently incremented by $30$ till it reaches $330$. For the `Bitcoin' dataset, for every $\Delta$ values, if we increase the $\gamma$ value beyond $2$, the maximal clique set becomes null. This can be explained by observing the no. of links per no. of static edges ratio, which is very small compared to the lifespan of the temporal network. Hence, we do not provide the plots in Figures \ref{Fig:results} and \ref{Fig:results_TS}. In case of `College Message' dataset, as the chosen $\Delta$ value is larger, hence the $\gamma$ value is incremented by $5$ till it goes to $20$ and then by $10$ till the maximal clique set becomes empty. 
%\par Next, we mention the aims and objectives of the experimentation.
\subsection{Aims and Objectives of the Experiment}
The goals of the experiments are $5$-folds.
\begin{enumerate}
\item With the change of $\Delta$ and $\gamma$, how the count of maximal cliques changes?
\item With the change of $\Delta$ and $\gamma$, how the highest cardinality among the vertex subsets of the maximal cliques changes?
\item With the change of $\Delta$ and $\gamma$, how the maximum duration of the contact changes?
\item From the computational perspective, with the change of $\Delta$ and $\gamma$, how computational time and space requirement change?
\item As mentioned previously, with $\gamma=1$ we can use the proposed methodology to enumerate $\Delta$ Cliques as well. Hence, our another experimental goal is to repeat all the previous $4$ objectives in the context of $\Delta$ Clique enumeration as well.
\end{enumerate}
\subsection{Algorithms Compared}
In our experiments, we compare the performance of the proposed methodology with the following methods from the literature.
\begin{itemize}
\item \textbf{Virad et al.'s Method} \cite{viard2016computing}: This is the first method proposed to enumerate maximal $\Delta$\mbox{-}Clique of a temporal network.
\item \textbf{Himmal et al.'s Method} \cite{himmel2017adapting}: This method incorporates the famous Born\mbox{-}Kerbosch Algorithm to improve the Virad et al.'s Method.
\item \textbf{Banerjee et al.'s Methods} \cite{DBLP:conf/comad/BanerjeeP19}: This is the existing maximal $(\Delta, \gamma)$\mbox{-}Clique proposed by us in one of our previous studies.
\end{itemize}
We obtain the source code of the first two methodologies as implemented by the respective authors. The proposed methodology is developed in Python 3.4 along with NetworkX 2.0. All the experiments have been carried out on a high performance computing cluster having $5$ nodes, and each of them having $40$ cores and $160$ GB of RAM. Implementations of the algorithms are available at \url{ https://github.com/BITHIKA1992/Delta-Gamma-Clique}.
\subsection{Experimental Results with Discussions}
Here, the experimental results are reported and discussed in detail. First, we focus on $\Delta$\mbox{-}Clique, which is equivalent to $(\Delta, \gamma)$\mbox{-}Clique with $\gamma=1$. The results have been given in Table \ref{Tab:Delta_Data}, \ref{Tab:Delta_Time}, and \ref{Tab:Delta_Space}.
\par Fixing $\gamma=1$, if we keep on increasing $\Delta$ value it is natural the maximum duration among the maximal cliques \footnote{In the rest of the part in this section, unless mentioned maximal clique means maximal $(\Delta, \gamma)$\mbox{-}Clique} will also be increasing. The reason behind this is that with the increase of $\Delta$ value, it is more likely that clique vertices will maintain at least one link for longer duration. Hence, for all the datasets, it has been observed that with the increase of $\Delta$, the maximum duration is also increasing. It is also important to observe that, with the increase of maximum duration for any one of the maximal cliques it may happen that not all the clique vertices will have at least one link in each $\Delta$ duration. In that case one maximal clique will be splitted into two or more cliques. Another possibility is that for a particular $\Delta$ value there are many maximal cliques having only two vertices. Now, if the $\Delta$ value is increased further, then there is a chance that this cliques will be obsolete and these may cause in decreasing the number of maximal cliques. Here, we highlight few results from Table \ref{Tab:Delta_Data}. It can be observed that when the $\Delta$ value has been increased from $3600$ to $43200$ for the College Message dataset, maximum duration is drastically increased from $21761$ to $403018$, however the number of maximal cliques reduced from $33933$ to $25635$. On the other hand for the same dataset when the $\Delta$ value has been incremented from $259200$ to $604800$, maximum duration is also changes from $2322612$ to $6334253$, however, in this case the number of maximal cliques is  increased from $21019$ to $21658$.  
\par Regarding the time and space requirement, it can be observed that for the College Message, Bitcoin, Infectious I (69), Infectious II (old)  dataset the proposed methodology is the fastest one compared to the existing methods. As an example, it can be observed from Table \ref{Tab:Delta_Time} that for $\Delta=604800$, the running time of the proposed methodology is $1.19$ seconds, whereas the same for the method proposed by Himel et al. \cite{himmel2017adapting} and Virad et al. \cite{viard2016computing} is $25.86$ and $133.53$ seconds, respectively. However, the running time of the proposed methodology is more in the  Hypertext dataset. This is due to the density of the dataset and this can be verified from Table \ref{Tab:Data_Stat}. In terms of space requirement, the proposed methodology is almost equivalent with that of the proposed by Himmel et al. \cite{himmel2016enumerating}. Other than the `Bitcoin' dataset, the space requirement of the Viard et al.'s \cite{viard2016computing} methodology is always more than both the proposed as well as the Himmel et al.'s \cite{himmel2017adapting} method. One important point is to observe from Table \ref{Tab:Delta_Space} is that, for all the datasets, in case of both the proposed and Himmel et al.'s \cite{himmel2017adapting} methodologies, with the increase of $\Delta$ value, space requirement does not increases much. In case of Virad et al.'s \cite{viard2016computing} method, computation starts with a link $(u,v,t)$ as a $\Delta$ Clique $(\{u,v\};[t_a,t_b])$, where $t_a=t_b=t$ and extending it by both vertex addition as well as time expansion. During this process, their method stores all the intermediate cliques and hence space requirement for this method is much higher compared to others. Here, we highlight few results from Table \ref{Tab:Delta_Space} as examples. For the `College Message' dataset, for $\Delta=604800$, the space requirement by Vired et al.'s \cite{viard2016computing} method, Himeal et al.'s \cite{himmel2017adapting} method and the proposed methodology are $2426$ MB, $127$ MB, and $138$ MB, respectively. However, for the `Bitcoin' dataset, for $\Delta=604800$ the space requirement for these methods are approximately $143$ MB, $366$ MB, and $161$ MB, respectively. Here, we want to highlight that for the Infectious I dataset with the $\Delta$ value as $6000$ and $12000$ both the  computational time and space requirement for Viard et al's \cite{viard2016computing} method too high, and hence we do not mention the results for this two cases.
\par Now, we proceed to describe the results for $(\Delta, \gamma)$\mbox{-}Clique. In Figure \ref{Fig:results}, we show the plots of how the number of maximal cliques, maximum duration and maximum cardinality are changing with the change in $\Delta$ and $\gamma$. It has been observed that for both the `College Message' and `Infectious II (old)' dataset for a fixed $\Delta$, if the $\gamma$ is increased the number of maximal cliques are decreasing. Recall that by the definition of $(\Delta, \gamma)$\mbox{-}Clique, if the value of $\gamma$ is more than $\Delta+1$, then certainly the maximal clique set will be empty. In our experiments, a supportive case has been found. For the `Infectious II (old)' dataset, it has been observed that when the $\Delta$ value is $300$, the $\gamma$ value can be increased till $16$ ($\frac{300}{20}+1=16$). Beyond that the maximal clique set becomes empty. For the `Infectious I (69)' dataset also we make similar observations. However, for initial $\Delta$ values such as $60$, $120$, $600$, $1200$ at the last $\gamma$ value (i.e., just before the maximal clique set becomes empty) the number of maximal cliques increased again. In both the `Hypertext' and `Infectious I (69)' dataset, we observe almost similar pattern. Now, we highlight few numerical results from our experiments. For the `Infectious I (69)' dataset for $\Delta=600$, if $\gamma$ changes from $25$ to $30$, the maximum cardinality drops down from $5$ from $4$. When the $\gamma$ value is further increased to $32$ the maximum cardinality comes down to $0$.
\par In all the datasets, it has been observed that maximum duration among the maximal cliques increases with the increase of $\Delta$ value. Also, the maximum duration decreases with the growth of the clique cardinality. For a fixed $\Delta$ value, the gradual change in $\gamma$ leads to lesser maximum duration. Regarding the maximum cardinality, all the datasets exhibit similar pattern. For a fixed $\Delta$ with a gradual change in $\gamma$ and for a fixed $\gamma$ with a  change in $\Delta$, the maximum clique cardinality decreases and increases, respectively. One major dataset specific observation is that, for fixed-small gamma, the change in maximal clique count is exponential with the increase of $\Delta$ in `Hypertext' and `Infectious I' dataset. Whereas the same is linear in `Infectious II' and `College Message' dataset. However, the increase in $\Delta$ is also exponential in `Infectious I'. This special effect of `Hypertext' can be answered by looking into the plot for maximum cardinality and maximum duration in Figure \ref{Fig:results}. It clearly indicates there exist a certain number of users which communicate very densely, resulting almost no changes in the maximum statistics w.r.t $\Delta$ and $\gamma$. Where the rest of the users follow a sparse communication and do not participate in $(\Delta, \gamma)$-clique formation.
\par Figure \ref{Fig:results_TS} shows the plots for change in computational time and space requirement with the change in $\Delta$ and $\gamma$. Regarding time and space requirement for $(\Delta, \gamma)$ Clique enumeration, the following observations are made. For all the datasets, there is a similarity between the time, space requirement and number of maximal cliques. In general, it has been observed that for a fixed $\Delta$, with the gradual change in $\gamma$, the computational time and space requirement for both the proposed as well as Himel et al.'s method \cite{himmel2017adapting} decreases sequentially, as the number of maximal cliques decreases. There are exceptions also. As an example, for the `Hypertext' dataset, for $\Delta=60$, when $\gamma$ value is increased from $2$ to $3$, the number of maximal cliques has been dropped from 4319 to 3378. However, the computational time and space requirement for the Himel et al.'s method $\Delta=60$ and $\gamma=2$ are $24.26$ sec. and $517.91$ MB., respectively. However, the same with $\gamma=3$ are $25.7$ sec. and $580.24$, respectively. Typically, both the time and space requirement depends upon the intermediate clique. During the enumeration process, if the number of intermediate cliques are more then both the time and space requirement will also be more.
\par From our experiments we can conclude both $\Delta$ and $(\Delta, \gamma)$ Clique enumeration, if the input dataset is sparse then the proposed methodology is better than the   

\begin{table}[H]
\centering
 \caption{Number of Maximal $\Delta$\mbox{-}Cliques, Maximum Duration, and Maximum Cardinality   Enumeration for different datasets}
    \label{Tab:Delta_Data}
%\resizebox{0.72 \textwidth}{!}{ 
\begin{tabular}{ | c | c | p{2 cm} | p{2 cm} | p{2 cm} | }
    \hline
    \multirow{ 2}{*}{\textbf{Dataset}} & \multirow{ 2}{*}{$\Delta$} &\multicolumn{3}{|c|}{\textbf{Algorithm}}\\
    \cline{3-5} 
     &  & \textbf{\# Maximal Cliques} & \textbf{Maximum Cardinality} & \textbf{Maximum Duration}   \\ \hline
        \multirow{ 7}{*}{ \textbf{Hypertext}}  &    60 & 7897 & 7 & 7640 \\ \cline{2-5} 
   & 120 & 6859 & 7 & 8140 \\ \cline{2-5} 
   & 180 & 6453 & 7 & 11520 \\ \cline{2-5} 
   & 240 & 6232 & 7 & 11640 \\ \cline{2-5} 
   & 300 & 6106 & 7& 11760 \\ \cline{2-5} 
   & 360 & 6025 & 7 & 11880  \\ \cline{2-5} 
   & 420 & 5980 & 7& 12000 \\ \cline{2-5}
   & 480 & 5952 & 7 & 12120 \\  \cline{2-5}
   & 540 & 5930 & 7 & 17600 \\  \cline{2-5}
   & 600 & 5913 &  7& 17720 \\  \cline{2-5}
   & 1800 & 5966 & 7 & 31980 \\  \cline{2-5}
   & 3600 & 6473 &  7& 35580 \\  \cline{2-5}
   & 7200 & 7727 &  7& 52020 \\  
    \hline

%    \multirow{ 7}{*}{ \textbf{Haggle}}  &    60 & 32889 & 6 & 237 \\ \cline{2-5} 
%   & 120 & 41868 & 7 & 580 \\ \cline{2-5} 
%  & 180 & 57085 & 10 & 854 \\ \cline{2-5} 
%   & 240 & 74398 & 11 & 1863 \\ \cline{2-5} 
%   & 300 & 96279 & 12 & 2655 \\ \cline{2-5} 
%   & 360 & 112751 & 13 &  4131 \\ \cline{2-5} 
%   & 420 & 136718 & 15 & 4480 \\ \cline{2-5}
%   & 480 & 141089 & 16 &  \\ \cline{2-5}
%   & 540 & 154774 & 18 &  \\ \cline{2-5}
%   & 600 & 136450 & 19 &  \\
%    \hline   

  \multirow{ 5}{*}{\textbf{College Message}}  
    &3600 & 33933  & 4   &  21761   \\ \cline{2-5} 
    &43200 & 25635 &  5  &  403018  \\ \cline{2-5} 
    &88640 & 22701 &  5 &  896134 \\ \cline{2-5} 
    &259200 & 21019 &  5  & 2322612  \\ \cline{2-5}
    &604800 & 21658 &   6 &  6334253 \\ 
    \hline

      \multirow{ 7}{*}{\textbf{Bitcoin}}   &   60 & 32144  & 3 & 180  \\ \cline{2-5} 
    &600& 27572  & 4 & 1800   \\ \cline{2-5} 
    &6000 & 26381  & 8 & 17986  \\ \cline{2-5} 
    &60000 & 26071 & 8 & 179640   \\ \cline{2-5} 
    &3600 & 26577 &  7 & 10791   \\  \cline{2-5} 
    &43200 & 26091 & 8  & 129422  \\  \cline{2-5}
    &88640 & 25970 & 8 &  265798  \\ \cline{2-5}
    &259200 & 26290 & 8 &  777572   \\   \cline{2-5}
    &604800 & 27149 & 8 & 1814344   \\     
   \hline

    \multirow{ 7}{*}{ \textbf{Infectious I (69)}}  &   60 & 161066 & 6 &  3760 \\ \cline{2-5} 
   & 120 & 138662 & 7 & 5180 \\ \cline{2-5} 
   &600 & 128392 & 10 & 11200 \\ \cline{2-5} 
   & 1200 & 139684 & 13 & 12400 \\ \cline{2-5} 
   & 6000 & 152121 & 16 & 22740 \\ \cline{2-5} 
   & 12000 & 152198 & 16 & 34740
 \\ 
    \hline

    \multirow{ 7}{*}{ \textbf{Infectious II (old)}}  &    60 & 9776 & 5 & 1860 \\ \cline{2-5} 
   & 120 & 9397 & 6 & 2900 \\ \cline{2-5} 
   & 180 & 9565 & 7 & 4280 \\ \cline{2-5} 
   & 240 & 9849 & 7 & 5160 \\ \cline{2-5} 
   & 300 & 10192 & 8 & 5280 \\ \cline{2-5} 
   & 360 & 10734 & 8 &  6480 \\ \cline{2-5} 
   & 420 & 11287 & 8 & 6600 \\ \cline{2-5}
   & 480 & 11571 & 9 & 8540 \\ \cline{2-5}
    & 540 & 11781 & 9 & 9580 \\ \cline{2-5}
     & 600 & 12123 & 10 & 9700 \\ 
    \hline
    \end{tabular}
%}
\end{table}

\begin{table}[H]

 \caption{Computational time requirement (in Secs.) for Maximal $\Delta$\mbox{-}clique ($(\Delta, \gamma)$\mbox{-}clique with $\gamma=1$) Enumeration for different datasets}
    \label{Tab:Delta_Time}
%\resizebox{0.72 \textwidth}{!}{ 
\begin{tabular}{ | c | c | c | c | c | }
    \hline
    \multirow{ 2}{*}{\textbf{Dataset}} & \multirow{ 2}{*}{$\Delta$} &\multicolumn{3}{|c|}{\textbf{Algorithm}}\\
    \cline{3-5} 
     &  & \textbf{Viard et al. \cite{viard2016computing}} & \textbf{Himmel et al. \cite{himmel2017adapting}} & \textbf{Proposed}   \\ \hline
        \multirow{ 7}{*}{ \textbf{Hypertext}}  &    60 & 16.02 & 6.14 & 53.9 \\ \cline{2-5} 
   & 120 & 18.11 & 4.48 & 26.14 \\ \cline{2-5} 
   & 180 & 20.17 & 3.8 & 16.35 \\ \cline{2-5} 
   & 240 & 21.73 & 3.73 & 11.76 \\ \cline{2-5} 
   & 300 & 23.11 & 3.62 & 9.73 \\ \cline{2-5} 
   & 360 & 24.02 & 3.31 & 8.55  \\ \cline{2-5} 
   & 420 & 25.54 & 3.41 & 7.76 \\ \cline{2-5}
   & 480 & 26.61 & 3.28 & 7.27 \\  \cline{2-5}
   & 540 & 28.61 & 3.23 & 6.6 \\  \cline{2-5}
   & 600 & 29.83 & 3.08 & 6.12 \\  \cline{2-5}
   & 1800 & 51.22 & 2.56 & 3.22 \\  \cline{2-5}
   & 3600 & 81.15 & 2.35 & 2.4 \\  \cline{2-5}
   & 7200 & 178.79 & 2.4 & 1.8 \\  
    \hline 
%    \multirow{ 7}{*}{ \textbf{Haggle}}  &    60 & 58.95 & 132.28 & 17249.28 \\ \cline{2-5} 
%   & 120 & 129.21 & 121.79 & 5940.86 \\ \cline{2-5} 
%  & 180 & 454.04 & 116.74 & 2964.28 \\ \cline{2-5} 
%   & 240 & 1777.06 & 120.81 & 1933.35 \\ \cline{2-5} 
%   & 300 & 5981.29 & 157.09 & 1603.76 \\ \cline{2-5} 
%   & 360 & 15679.59 & 262.94 & 2128.83  \\ \cline{2-5} 
%   & 420 & 43137.68 & 511.16 & 11862.56 \\ \cline{2-5}
%   & 480 &  & 1073.33 & 56397.09 \\ \cline{2-5}
%   & 540 &  & 2074.01 & x \\ \cline{2-5}
%   & 600 &  & 3498.75 & x \\
%    \hline   
    
  \multirow{ 5}{*}{\textbf{College Message}}  
    &3600 & 35.25  & 41   & 19.84    \\ \cline{2-5} 
    &43200 & 43.02 & 31.38   & 4.19   \\ \cline{2-5} 
    &88640 & 52.29 & 28.56  & 2.28   \\ \cline{2-5} 
    &259200 & 84.05  & 27.61   & 1.41  \\ \cline{2-5}
    &604800 & 133.53 & 25.86   &1.19   \\ 
    
    \hline
    
      \multirow{ 7}{*}{\textbf{Bitcoin}}   &   60 & 16.9  & 200.1 & 2.99  \\ \cline{2-5} 
    &600& 17.62  & 195.75 & 2.48   \\ \cline{2-5} 
    &6000 & 18.51  & 193.48 & 2.39  \\ \cline{2-5} 
    &60000 & 20.29 & 195.51 & 2.32   \\ \cline{2-5} 
    &3600 & 18.31 & 196.49  & 2.36   \\  \cline{2-5} 
    &43200 & 19.58 & 193.74  & 2.41  \\  \cline{2-5}
    &88640 & 20.69 & 190.93 & 2.3    \\ \cline{2-5}
    &259200 & 22.6 & 191.49 & 2.26    \\   \cline{2-5}
    &604800 & 29.69 & 193.52 & 2.34    \\     
   \hline
    \multirow{ 7}{*}{ \textbf{Infectious I (69)}}  &   60 & 274.13 & 1025.01 & 80.75 \\ \cline{2-5} 
   & 120 & 405.84 & 998.53 & 46.88 \\ \cline{2-5} 
   &600 & 2659.82 & 1043.46 & 30.87\\ \cline{2-5} 
   & 1200 & 9824.58 & 1062.69 & 84.13 \\ \cline{2-5} 
   & 6000 & NA & 1238.75 & 108.34 \\ \cline{2-5} 
   & 12000 & NA & 1266.8 & 108.03
 \\ 
    \hline
    \multirow{ 7}{*}{ \textbf{Infectious II (old)}}  &    60 & 10.77 & 4.71 & 4.41 \\ \cline{2-5} 
   & 120 & 17.36 & 4.12 & 2.39 \\ \cline{2-5} 
   & 180 & 26.41 & 4.04 & 1.7 \\ \cline{2-5} 
   & 240 & 37.74 & 3.97 & 1.81 \\ \cline{2-5} 
   & 300 & 51.8 & 4.13 & 1.76 \\ \cline{2-5} 
   & 360 & 68.27 & 4.34 & 1.88  \\ \cline{2-5} 
   & 420 & 92.06 & 4.68 & 2.21 \\ \cline{2-5}
   & 480 & 122.26 & 5.12 & 2.85 \\ \cline{2-5}
    & 540 & 159.73 & 5.45 & 2.75 \\ \cline{2-5}
     & 600 & 201.34 & 5.96 & 3.5 \\ 
    \hline
    \end{tabular}

%}
\end{table}

\begin{table}[H]

 \caption{Space requirement (in MB) for Maximal $\Delta$\mbox{-}clique ($(\Delta, \gamma)$\mbox{-}clique with $\gamma=1$) Enumeration for different datasets}
    \label{Tab:Delta_Space}
%\resizebox{0.72 \textwidth}{!}{ 
\begin{tabular}{ | c | c | c | c | c | }
    \hline
    \multirow{ 2}{*}{\textbf{Dataset}} & \multirow{ 2}{*}{$\Delta$} &\multicolumn{3}{|c|}{\textbf{Algorithm}}\\
    \cline{3-5} 
     &  & \textbf{Viard et al. \cite{viard2016computing}} & \textbf{Himmel et al. \cite{himmel2017adapting}} & \textbf{Proposed}   \\ \hline
        \multirow{ 7}{*}{ \textbf{Hypertext}}  &    60 & 208.414 & 104.05 & 108.99 \\ \cline{2-5} 
   & 120 & 220.8515 & 103.37 & 108.42 \\ \cline{2-5} 
   & 180 & 237.2461 & 103.23 & 108.3672 \\ \cline{2-5} 
   & 240 & 246.3867 & 103.01 & 108.2734 \\ \cline{2-5} 
   & 300 & 256.2031 & 102.9 & 108.3984 \\ \cline{2-5} 
   & 360 & 267.3906 & 102.85 & 108.4921  \\ \cline{2-5} 
   & 420 & 280.3828 & 102.8 & 108.4726 \\ \cline{2-5}
   & 480 & 290.6992 & 102.78 & 108.539 \\  \cline{2-5}
   & 540 & 307.7148 & 102.74 & 108.6132 \\  \cline{2-5}
   & 600 & 318.1289 & 102.72 & 108.6836 \\  \cline{2-5}
   & 1800 & 512.9609 & 102.7 & 109.8047 \\  \cline{2-5}
   & 3600 & 787.2031 & 103 & 111.836 \\  \cline{2-5}
   & 7200 & 1739.3164 & 103.78 & 116.1601 \\  
    \hline 
     
%    \multirow{ 7}{*}{ \textbf{Haggle}}  &    60 & 233.3398 & 123.3945 & 142.2851 \\ \cline{2-5} 
%   & 120 & 410.289 & 129.0273 & 175.8203 \\ \cline{2-5} 
%  & 180 & 1165.8359 & 141.1718 & 291.1289 \\ \cline{2-5} 
%   & 240 & 4260.4609 & 156.3672 & 561.125 \\ \cline{2-5} 
%   & 300 & 13803.1796 & 179.8242 & 1177.8789 \\ \cline{2-5} 
%   & 360 & 51469.7656 & 200.3125 &  2722.2304 \\ \cline{2-5} 
%   & 420 &  & 228.6796 & 5829 \\ \cline{2-5}
%   & 480 &  & 239.3515 &  \\ \cline{2-5}
%   & 540 &  & 258.1562 & \\ \cline{2-5}
%   & 600 &  & 240.3125 &  \\
%    \hline       
  \multirow{ 5}{*}{\textbf{College Message}}  
    &3600 & 371.3281 & 139.3593  &  147.6992   \\ \cline{2-5} 
    &43200 & 545.2851 &  133.039  & 141.7187   \\ \cline{2-5} 
    &88640 & 726.289 & 130.9336 &  139.9922  \\ \cline{2-5} 
    &259200 &  1280.8828 & 127.4179   & 135.2226  \\ \cline{2-5}
    &604800 & 2426.1211  &  127.4531  & 138.7031  \\ 
    
    \hline

      \multirow{ 7}{*}{\textbf{Bitcoin}}   &   60 & 147.12  & 200.61 &  159.4 \\ \cline{2-5} 
    &600& 143.85  & 210.43 &  156.08  \\ \cline{2-5} 
    &6000 & 143.01  & 222.72 & 156.9  \\ \cline{2-5} 
    &60000 & 142.88 & 240.96 &  157.83  \\ \cline{2-5} 
    &3600 & 142.97 &  221.11 & 156.39  \\  \cline{2-5} 
    &43200 & 142.71 & 234.76  & 157.39 \\  \cline{2-5}
    &88640 & 142.57 & 249.75 &  158.06   \\ \cline{2-5}
    &259200 & 142.73 & 287.98 &  159.32   \\   \cline{2-5}
    &604800 & 143.39 & 366.88 & 161.57    \\     
   \hline

    \multirow{ 7}{*}{ \textbf{Infectious I (69)}}  &   60 & 2590.3007 & 285.1953 & 356.5664 \\ \cline{2-5} 
   & 120 & 3914.3789 & 265.9257 & 353.2187 \\ \cline{2-5} 
   &600 & 18116.6328 & 261.2109 & 522.9687 \\ \cline{2-5} 
   & 1200 & 72627.1015 & 277.9687 & 1016.875 \\ \cline{2-5} 
   & 6000 & NA & 292.6367 & 3075.8906 \\ \cline{2-5} 
   & 12000 & NA & 292.6992 & 3096.7539
 \\ 
    \hline

    \multirow{ 7}{*}{ \textbf{Infectious II (old)}}  &    60 & 196.4843 & 106.2383 & 111.0976 \\ \cline{2-5} 
   & 120 & 265.6796 & 105.8007 & 112.2539 \\ \cline{2-5} 
   & 180 & 338.6093 & 105.8242 & 113.7656 \\ \cline{2-5} 
   & 240 & 428.7656 & 105.9609 & 115.3945 \\ \cline{2-5} 
   & 300 & 540.0898 & 106.1875 & 117.7031 \\ \cline{2-5} 
   & 360 & 660.4726 & 106.6562 & 120.957  \\ \cline{2-5} 
   & 420 & 815.0351 & 107.1406 & 124.4296 \\ \cline{2-5}
   & 480 & 1012.617 & 107.5234 & 129.9687 \\ \cline{2-5}
    & 540 & 1244.976 & 107.7695 & 133.1758 \\ \cline{2-5}
     & 600 & 1498.25 & 108.1679 & 137.4414 \\ 
    \hline

    \end{tabular}

%}
\end{table}

\begin{figure*}[!htbp]
	\centering
	\begin{tabular}{ccc}
		\textbf{\underline{Clique Count}} & \textbf{\underline{Maximum Cardinality}} & \textbf{\underline{Maximum Duration}} \\
		\includegraphics[width=5cm,height=4.5cm]{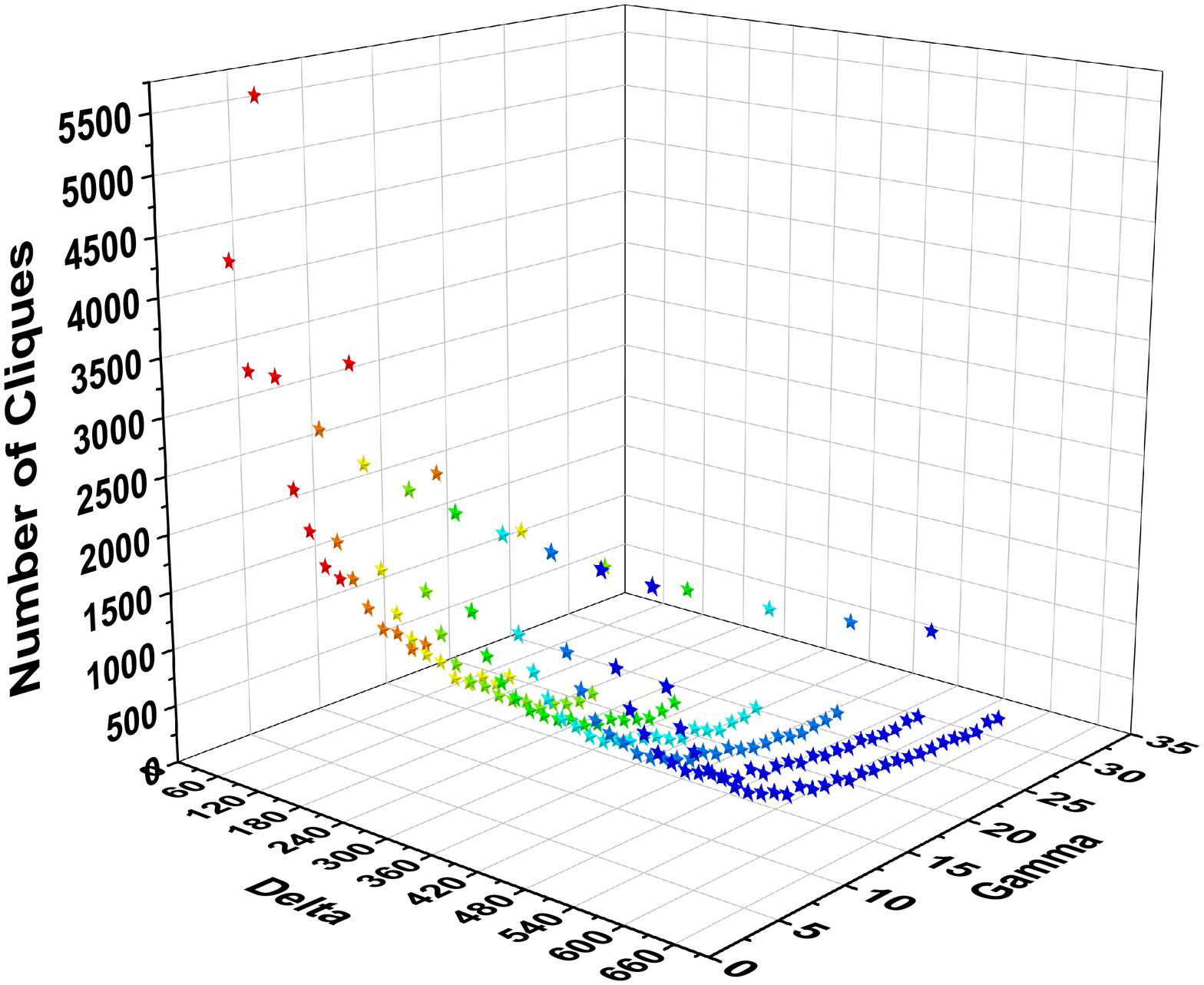} 
		&
		\includegraphics[width=5cm,height=4.5cm]{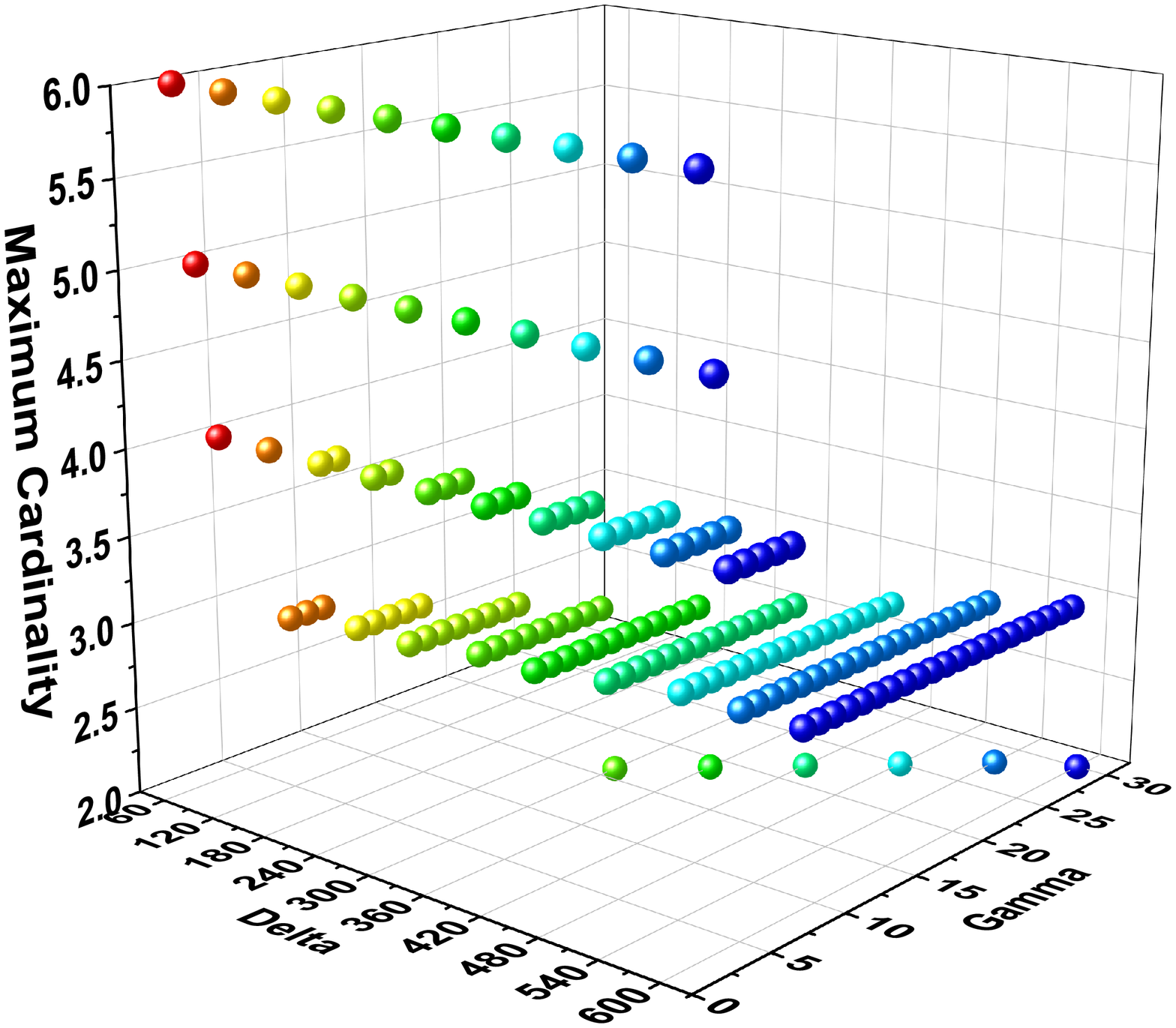}
		&
		\includegraphics[width=5cm,height=4.5cm]{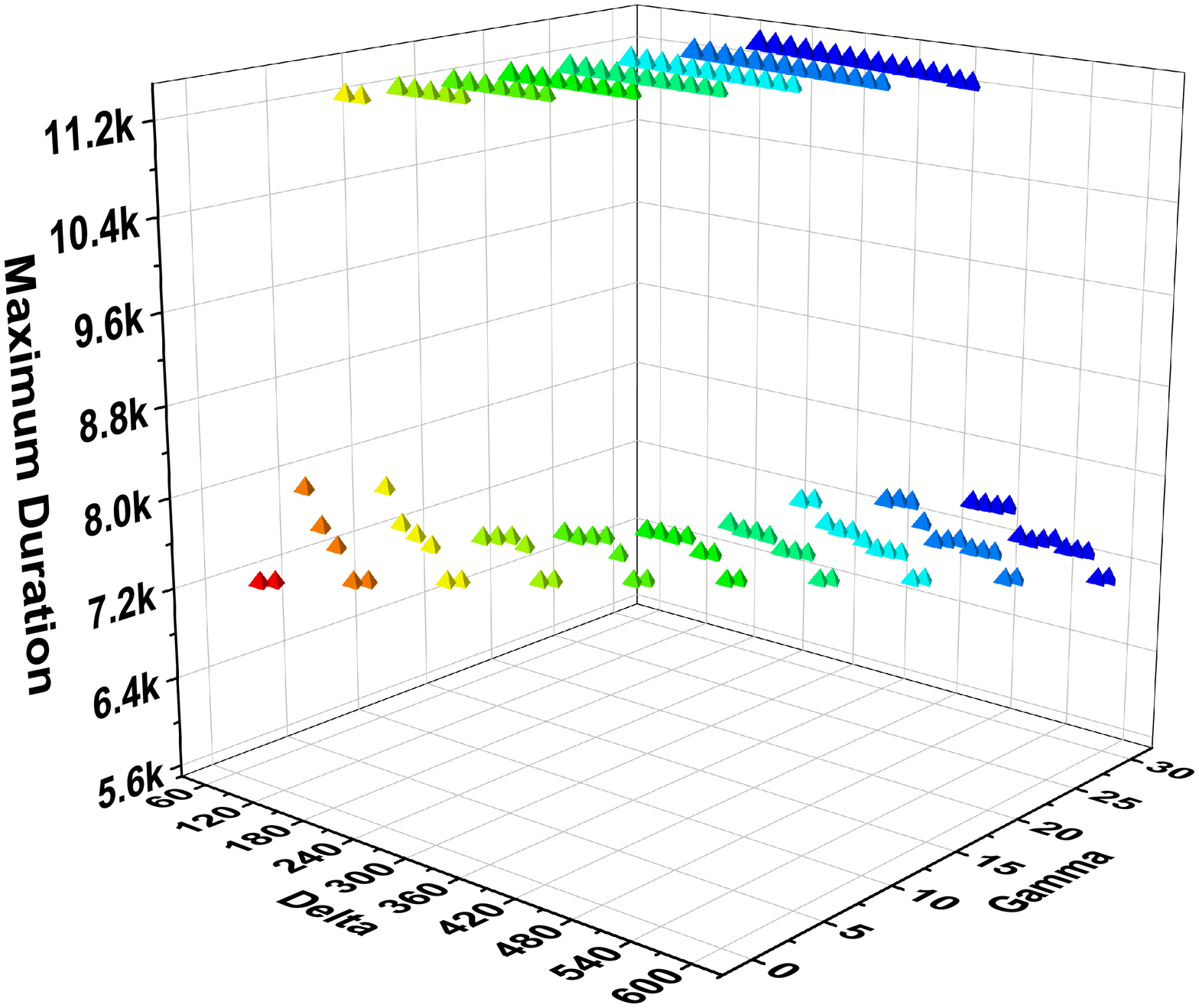}
		\\
		
		& (a) Hypertext & \\
		%
		% \includegraphics[width=5cm,height=4.0cm]{haggle_clique_count_plot.png} 
		% &
		% \includegraphics[width=5cm,height=4.0cm]{Haggle_Maximum_Cardinality} 
		% &
		% \includegraphics[width=5cm,height=4.0cm]{Haggle_Maximum_Duration}
		% \\
		%& (b) Haggle Dataset & \\
		
		\includegraphics[width=5cm,height=4.0cm]{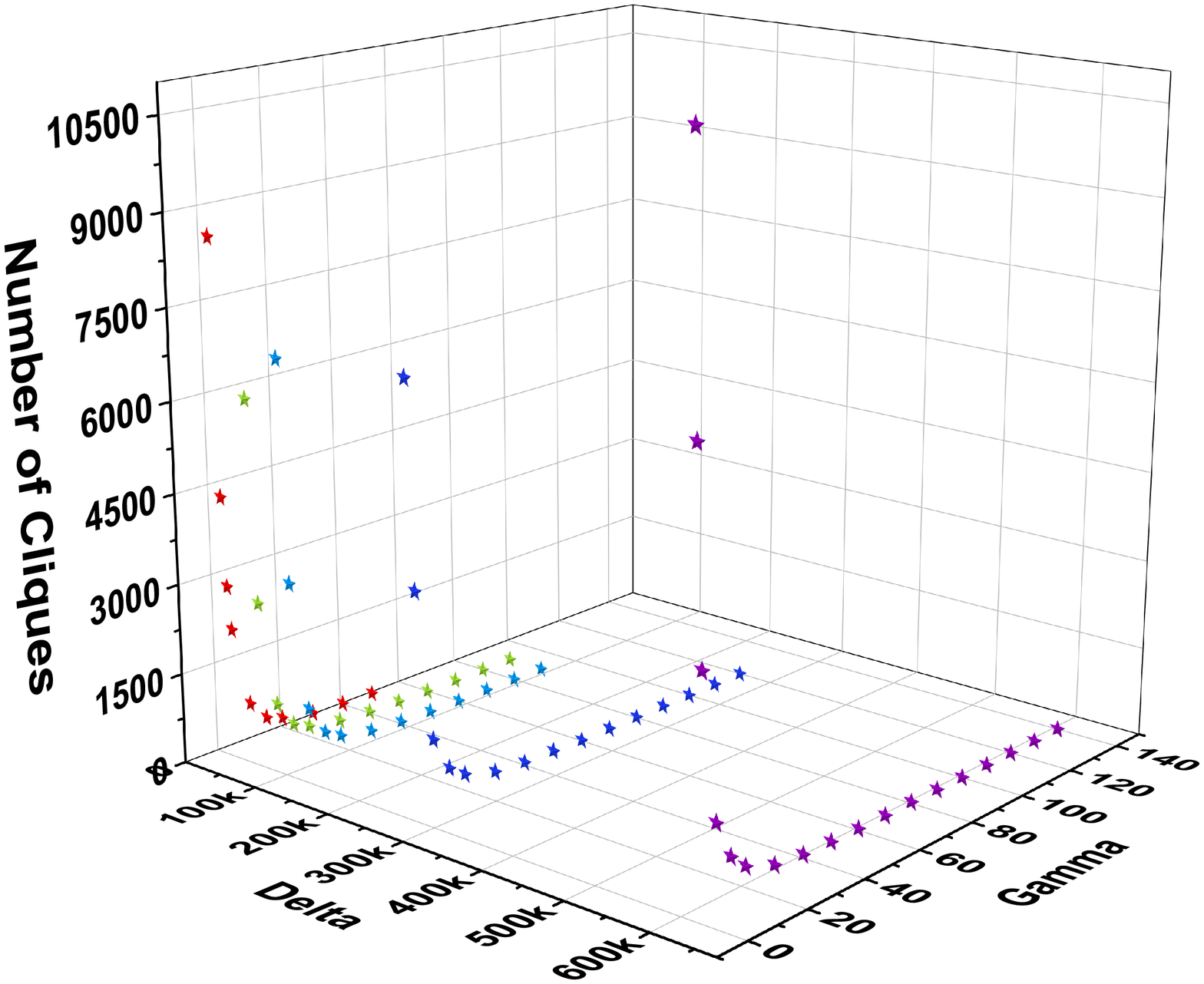} 
		&
		\includegraphics[width=5cm,height=4.0cm]{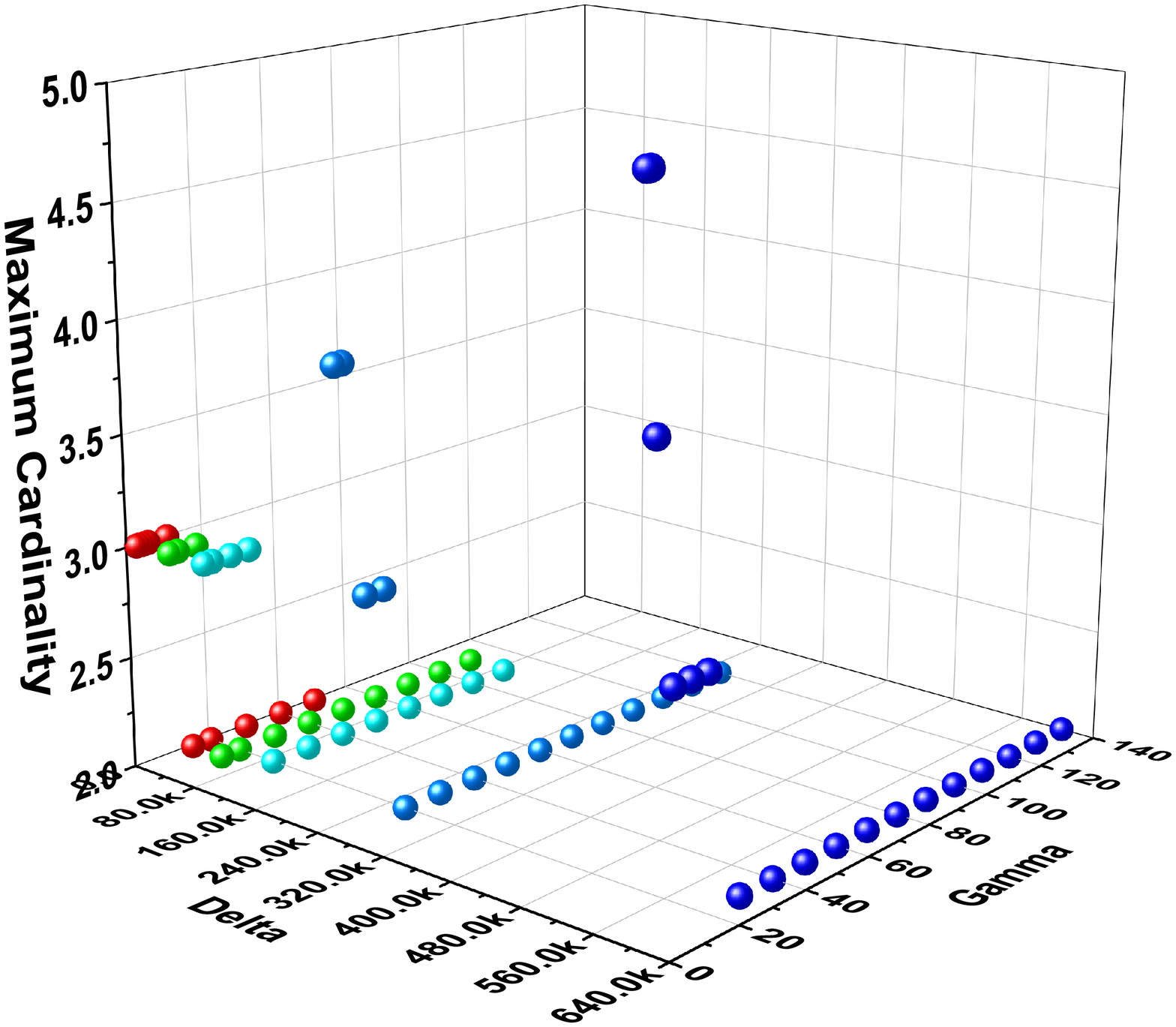} 
		&
		\includegraphics[width=5cm,height=4.0cm]{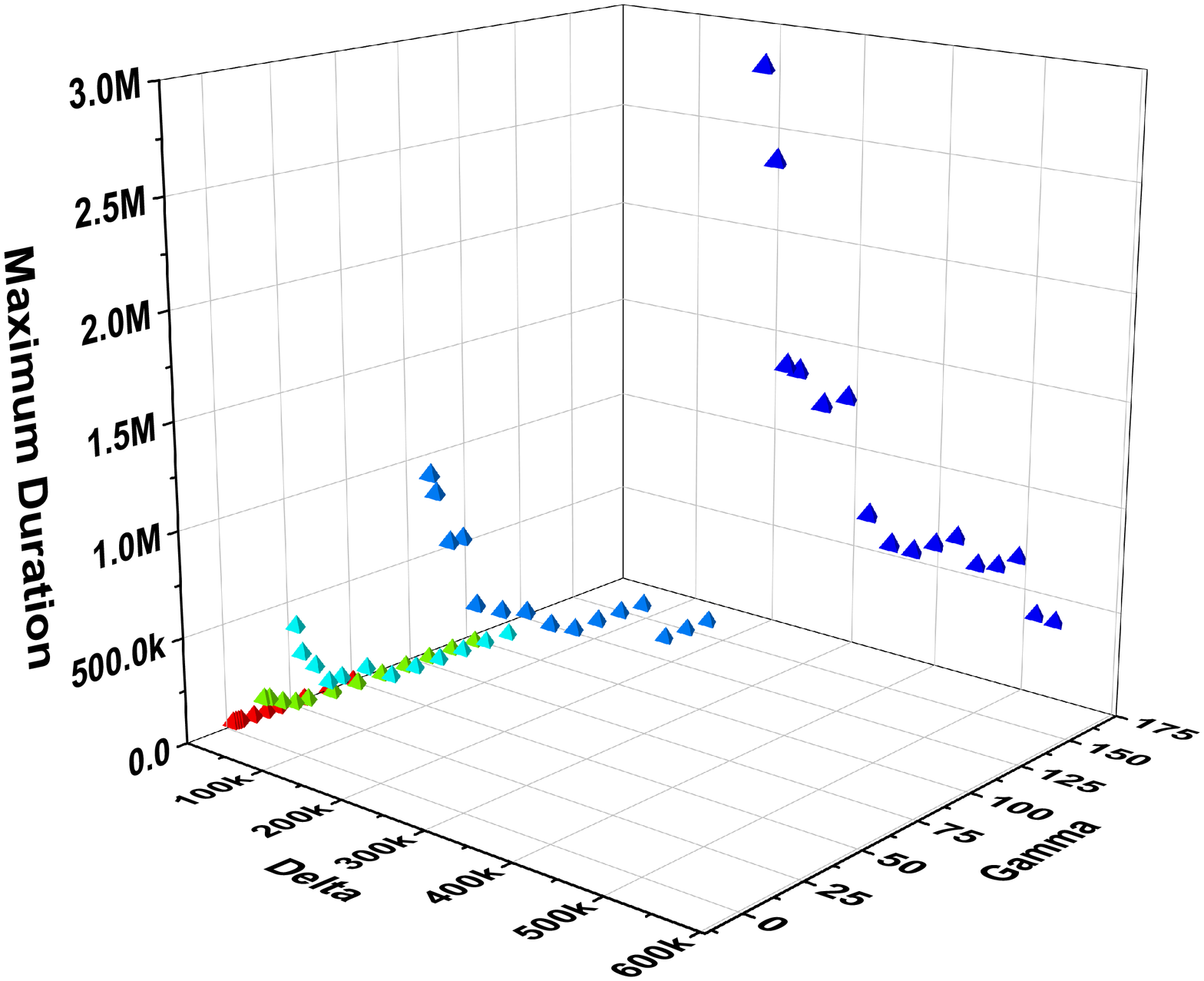} 
		\\
		& (b) College Message &\\
		\includegraphics[width=5cm,height=4.0cm]{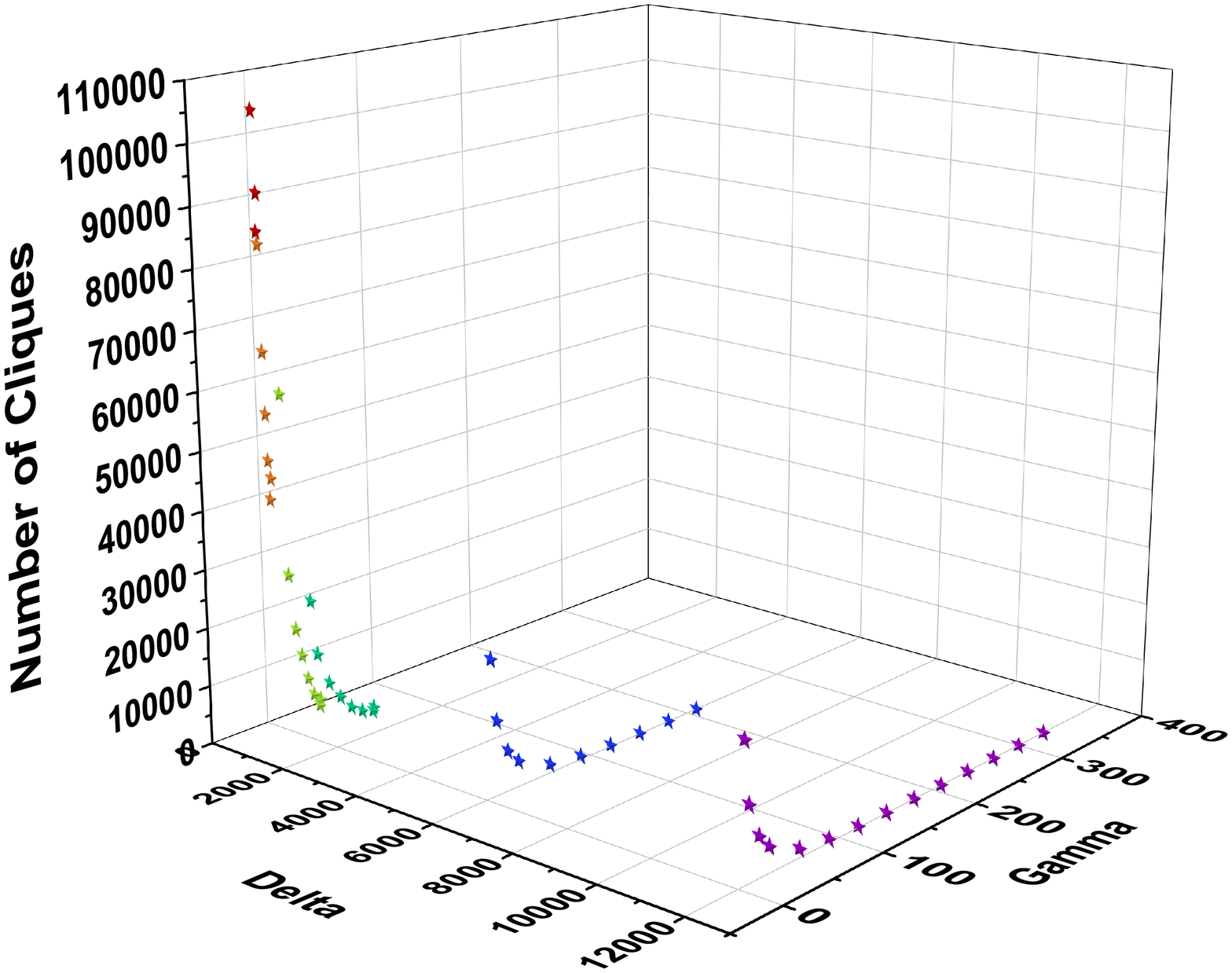} 
		&
		\includegraphics[width=5cm,height=4.0cm]{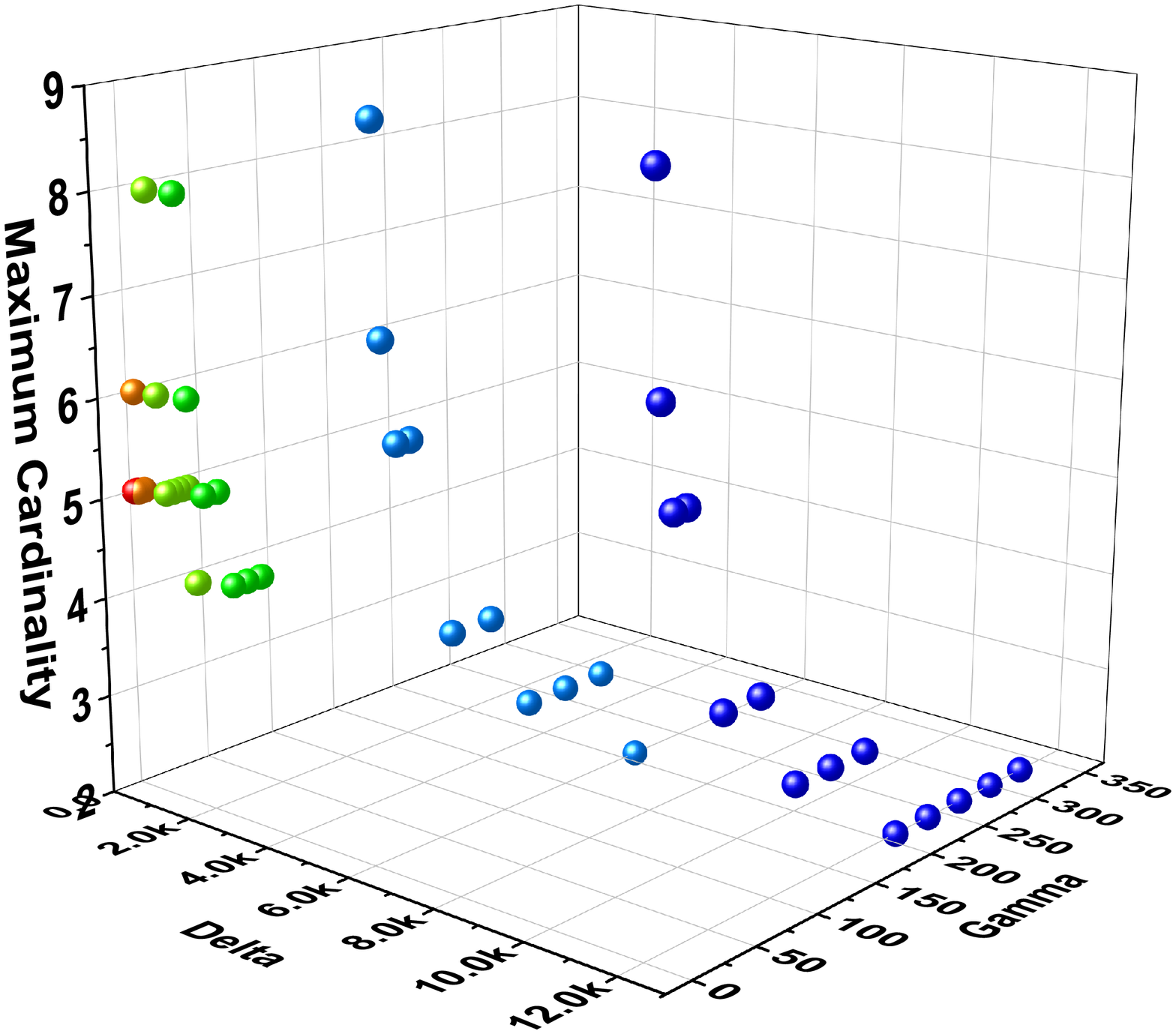} 
		&
		\includegraphics[width=5cm,height=4.0cm]{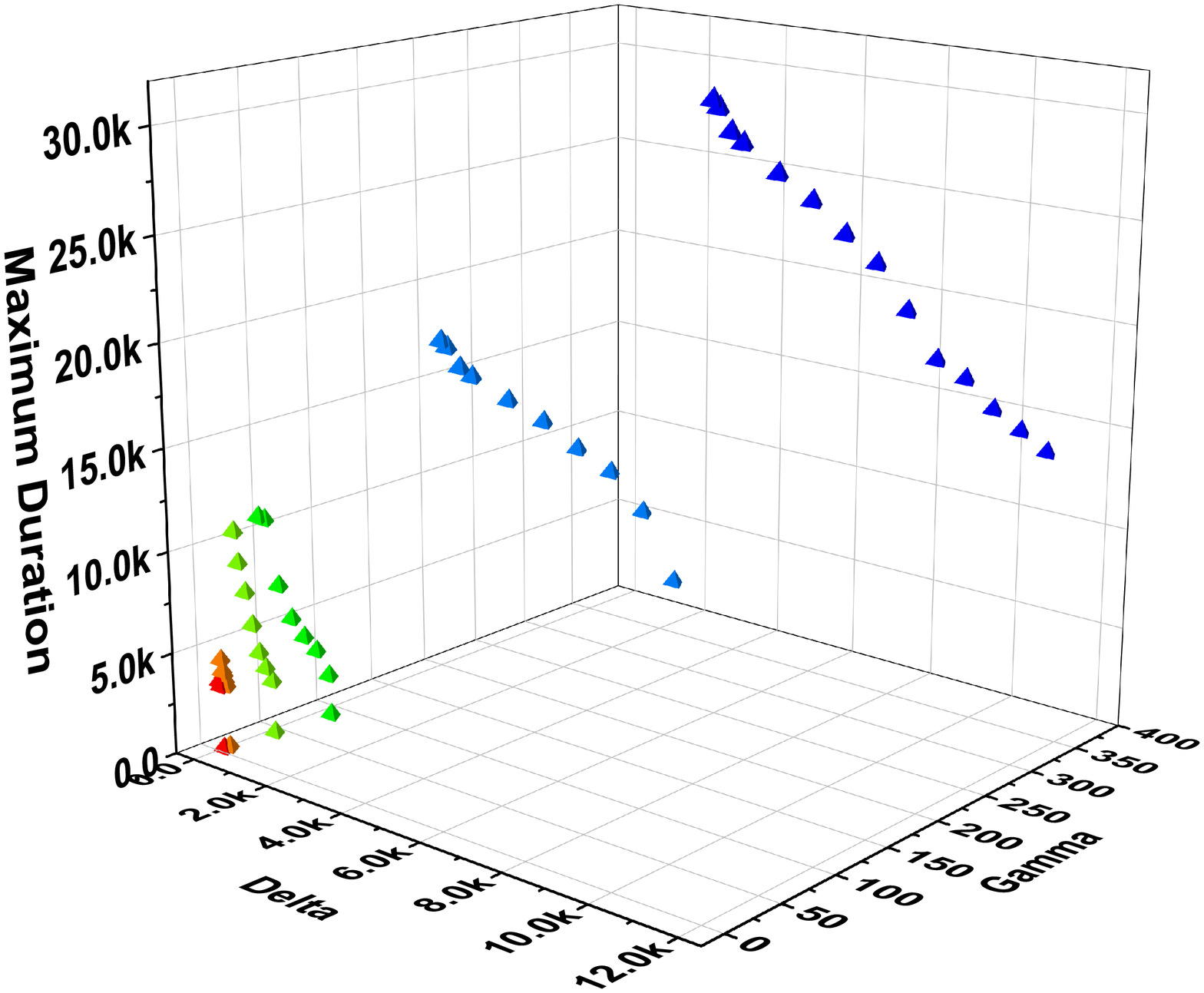} 
		\\
		& (c) Infectious I & \\
		
		\includegraphics[width=5cm,height=4.0cm]{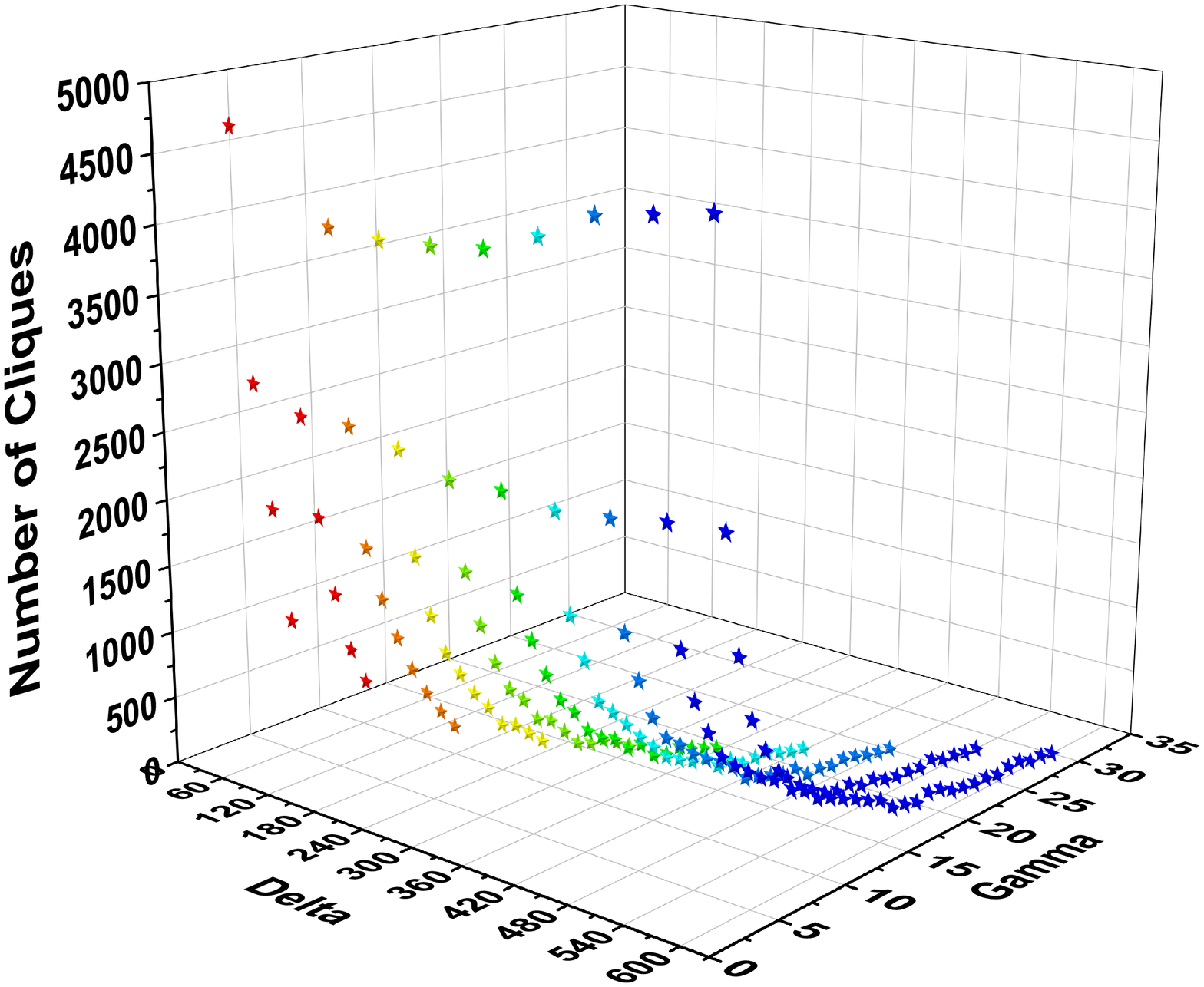}
		&
		\includegraphics[width=5cm,height=4.0cm]{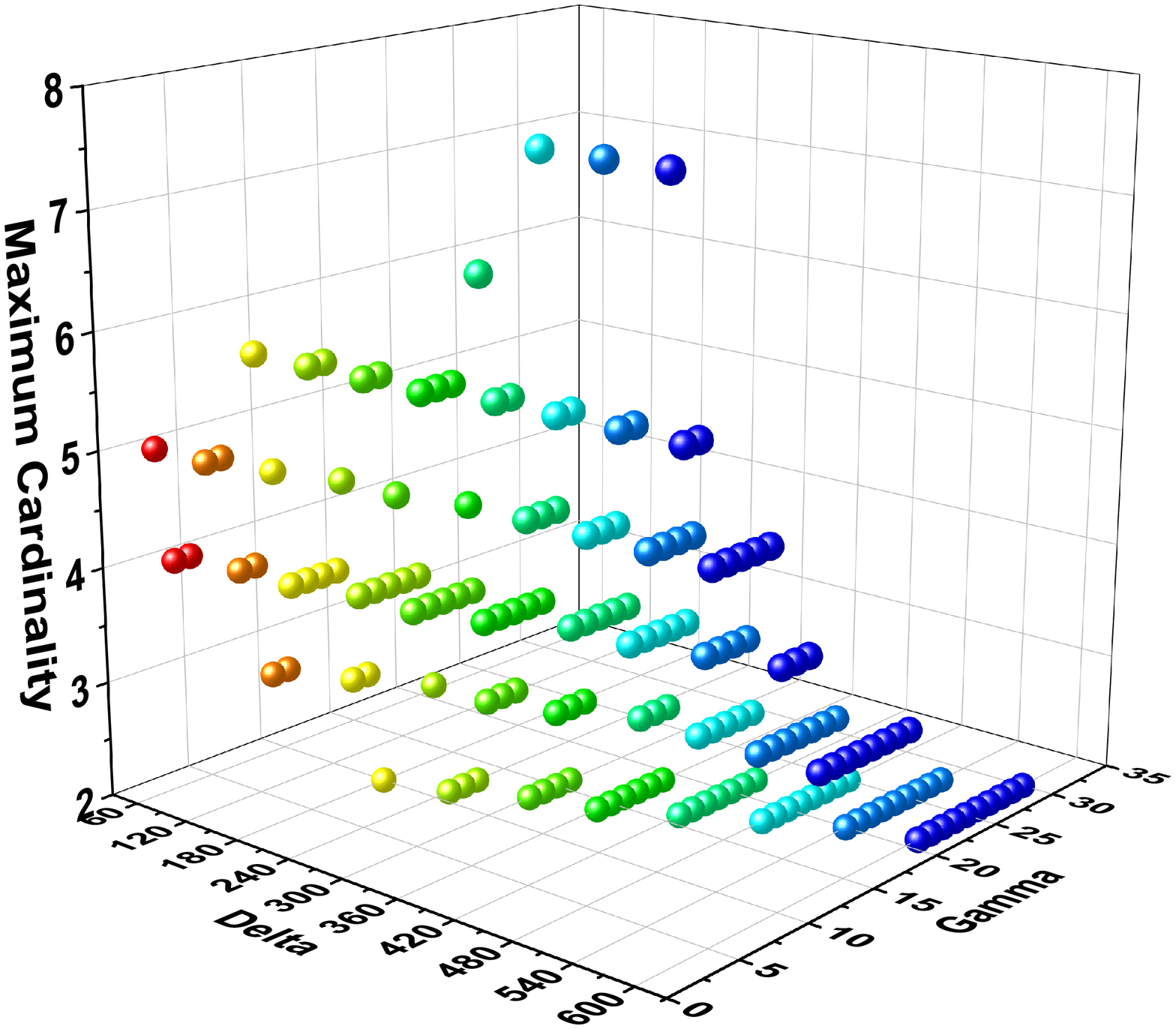}
		&
		\includegraphics[width=5cm,height=4.0cm]{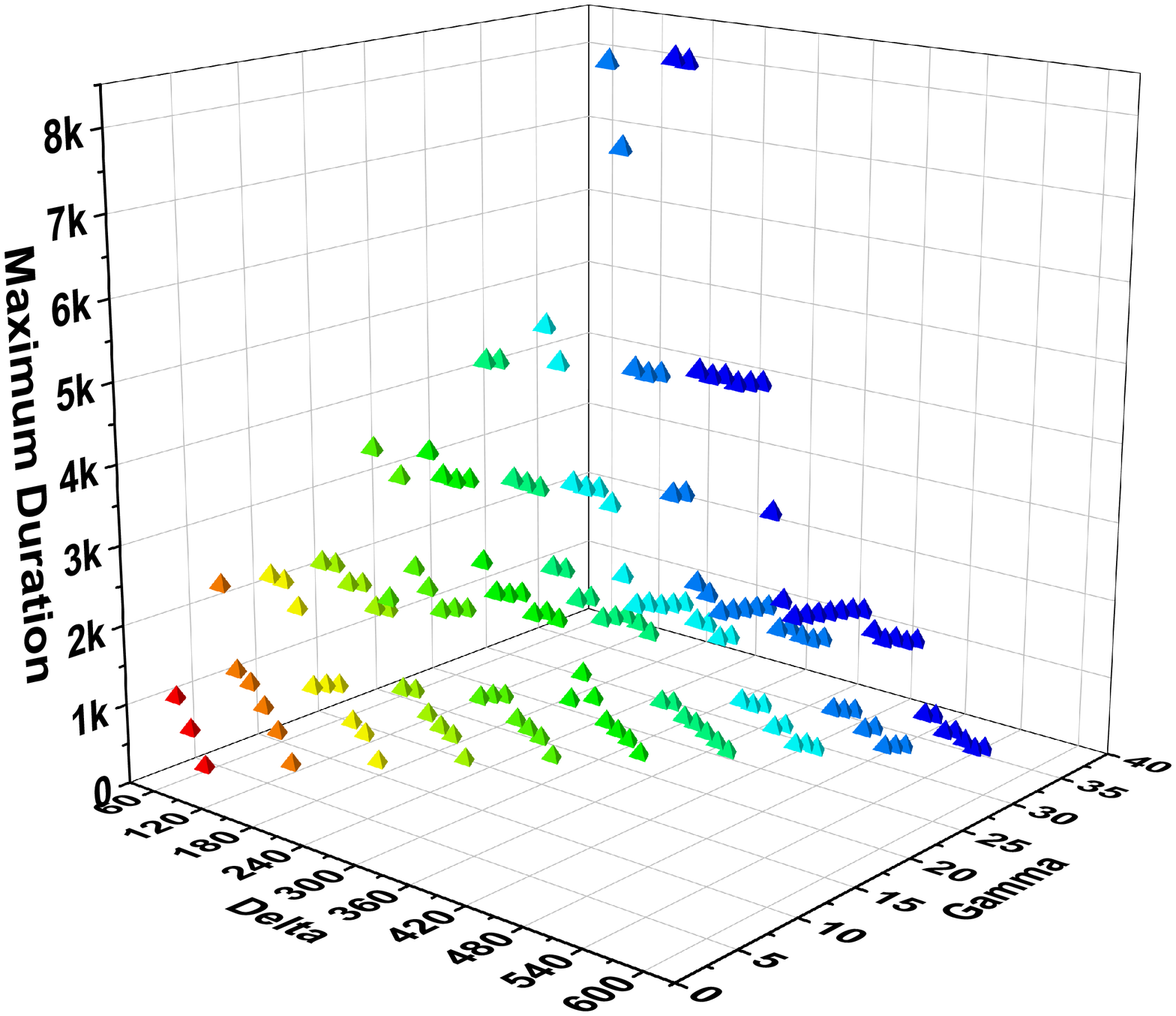}
		\\
		& (d) Infectious II  &\\
		
	\end{tabular}
	\caption{Plots for the change in Clique Count, Maximum Caridinality, and Maximum Duration with the change of $\Delta$ and $\gamma$ for different datasets}
	\label{Fig:results}
\end{figure*}

\begin{figure*}[!htbp]
	\centering
	\begin{tabular}{cc}
		\textbf{\underline{Computational Time}} & \textbf{\underline{Space Requirements}}  \\
		\includegraphics[width=5cm,height=4.5cm]{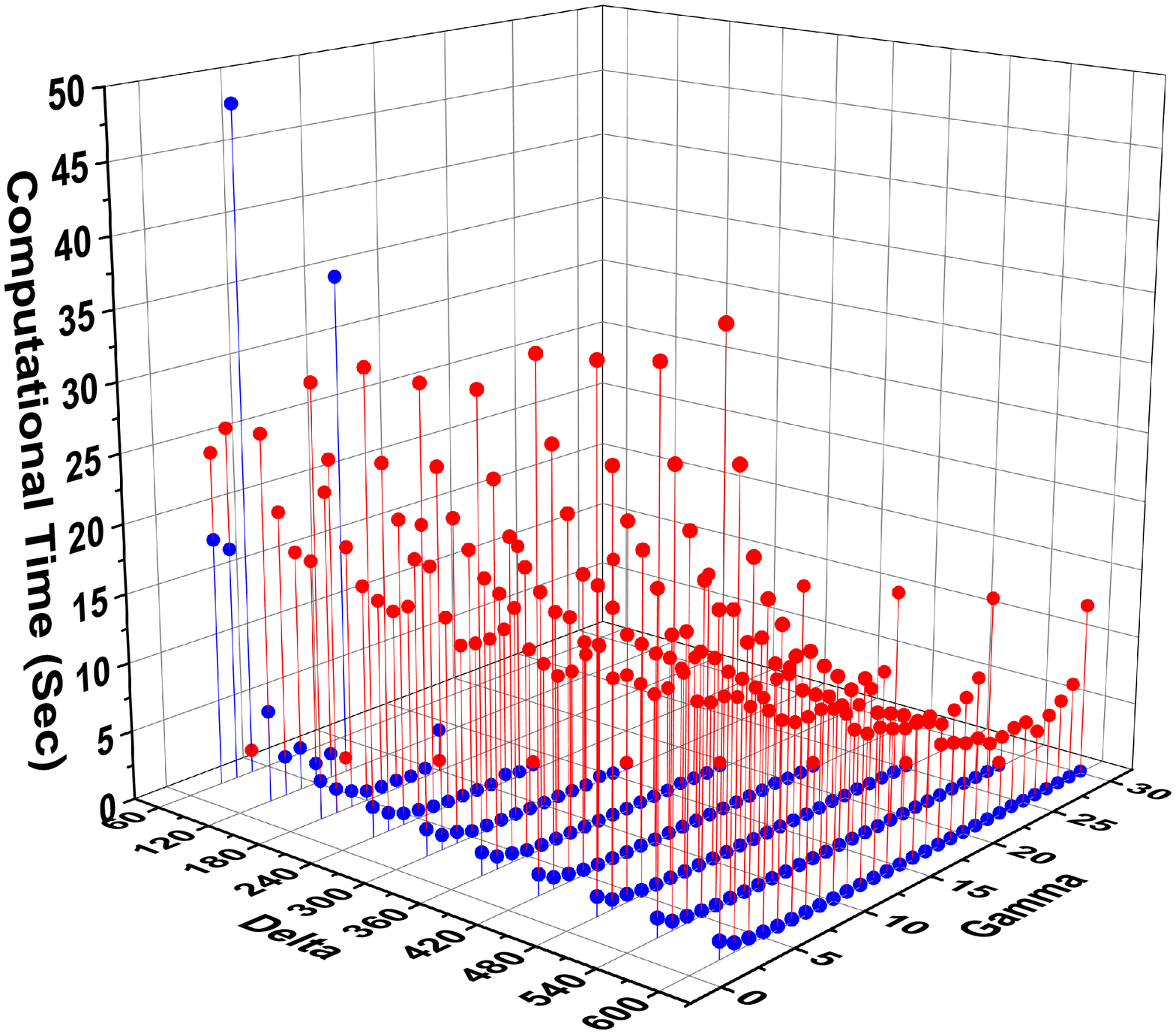} 
		&
		\includegraphics[width=5cm,height=4.5cm]{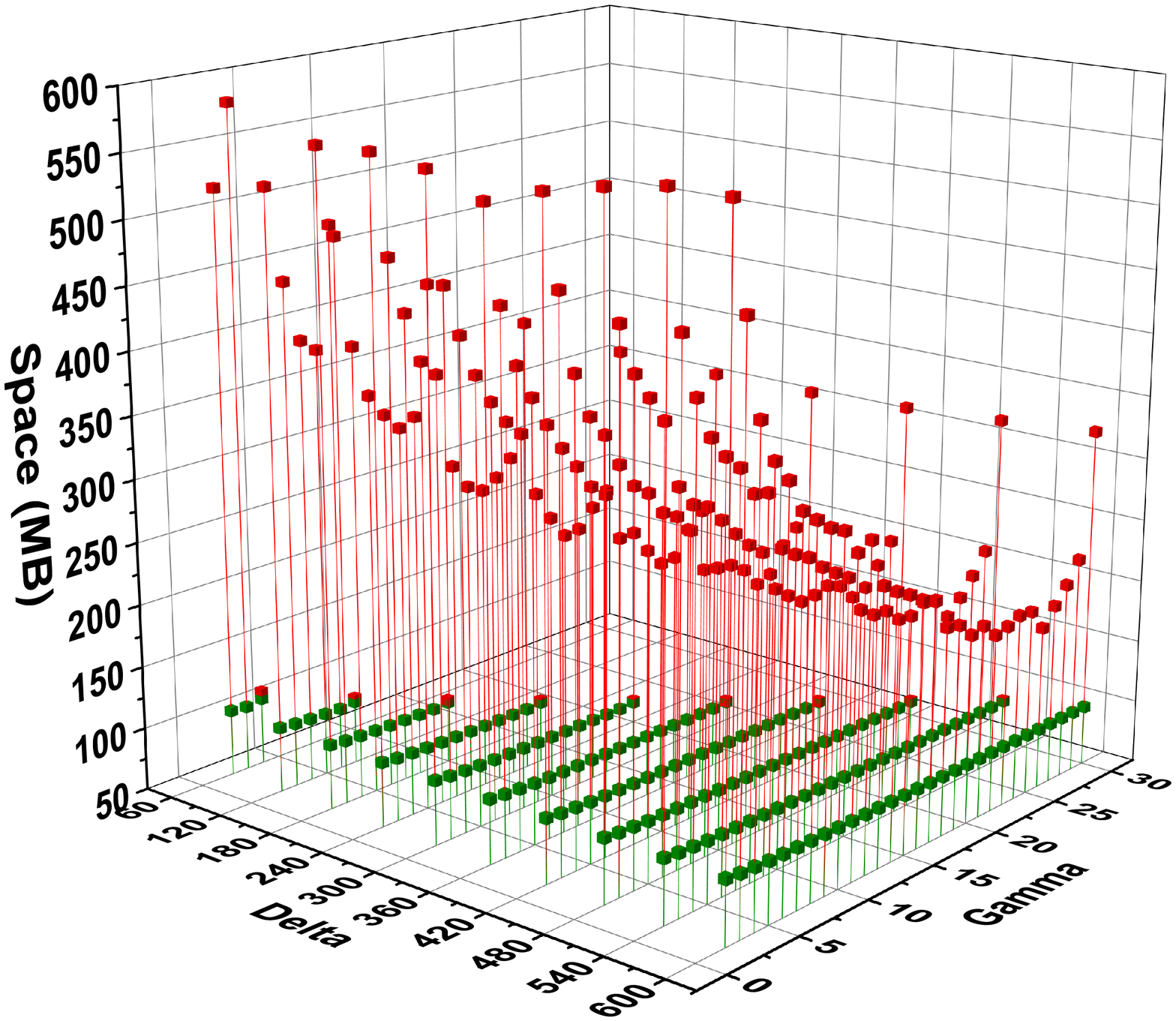}
		
		\\
		
		\hspace{2 cm}  (a) Hypertext \\
		%
		% \includegraphics[width=5cm,height=4.0cm]{haggle_Computational_Time} 
		% &
		% \includegraphics[width=5cm,height=4.0cm]{haggle_Space} 
		% 
		% \\
		%\hspace{2 cm} (b) Haggle Dataset  \\
		\includegraphics[width=5cm,height=4.0cm]{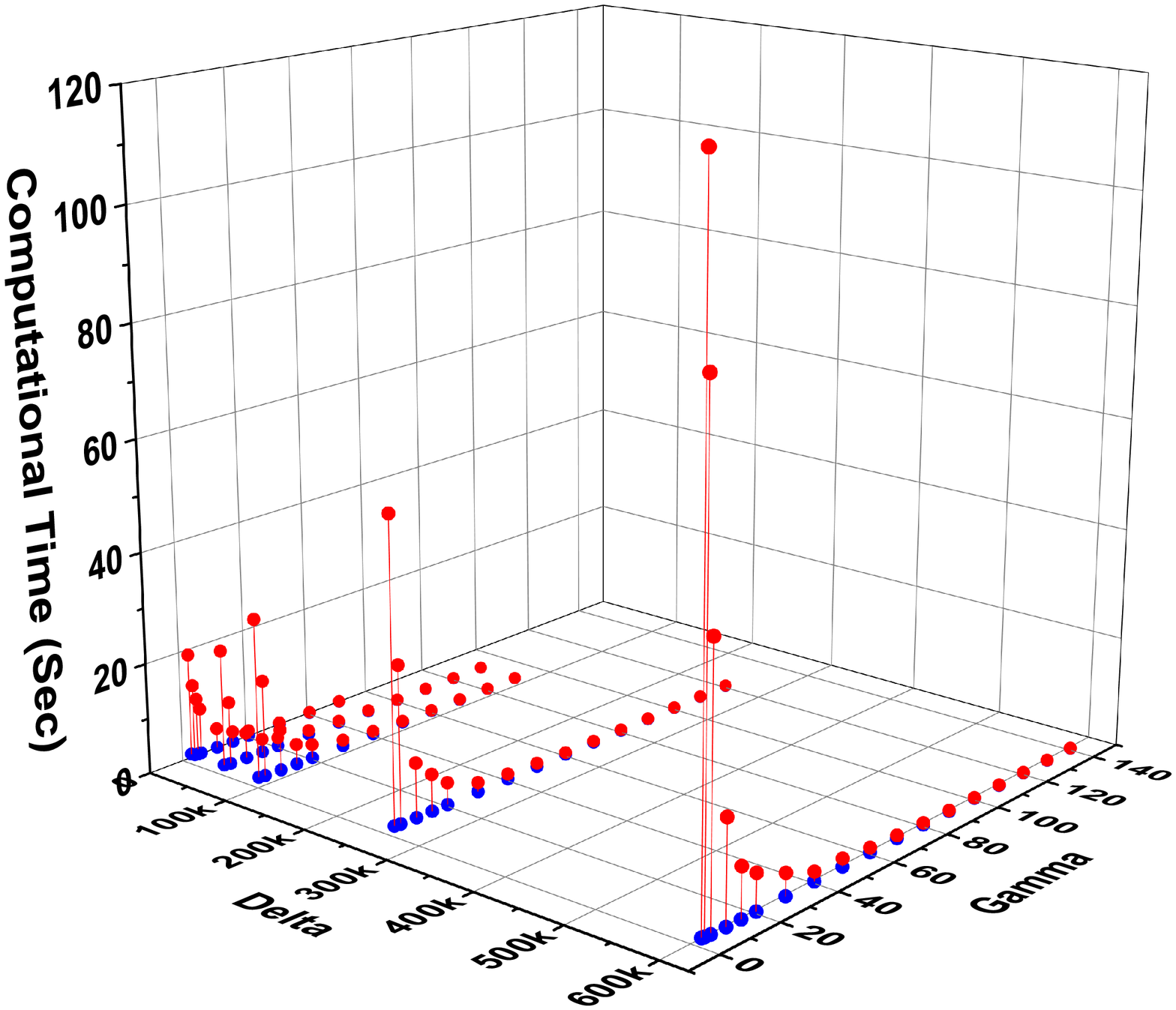} 
		&
		\includegraphics[width=5cm,height=4.0cm]{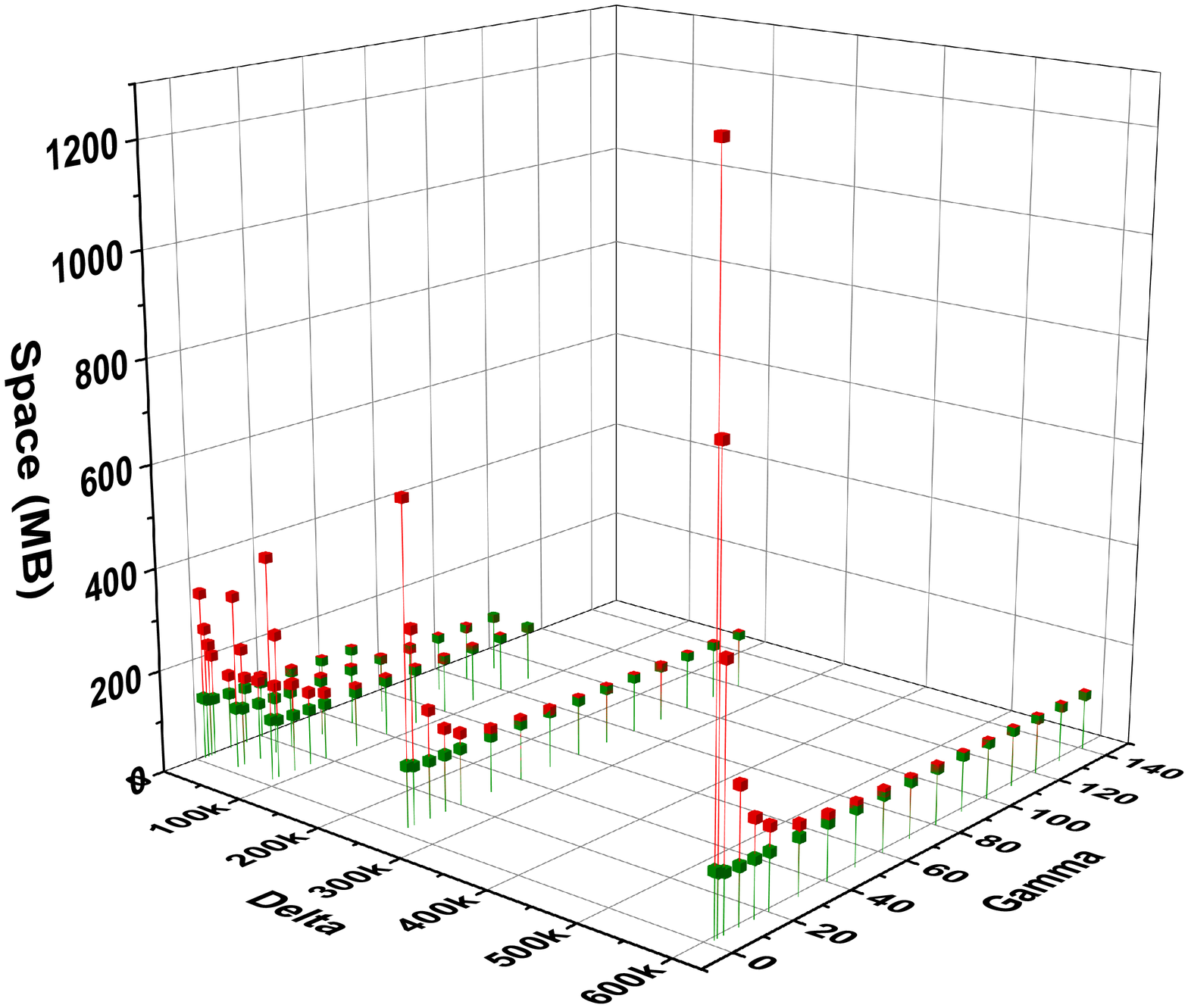} 
		
		\\
		\hspace{2 cm} (b) College Message  \\
		\includegraphics[width=5cm,height=4.0cm]{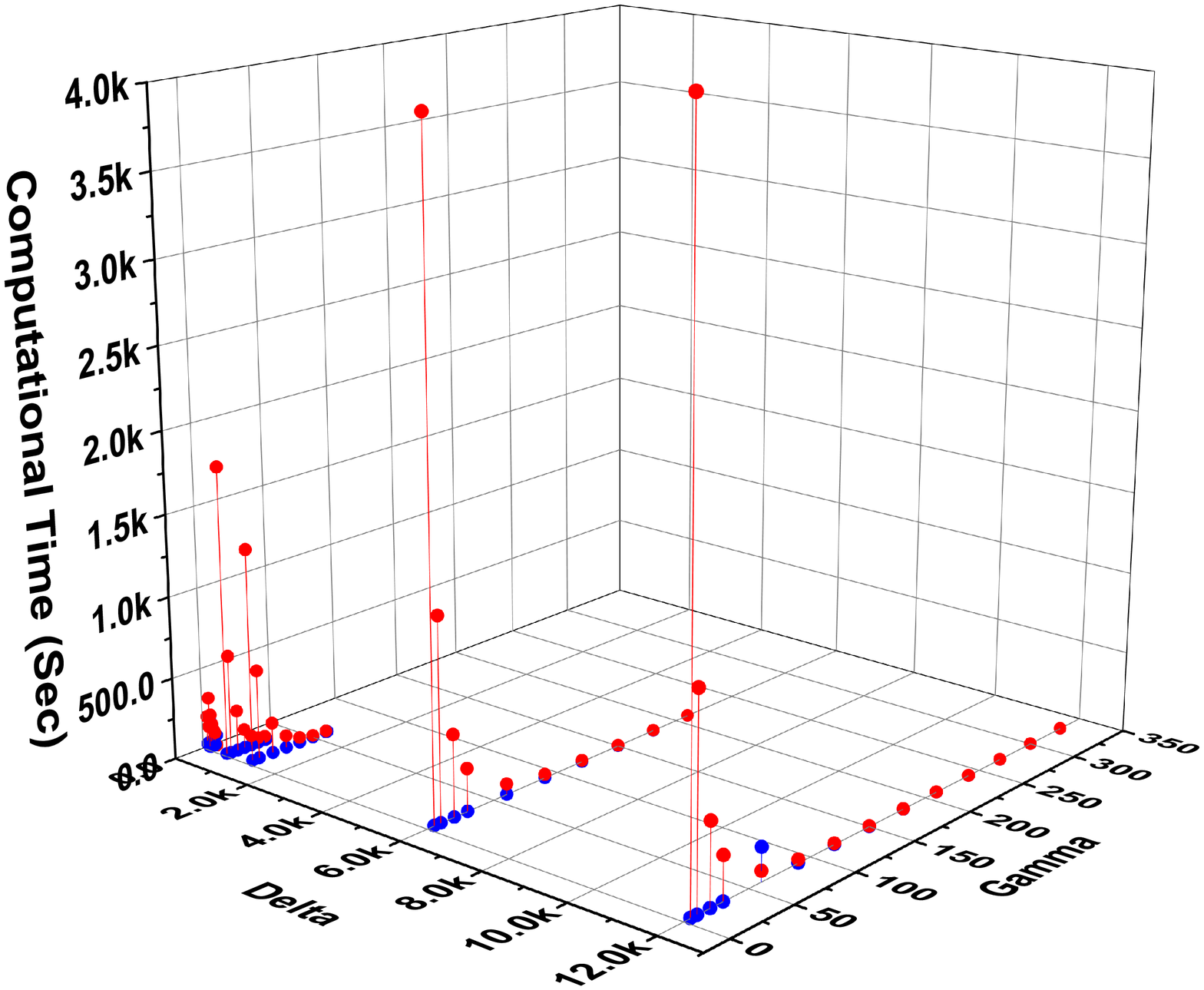} 
		&
		\includegraphics[width=5cm,height=4.0cm]{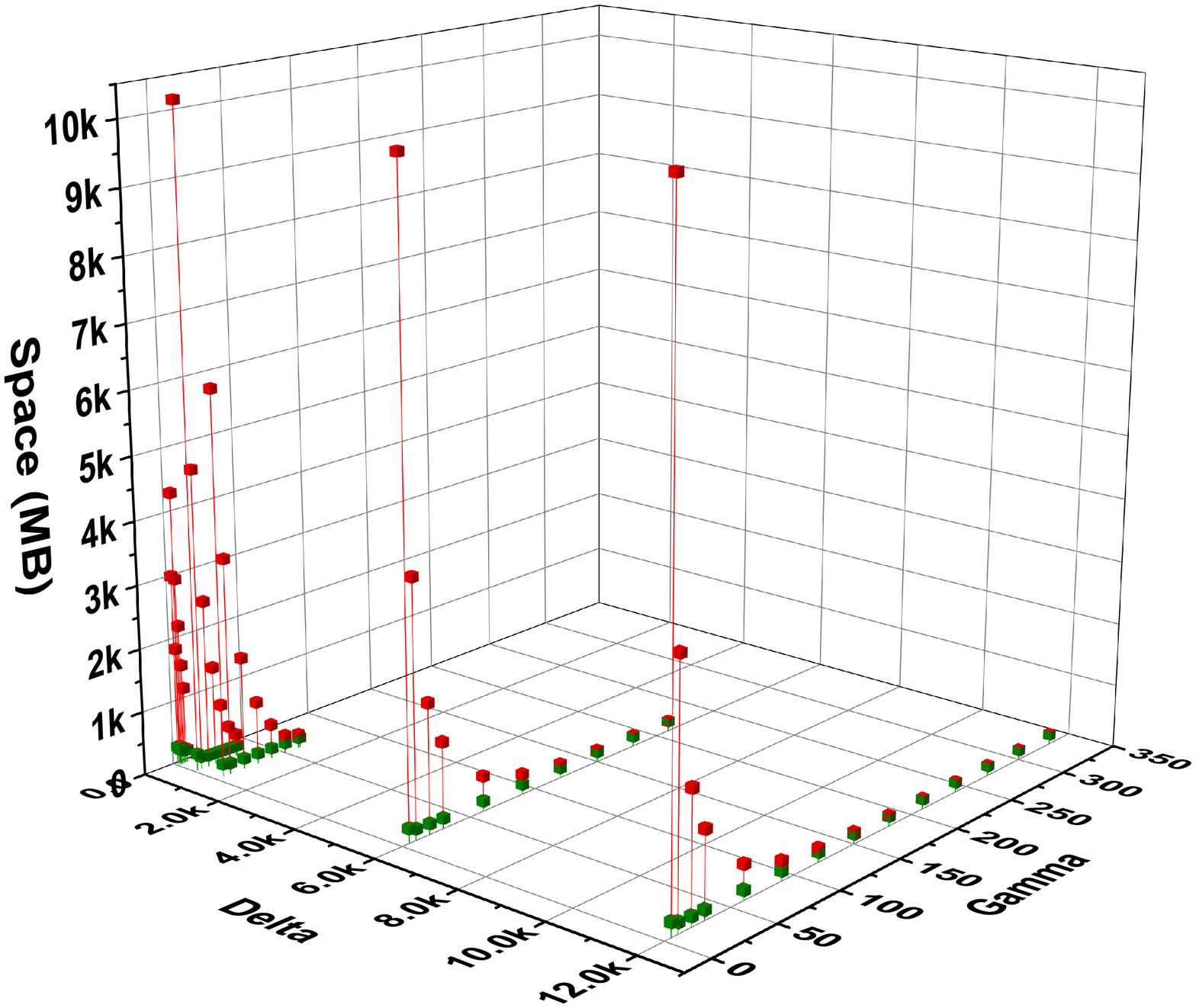} 
		
		\\
		\hspace{2 cm} (c) Infectious I  \\
		
		\includegraphics[width=5cm,height=4.0cm]{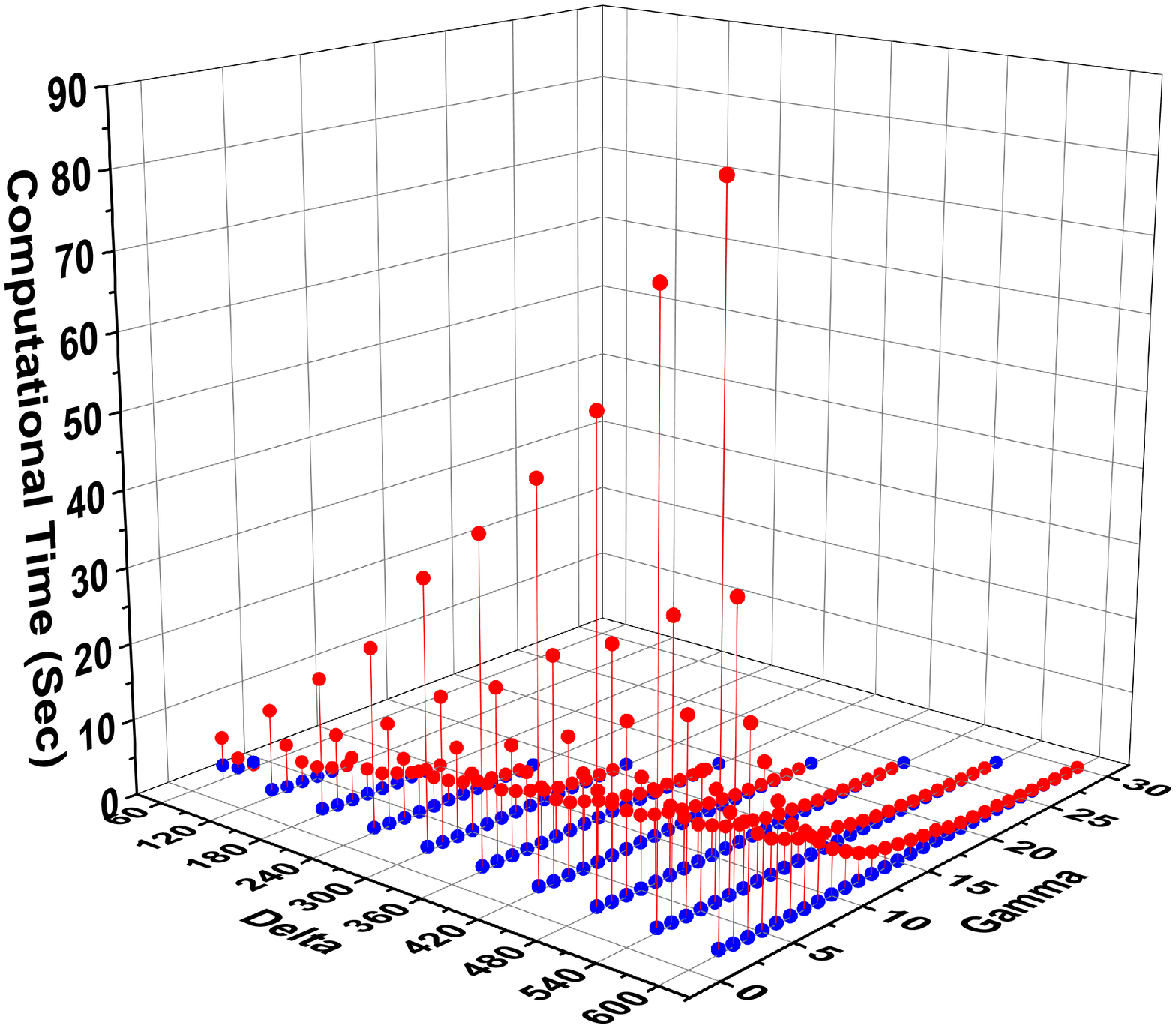}
		&
		\includegraphics[width=5cm,height=4.0cm]{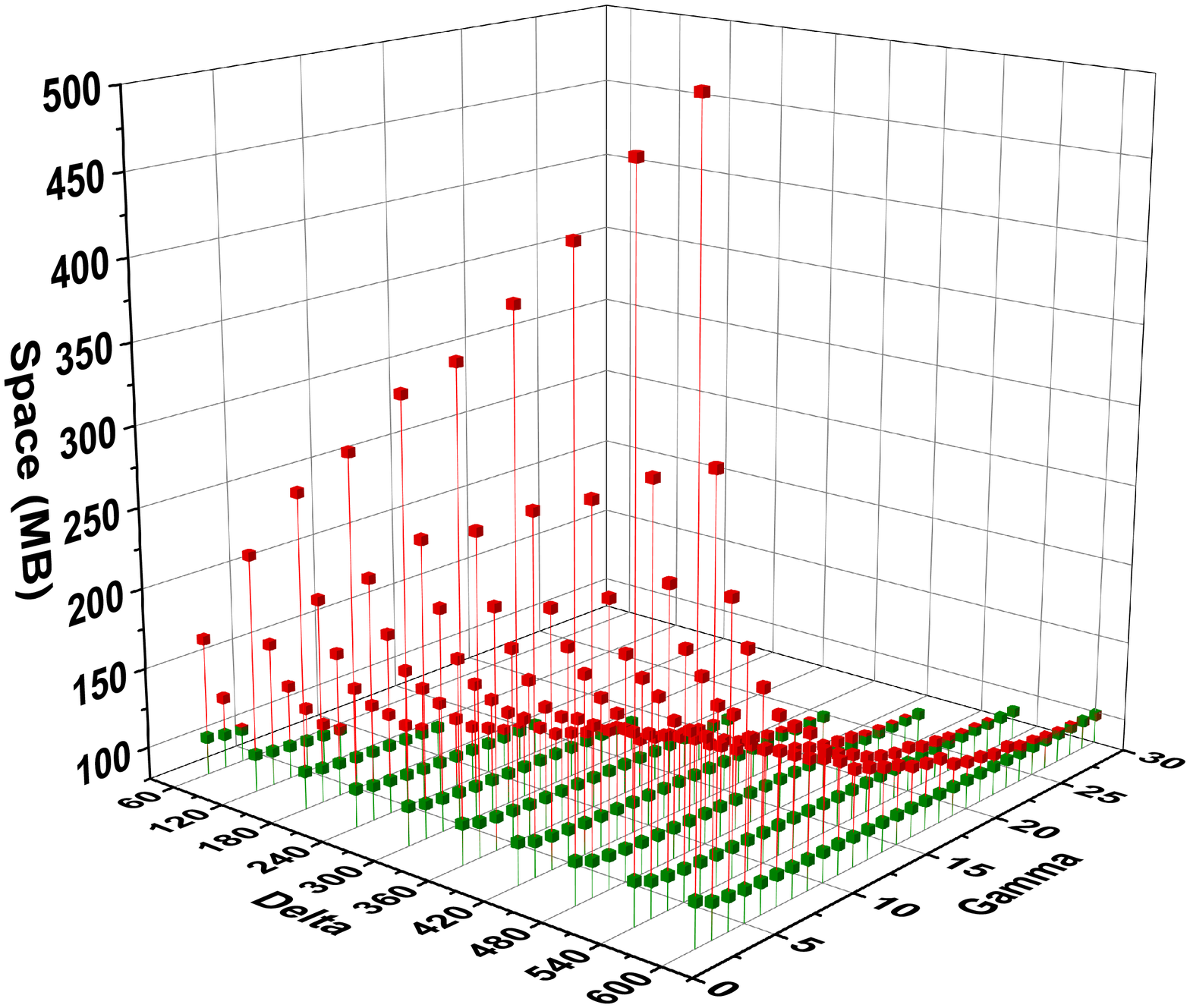}
		
		\\
		\hspace{2 cm} (d) Infectious II 
	\end{tabular}
	\caption{Plots for the change in Computational Time (in Secs), Space Requirement (in MB) with the change of $\Delta$ and $\gamma$ for different datasets}
	\label{Fig:results_TS}
\end{figure*}

\section{Conclusion and Future Directions} \label{Sec:CFD}
In this paper, we have proposed a methodology to enumerate all the maximal $(\Delta, \gamma)$\mbox{-}cliques present in a temporal network. The proposed methodology has been analyzed for time and space requirements, and also its correctness has been shown. To highlight its effectiveness, we have compared the execution time of the proposed methodology on five real\mbox{-}world publicly available datasets over the existing methods from the literature. Now, this study can be extended in the following directions. A different methodology can be processed for enumerating maximal $(\Delta,\gamma)$\mbox{-}cliques. In many real\mbox{-}world applications, links are probabilistic in nature. This phenomenon can be incorporated into our study. %Finally, different datasets arising in different situations can be analyzed by the proposed methodology.
%\begin{acknowledgements}
%\end{acknowledgements}


\bibliographystyle{spbasic}      % basic style, author-year citations
\bibliography{bare_conf}   % name your BibTeX data base

\begin{thebibliography}{37}
\providecommand{\natexlab}[1]{#1}
\providecommand{\url}[1]{{#1}}
\providecommand{\urlprefix}{URL }
\expandafter\ifx\csname urlstyle\endcsname\relax
  \providecommand{\doi}[1]{DOI~\discretionary{}{}{}#1}\else
  \providecommand{\doi}{DOI~\discretionary{}{}{}\begingroup
  \urlstyle{rm}\Url}\fi
\providecommand{\eprint}[2][]{\url{#2}}

\bibitem[{Akkoyunlu(1973)}]{akkoyunlu1973enumeration}
Akkoyunlu EA (1973) The enumeration of maximal cliques of large graphs. SIAM
  Journal on Computing 2(1):1--6

\bibitem[{Al-Naymat et~al.(2007)Al-Naymat, Chawla, and
  Arunasalam}]{al2007enumeration}
Al-Naymat G, Chawla S, Arunasalam B (2007) Enumeration of maximal clique for
  mining spatial co-location patterns

\bibitem[{Banerjee and Pal(2019)}]{DBLP:conf/comad/BanerjeeP19}
Banerjee S, Pal B (2019) On the enumeration of maximal ({\(\Delta\)},
  {\(\gamma\)})-cliques of a temporal network. In: Proceedings of the {ACM}
  India Joint International Conference on Data Science and Management of Data,
  {COMAD/CODS} 2019, Kolkata, India, January 3-5, 2019, pp 112--120,
  \doi{10.1145/3297001.3297015},
  \urlprefix\url{https://doi.org/10.1145/3297001.3297015}

\bibitem[{Bhowmick and Seah(2015)}]{bhowmick2015clustering}
Bhowmick SS, Seah BS (2015) Clustering and summarizing protein-protein
  interaction networks: A survey. IEEE Transactions on Knowledge and Data
  Engineering 28(3):638--658

\bibitem[{Bron and Kerbosch(1973{\natexlab{a}})}]{bron1973algorithm}
Bron C, Kerbosch J (1973{\natexlab{a}}) Algorithm 457: finding all cliques of
  an undirected graph. Communications of the ACM 16(9):575--577

\bibitem[{Bron and Kerbosch(1973{\natexlab{b}})}]{DBLP:journals/cacm/BronK73}
Bron C, Kerbosch J (1973{\natexlab{b}}) Finding all cliques of an undirected
  graph (algorithm 457). Commun {ACM} 16(9):575--576

\bibitem[{Chaintreau et~al.(2007)Chaintreau, Hui, Crowcroft, Diot, Gass, and
  Scott}]{chaintreau2007impact}
Chaintreau A, Hui P, Crowcroft J, Diot C, Gass R, Scott J (2007) Impact of
  human mobility on opportunistic forwarding algorithms. IEEE Transactions on
  Mobile Computing 6(6):606--620

\bibitem[{Chen et~al.(2016)Chen, Fang, Wang, Suo, Li, and
  Ives}]{chen2016parallelizing}
Chen Q, Fang C, Wang Z, Suo B, Li Z, Ives ZG (2016) Parallelizing maximal
  clique enumeration over graph data. In: International Conference on Database
  Systems for Advanced Applications, Springer, pp 249--264

\bibitem[{Cheng et~al.(2010)Cheng, Ke, Fu, Yu, and Zhu}]{cheng2010finding}
Cheng J, Ke Y, Fu AWC, Yu JX, Zhu L (2010) Finding maximal cliques in massive
  networks by h*-graph. In: Proceedings of the 2010 ACM SIGMOD International
  Conference on Management of data, ACM, pp 447--458

\bibitem[{Cheng et~al.(2011)Cheng, Ke, Fu, Yu, and Zhu}]{cheng2011finding}
Cheng J, Ke Y, Fu AWC, Yu JX, Zhu L (2011) Finding maximal cliques in massive
  networks. ACM Transactions on Database Systems (TODS) 36(4):21

\bibitem[{Cheng et~al.(2012)Cheng, Zhu, Ke, and Chu}]{cheng2012fast}
Cheng J, Zhu L, Ke Y, Chu S (2012) Fast algorithms for maximal clique
  enumeration with limited memory. In: Proceedings of the 18th ACM SIGKDD
  international conference on Knowledge discovery and data mining, ACM, pp
  1240--1248

\bibitem[{Du et~al.(2006)Du, Wu, Xu, Wang, and Pei}]{du2006parallel}
Du N, Wu B, Xu L, Wang B, Pei X (2006) A parallel algorithm for enumerating all
  maximal cliques in complex network. In: Sixth IEEE International Conference
  on Data Mining-Workshops (ICDMW'06), IEEE, pp 320--324

\bibitem[{Eppstein and Strash(2011)}]{eppstein2011listing}
Eppstein D, Strash D (2011) Listing all maximal cliques in large sparse
  real-world graphs. Experimental Algorithms pp 364--375

\bibitem[{Eppstein et~al.(2013{\natexlab{a}})Eppstein, L{\"o}ffler, and
  Strash}]{eppstein2013listing1}
Eppstein D, L{\"o}ffler M, Strash D (2013{\natexlab{a}}) Listing all maximal
  cliques in large sparse real-world graphs. Journal of Experimental
  Algorithmics (JEA) 18:3--1

\bibitem[{Eppstein et~al.(2013{\natexlab{b}})Eppstein, L{\"o}ffler, and
  Strash}]{eppstein2013listing}
Eppstein D, L{\"o}ffler M, Strash D (2013{\natexlab{b}}) Listing all maximal
  cliques in large sparse real-world graphs. Journal of Experimental
  Algorithmics (JEA) 18:3--1

\bibitem[{Garey and Johnson(2002)}]{garey2002computers}
Garey MR, Johnson DS (2002) Computers and intractability, vol~29. wh freeman
  New York

\bibitem[{Himmel et~al.(2016)Himmel, Molter, Niedermeier, and
  Sorge}]{himmel2016enumerating}
Himmel AS, Molter H, Niedermeier R, Sorge M (2016) Enumerating maximal cliques
  in temporal graphs. In: Advances in Social Networks Analysis and Mining
  (ASONAM), 2016 IEEE/ACM International Conference on, IEEE, pp 337--344

\bibitem[{Himmel et~al.(2017)Himmel, Molter, Niedermeier, and
  Sorge}]{himmel2017adapting}
Himmel AS, Molter H, Niedermeier R, Sorge M (2017) Adapting the bron--kerbosch
  algorithm for enumerating maximal cliques in temporal graphs. Social Network
  Analysis and Mining 7(1):35

\bibitem[{Holme and Saram{\"a}ki(2012)}]{holme2012temporal}
Holme P, Saram{\"a}ki J (2012) Temporal networks. Physics reports
  519(3):97--125

\bibitem[{Holme and Saram{\"a}ki(2013)}]{holme2013temporal}
Holme P, Saram{\"a}ki J (2013) Temporal networks. Springer

\bibitem[{Hou et~al.(2016)Hou, Wang, Chen, Suo, Fang, Li, and
  Ives}]{hou2016efficient}
Hou B, Wang Z, Chen Q, Suo B, Fang C, Li Z, Ives ZG (2016) Efficient maximal
  clique enumeration over graph data. Data Science and Engineering
  1(4):219--230

\bibitem[{Hulovatyy et~al.(2015)Hulovatyy, Chen, and
  Milenkovi{\'c}}]{hulovatyy2015exploring}
Hulovatyy Y, Chen H, Milenkovi{\'c} T (2015) Exploring the structure and
  function of temporal networks with dynamic graphlets. Bioinformatics
  31(12):i171--i180

\bibitem[{Isella et~al.(2011)Isella, Stehl{\'e}, Barrat, Cattuto, Pinton, and
  Van~den Broeck}]{isella2011s}
Isella L, Stehl{\'e} J, Barrat A, Cattuto C, Pinton JF, Van~den Broeck W (2011)
  What's in a crowd? analysis of face-to-face behavioral networks. Journal of
  theoretical biology 271(1):166--180

\bibitem[{Kumar et~al.(2016)Kumar, Spezzano, Subrahmanian, and
  Faloutsos}]{kumar2016edge}
Kumar S, Spezzano F, Subrahmanian V, Faloutsos C (2016) Edge weight prediction
  in weighted signed networks. In: Data Mining (ICDM), 2016 IEEE 16th
  International Conference on, IEEE, pp 221--230

\bibitem[{Kumar et~al.(2018)Kumar, Hooi, Makhija, Kumar, Faloutsos, and
  Subrahmanian}]{kumar2018rev2}
Kumar S, Hooi B, Makhija D, Kumar M, Faloutsos C, Subrahmanian V (2018) Rev2:
  Fraudulent user prediction in rating platforms. In: Proceedings of the
  Eleventh ACM International Conference on Web Search and Data Mining, ACM, pp
  333--341

\bibitem[{Masuda and Holme(2017)}]{masuda2017temporal}
Masuda N, Holme P (2017) Temporal Network Epidemiology. Springer

\bibitem[{Molter et~al.(2019)Molter, Niedermeier, and
  Renken}]{molter2019enumerating}
Molter H, Niedermeier R, Renken M (2019) Enumerating isolated cliques in
  temporal networks. In: International Conference on Complex Networks and Their
  Applications, Springer, pp 519--531

\bibitem[{Mukherjee et~al.(2015)Mukherjee, Xu, and
  Tirthapura}]{mukherjee2015mining}
Mukherjee AP, Xu P, Tirthapura S (2015) Mining maximal cliques from an
  uncertain graph. In: Data Engineering (ICDE), 2015 IEEE 31st International
  Conference on, IEEE, pp 243--254

\bibitem[{Mukherjee et~al.(2016)Mukherjee, Xu, and
  Tirthapura}]{mukherjee2016enumeration}
Mukherjee AP, Xu P, Tirthapura S (2016) Enumeration of maximal cliques from an
  uncertain graph. IEEE Transactions on Knowledge and Data Engineering
  29(3):543--555

\bibitem[{Panzarasa et~al.(2009)Panzarasa, Opsahl, and
  Carley}]{panzarasa2009patterns}
Panzarasa P, Opsahl T, Carley KM (2009) Patterns and dynamics of users'
  behavior and interaction: Network analysis of an online community. Journal of
  the Association for Information Science and Technology 60(5):911--932

\bibitem[{Rossi et~al.(2014)Rossi, Gleich, Gebremedhin, and
  Patwary}]{rossi2014fast}
Rossi RA, Gleich DF, Gebremedhin AH, Patwary MMA (2014) Fast maximum clique
  algorithms for large graphs. In: Proceedings of the 23rd International
  Conference on World Wide Web, ACM, pp 365--366

\bibitem[{Rossi et~al.(2015)Rossi, Gleich, and Gebremedhin}]{rossi2015parallel}
Rossi RA, Gleich DF, Gebremedhin AH (2015) Parallel maximum clique algorithms
  with applications to network analysis. SIAM Journal on Scientific Computing
  37(5):C589--C616

\bibitem[{Schmidt et~al.(2009)Schmidt, Samatova, Thomas, and
  Park}]{schmidt2009scalable}
Schmidt MC, Samatova NF, Thomas K, Park BH (2009) A scalable, parallel
  algorithm for maximal clique enumeration. Journal of Parallel and Distributed
  Computing 69(4):417--428

\bibitem[{Viard et~al.(2015)Viard, Latapy, and Magnien}]{viard2015revealing}
Viard J, Latapy M, Magnien C (2015) Revealing contact patterns among
  high-school students using maximal cliques in link streams. In: Proceedings
  of the 2015 IEEE/ACM International Conference on Advances in Social Networks
  Analysis and Mining 2015, ACM, pp 1517--1522

\bibitem[{Viard et~al.(2016)Viard, Latapy, and Magnien}]{viard2016computing}
Viard T, Latapy M, Magnien C (2016) Computing maximal cliques in link streams.
  Theoretical Computer Science 609:245--252

\bibitem[{Xiang et~al.(2013)Xiang, Guo, and Aboulnaga}]{xiang2013scalable}
Xiang J, Guo C, Aboulnaga A (2013) Scalable maximum clique computation using
  mapreduce. In: 2013 IEEE 29th International Conference on Data Engineering
  (ICDE), IEEE, pp 74--85

\bibitem[{Zou et~al.(2010)Zou, Li, Gao, and Zhang}]{zou2010finding}
Zou Z, Li J, Gao H, Zhang S (2010) Finding top-k maximal cliques in an
  uncertain graph. In: 2010 IEEE 26th International Conference on Data
  Engineering (ICDE 2010), IEEE, pp 649--652

\end{thebibliography}


\end{document}